\documentclass[10pt,draft,reqno]{amsart}
     \makeatletter
     \def\section{\@startsection{section}{1}%
     \z@{.7\linespacing\@plus\linespacing}{.5\linespacing}%
     {\bfseries
     \centering
     }}
     \def\@secnumfont{\bfseries}
     \makeatother
\newtheorem{theorem}{Theorem}[section]
\newtheorem{lemma}[theorem]{Lemma}
\newtheorem{proposition}[theorem]{Proposition}
\newtheorem{corollary}[theorem]{Corollary}
\theoremstyle{definition}
\newtheorem{definition}[theorem]{Definition}
\newtheorem{example}[theorem]{Example}
\theoremstyle{remark}
\newtheorem{notation}{Notation}
\newtheorem{remark}[theorem]{Remark}
\numberwithin{equation}{section}
\begin{document}

\title[A ``non--commutative'' $\rm{sl}(2,\mathbb{C})$]{The $n$-dimensional
quadratic Heisenberg algebra as a ``non--commutative''
$\rm{sl}(2,\mathbb{C})$}

\author{Luigi Accardi}
\address{Luigi Accardi: Centro Vito Volterra, Universit\`{a} di Roma Tor Vergata, via\break Columbia  2, 00133 Roma, Italy}
\email{accardi@volterra.mat.uniroma2.it}
\author[Andreas Boukas]{Andreas Boukas}
\address{Andreas Boukas: Centro Vito Volterra and Hellenic Open University, Graduate School of Mathematics, Patras, Greece}
\email{boukas.andreas@ac.eap.gr}
\author{Yun-Gang Lu}
\address{Yun-Gang Lu: Dipartimento di Matematica, Universit\`a di Bari,    n.4, Via E. Orabona,  70125 Bari, Italy}
\email{yungang.lu@uniba.it}




\begin{abstract}
We prove that the commutation relations among the generators of
the quadratic Heisenberg algebra of dimension $n\in\mathbb{N}$,
look like a kind of \textit{non-commutative extension} of
$\hbox{sl}(2, \mathbb{C})$ (more precisely of its unique
$1$--dimensional central extension), denoted
$\hbox{heis}_{2;\mathbb{C}}(n)$ and called the complex
$n$--dimensional quadratic Boson algebra. This
\textit{non-commutativity} has a different nature from the one
considered in quantum groups.
We prove the exponentiability of these algebras (for any $n$) in
the Fock representation. We obtain the group multiplication law,
in coordinates of the first and second kind, for the quadratic
Boson group and we show that, in the case of the adjoint
representation, these multiplication laws can be expressed in
terms of a generalization of the Jordan multiplication. We
investigate the connections between these two types of coordinates
(disentangling formulas). From this we deduce a new proof of the
expression of the vacuum characteristic function of homogeneous
quadratic boson fields.
\end{abstract}

\maketitle

\section{Introduction}\label{intro}

The \textbf{complex $n$--dimensional Heisenberg algebra}, denoted
$\hbox{heis}_{1,\mathbb{C}}(n)$, is the complex $*$--Lie algebra
with generators
\[
\mathbf{1} \hbox{ (central element) } \ , \ a^{\dagger}_j, a_k, j,
k\in\{1,\dots ,n\} \  ,
\]
commutation relations
\begin{equation}\label{CR-n-dim-Heis-alg}
\lbrack a_k, a^{\dagger}_j \rbrack=\delta_{j,k}\mathbf{1}\,\,;\,\,
\lbrack a_k, \mathbf{1} \rbrack=\lbrack a^{\dagger}_j, \mathbf{1}
\rbrack= \lbrack a_k, a_j \rbrack = \lbrack a^{\dagger}_k,
a^{\dagger}_j \rbrack =0 \  ,
\end{equation}
and involution given by
\begin{equation}\label{inv-Heis-n}
a_j^*=a^{\dagger}_j\qquad ;\qquad\left(a^{\dagger}_j\right)^*=a_j
\  .
\end{equation}

\smallskip

The \textbf{$n$--dimensional Heisenberg algebra}, denoted $\hbox{heis}_{1}(n)$, is the real
$*$--Lie sub--algebra of $\hbox{heis}_{1,\mathbb{C}}(n)$ consisting of its skew--adjoint elements.
The associated local  $*$--Lie groups are denoted respectively $\hbox{Heis}_{1,\mathbb{C}}(n)$ and $\hbox{Heis}_{1}(n)$. The universal enveloping algebra of the $n$--dimensional Heisenberg
$*$--algebra, also called Boson or CCR algebra, contains as a $*$--Lie sub--algebra the
complex linear span of the \textbf{normally ordered products of pairs} of creation or annihilation
operators that we denote $\hbox{heis}_{2;\mathbb{C}}(n)$ (see Section \ref{sec-heis(2,C,n)}).

\smallskip

This $*$--Lie algebra can be considered as a representation of an abstract
finite--dimensional complex $*$--Lie algebra that we call the
\textbf{complex $n$--dimensional quadratic Boson algebra}
(Definition \ref{n-dim-quadr-Bos-alg}) and, in this paper we
denote it with the same symbol $\hbox{heis}_{2;\mathbb{C}}(n)$ used
for its Boson representation (which is faithful by construction).
It is known that $\hbox{heis}_{2;\mathbb{C}}(1)$ is isomorphic to a
$1$--dimensional (necessarily trivial) central extension of
$\hbox{sl}(2,\mathbb{C})$. This fact was exploited in
\cite{[AcOuRe10-QuadrHeisGroup]}, \cite{[rebei16]} to construct a
projective representation of $\hbox{heis}_{2;\mathbb{C}}(1)$ (the
quadratic analogue of the Weyl representation) and to identify the
composition law of the associated Lie group (the quadratic
analogue of the Heisenberg group).

\smallskip

The structure of this composition law was simplified in
\cite{[AcDhaRe18]} and recently a substantial step forward in this
direction has been achieved in \cite{[RebRguiAlHus20]} with the
identification of the one--mode quadratic Heisenberg group with
the projective group $PSU(1,1)$ and an explicit realization of its
holomorphic representation. The multi--dimensional extension of
these results is essential to extend the presently available
results concerning the \textbf{vacuum distributions of the
Virasoro fields} which, in their truncated form (see
\cite{[AccardiBoukasLu16OSID]}), are elements of the
(non--homogeneous) quadratic algebra. The attempt to solve this
problem has been the starting point of the present paper. As it
happens for the theory of orthogonal polynomials, difficulties
undergo a veritable phase transition in the passage from dimension
$1$ to dimensions $\ge 2$. These difficulties disappear in the
product (or factorizable) case where calculations can be reduced
to the $1$--dimensional case up to a linear transformation in the
$1$--particle space. However, as shown in \cite{AcBou18a}, there
exist quadratic fields for which such a reduction is impossible.
Furthermore there are many indications that the truncated Virasoro
fields generically (i.e. for most choices of the parameters
defining them) belong to this class, even if a proof of this fact
is at the moment not available.

\smallskip

With these motivations we have carried out in the past years a systematic investigation of
different aspects of quadratic boson fields \cite{[AcBou-COSA10]},
\cite{[AcBou-PMS15]}, \cite{AcBou18a},
\cite{[AccardiBoukasLu16OSID]}.

\smallskip

In particular, in Lemma 6 of \cite{AcBou18a} it was remarked that, in dimensions $\ge 2$,
the commutation relations among generators of the quadratic Heisenberg algebra (see Definition
\ref{n-dim-quadr-Bos-alg}) look like a \textbf{a kind of non--commutative} extension of the
relations defining $\hbox{sl}(2,\mathbb{C})$ (more precisely of its unique $1$--dimensional
central extension), even if this terminology might seem weird since $\hbox{sl}(2,\mathbb{C})$
is itself non--commutative.

\smallskip

When $n=1$ this analogy becomes strict in the sense that, as already mentioned,
\begin{equation}\label{Heis(2,1)-equiv-SL(2,R)}
\hbox{sl}(2,\mathbb{R}) \equiv \hbox{heis}_{2}(1)\  ,
\end{equation}
where $\equiv$ means \textbf{$*$--isomorphism of $*$--Lie
algebras}.

\smallskip

 The complex $n$--dimensional quadratic Boson algebra was implicitly introduced,
with different notations and for different purposes, in Section 4
of the deep paper \cite{FP}, where a matrix representation for it
was constructed. However the \textit{non--commutative}
$\hbox{sl}(2,\mathbb{C})$ was not considered in that paper, where
the analogy with $\hbox{sl}(N,\mathbb{R})$ (for some Natural
integer $N$) was emphasized. We prove (see Lemma
\ref{lm-Dioph-eq}) that this analogy in general does not hold. In
fact, recent results obtained after the completion of this paper
show that it holds only in the case of $\hbox{sl}(2,\mathbb{C})$,
i.e. in the case of $1$ degree of freedom.

\smallskip

The additional non--commutativity, arising in $\hbox{heis}_{2}(n)$
(and more generally in $\hbox{heis}_{2,\mathbb{C}}(n)$) with
respect to $\hbox{sl}(2,\mathbb{R})$, makes formulas more implicit
due to the fact that, even in the finite dimensional case, there
is no explicit form for the exponential (\ref{df-W(x,A,B,C)})
except for very special cases. This exponential plays a crucial
role in the Feinsilver--Pap splitting Lemma that is one of the
main tools used in this paper. In fact, in the first part of this
paper (sections \ref{grp-law-coord-2dkd},
\ref{grp-law-coord-1st-kd}), we use this lemma to
\textbf{determine the composition law} of the quadratic Heisenberg
group, both in first and second type Lie--coordinates.  These
results are used in section \ref{Weyl-ops-Vac-Char-Fctn} to give a
\textbf{new deduction of the vacuum characteristic function} of
hermitian homogeneous quadratic fields (vacuum expectation of
quadratic Weyl operators). In the same section we find an
inductive relation for the scalar product of a special class of
$n$--particle vectors. Notice that, even in the $1$--dimensional
case, it took several years to find the explicit form of the
scalar product of two homogeneous quadratic $n$--particle vectors
(see Lemma 2.2 in \cite{BarhOuerRiah08}). In Proposition
\ref{proj2}, \textbf{we give an inductive formula for this scalar
product}.

\smallskip

In section \ref{adj-act-quad-grp} we prove that the adjoint action of the quadratic Lie group on
the quadratic $*$--Lie algebra has a simple and explicit expression.
This result, combined with the Feinsilver--Pap splitting Lemma, is used to put products of
quadratic Weyl operators in \textit{semi--normal form}. This differs from usual normal form
because between an exponential of creators and an exponential of annihilators one finds, instead of
a single exponential of a number type operator, a product of such exponentials.
Since the vacuum vector is left invariant by exponentials of number type operators, this gives
a tool to compute scalar products of quadratic coherent vectors much more
explicitly than with the Baker--Campbell--Hausdorff formula.

\smallskip

Finally, in section \ref{Exp-quadr-alg-in-Fck-repr} (Theorem
\ref{2eqr-06}), we extend to the multi--dimensional case the
estimates, proved in \cite{[AcOuRe10-QuadrHeisGroup]} for the
$1$--dimensional case. These allow to prove that, in the Fock
representation, the number vectors are analytic vectors for the
elements of the quadratic algebra. Thus, by Nelson's theorem, the
hermitian elements of this algebra are essentially self--adjoint
and their exponential series converges strongly on the domain of
number vectors. From this the existence and unitarity of the
quadratic Weyl operators follows.

\section{The Complex $n$--dimensional Quadratic Boson Algebra}
\label{sec-heis(2,C,n)}

In this section, we identify the complex $n$--dimensional
quadratic Boson algebra with a non--commutative version of
$\hbox{sl}(2,\mathbb{C})$.

Denote $M_{n}(\mathbb{C})$ (resp. $M_{n,sym}(\mathbb{C})$) the
algebra of $n\times n$ complex matrices (resp. symmetric matrices)
and, for $M\equiv (M_{j,k})\in M_{n}(\mathbb{C})$, define the
transpose, conjugate and adjoint of $M$ in the standard way:
$$
(M^T)_{j,k}:= M_{k,j}\,\,;\,\, (\overline{M})_{j,k}:=
\overline{M_{j,k}}\,\,;\,\, (M^*)_{j,k}:= (\overline{M})^T_{j,k}=
\overline{(M)^T}_{j,k} = \overline{M_{k,j}} \  .
$$
We identify $\hbox{heis}_{1,\mathbb{C}}(n)$ to a sub--algebra of its universal enveloping $*$--algebra
and consider the space of all \textbf{homogeneous quadratic expressions} in the generators
\eqref{CR-n-dim-Heis-alg}:

\begin{align*}
a^{\dagger}A a^{\dagger} :=&\sum_{j,k=1}^n
A_{j,k}a^{\dagger}_ja^{\dagger}_k \  ,\\
a^{\dagger} B
a:=&\sum_{j,k=1}^nB_{j,k}a^{\dagger}_ja_k\  ,\\
aCa:=&\sum_{j,k=1}^nC_{j,k}a_ja_k \  .
\end{align*}
Note that, since creators (resp. annihilators) mutually commute, one has
\begin{equation}\label{emerg-symm-matr}
a^{\dagger}A a^{\dagger}
=\sum_{j,k=1}^n\frac{1}{2}\left(A_{j,k}+A_{k,j}\right)
a^{\dagger}_ja^{\dagger}_k \,\,;\,\, aC a
=\sum_{j,k=1}^n\frac{1}{2}\left(C_{j,k}+C_{k,j}\right) a_ja_k \  ,
\end{equation}
i.e. the expressions $a^{\dagger}A a^{\dagger}$ and $aCa$ are parametrized by
\textbf{symmetric matrices}, $A^T=A$ and $C^T=C$.
Denote
\begin{align*}
\hbox{heis}_{2;\mathbb{C}}(n):= &\mathbb{C}\cdot\mathbf{1}\oplus
\mathbb{C}\hbox{--linear span of
} \\
&\left\{a^{\dagger}A a^{\dagger} \, , \,  a^{\dagger}B a \, , \, a
C a \ : \ B\in M_{n}(\mathbb{C}) \, , \, A, C\in M_{n,
sym}(\mathbb{C})\right\}\notag \  .
\end{align*}
Because of the linear independence of the set $\{\mathbf{1} \, ,
\, a^{\dagger}_ja^{\dagger}_k \, , \, a^{\dagger}_ja_k \, , \,
a_ja_k\}$, \newline $j,k\in\{1\dots, n\}$, $\hbox{heis}_{2;\mathbb{C}}(n)$ is the range of the vector space
isomorphism
\begin{align}
(c1_{M_{n}}, A,B,C)\in  \mathbb{C}\cdot 1_{M_{n}}\times M_{n,
sym}(\mathbb{C})\times M_{n}(\mathbb{C})\times M_{n,
sym}(\mathbb{C})  \label{lin-iso-heis-2C(n)} \\
\mapsto c\mathbf{1} + a^{\dagger}Aa^{\dagger} + a^{\dagger}Ba +
aCa\in \hbox{heis}_{2;\mathbb{C}}(n) \notag  \  .
\end{align}
In the following we will use the notation
\[
c\mathbf{1} \equiv c \qquad;\qquad c\in\mathbb{C}  \  .
\]
The involution \eqref{inv-Heis-n} on $\hbox{heis}_{1,\mathbb{C}}(n)$ induces an involution on
$\hbox{heis}_{2;\mathbb{C}}(n)$ given by
\[
(a^{\dagger}A a^{\dagger})^*= aA^*a \,\,;\,\, (a^{\dagger}B a)^*=
a^{\dagger}B^* a \,\,;\,\, (aC a)^*= a^{\dagger}C^* a^{\dagger}  \
.
\]
With this involution \eqref{lin-iso-heis-2C(n)} is a
$*$--isomorphism of $*$--vector spaces.
\begin{lemma}\label{le2}
$\hbox{heis}_{2;\mathbb{C}}(n)$ is a $*$--Lie algebra with
involution given by \eqref{inv-Heis-n}, central element
$\mathbf{1}$ and with the following commutation relations:
\begin{align}
\lbrack aMa, a^{\dagger}Na^{\dagger} \rbrack =& 2 \,\hbox{Tr}(NM)
+ 4\,a^{\dagger}NM a \,\,;\,\, M, N\in M_{n,sym}(\mathbb{C})\label{a}\  ,\\
\lbrack aMa, a^{\dagger}Na\rbrack = & a\left(MN + (MN)^T\right)a
\,\,;\,\, M\in M_{n,sym}(\mathbb{C}) \ , \   N\in
M_{n}(\mathbb{C})\label{b}\  ,\\
\lbrack a^{\dagger}Ma, a^{\dagger}Na\rbrack =&a^{\dagger}\lbrack
M,N \rbrack a \,\,;\,\, M, N\in M_{n}(\mathbb{C})\  ,\label{c}\\
\lbrack a^{\dagger}Ma^{\dagger}, a^{\dagger}Na^{\dagger} \rbrack
=&\lbrack aMa, aNa \rbrack= 0 \,\,;\,\, M, N\in
M_{n,sym}(\mathbb{C})\  ,\label{d}\\
\lbrack a^{\dagger}Ma, a^{\dagger}Na^{\dagger}\rbrack =&
a^{\dagger}\left(MN + (MN)^T\right)a^{\dagger} \quad; \  N\in
M_{n,sym}(\mathbb{C}) \, , \,   M\in M_{n}(\mathbb{C})\label{e}  \
.
\end{align}
\end{lemma}
\begin{proof}
Throughout this proof, summation over repeated indexes is understood.
We will use the algebraic identities
\begin{align}
\lbrack a_ia_j, a^{\dagger}_h a^{\dagger}_k\rbrack =&\lbrack
a_ia_j, a^{\dagger}_h\rbrack a^{\dagger}_k+a^{\dagger}_h\lbrack
a_ia_j,a^{\dagger}_k\rbrack \label{[aiaj,a+h a+k]}\\
=&a_i\lbrack a_j,a^{\dagger}_h\rbrack a^{\dagger}_k +\lbrack
a_i,a^{\dagger}_h\rbrack a_ja^{\dagger}_k+a^{\dagger}_ha_i\lbrack
a_j,a^{\dagger}_k\rbrack +a^{\dagger}_h\lbrack
a_i,a^{\dagger}_k\rbrack a_j \notag\\
=&a_ia^{\dagger}_k\delta_{jh}+a_ja^{\dagger}_k\delta_{ih}+a^{\dagger}_ha_i\delta_{jk}
+a^{\dagger}_ha_j\delta_{ik}\notag\\
=&\delta_{ik}\delta_{jh} + a^{\dagger}_ka_i\delta_{jh}
+\delta_{jk}\delta_{ih} + a^{\dagger}_ka_j\delta_{ih}
+a^{\dagger}_ha_i\delta_{jk}+a^{\dagger}_ha_j\delta_{ik}\notag\  ,
\end{align}
and
\begin{align}
 \lbrack a_i^{\dagger}a_j,a^{\dagger}_ha_k\rbrack =&\lbrack
a_i^{\dagger}a_j, a^{\dagger}_h\rbrack a_k+a^{\dagger}_h\lbrack
a_i^{\dagger}a_j, a_k\rbrack
\label{[ai+aj,ah+ak]}\\
=&a_i^{\dagger}\lbrack a_j, a^{\dagger}_h\rbrack
a_k+a^{\dagger}_h\lbrack a_i^{\dagger}, a_k\rbrack a_j\notag\\
=&a_i^{\dagger}a_k \delta_{j,h} - a^{\dagger}_ha_j\delta_{i,k}
\notag \ .
\end{align}
Equation (\ref{a}) follows from
\begin{align*}
&\lbrack aMa ,a^{\dagger}Na^{\dagger}\rbrack = M_{ij}N_{hk}\lbrack
a_ia_j, a^{\dagger}_ha^{\dagger}_k\rbrack \\
=&^{\eqref{[aiaj,a+h a+k]}} M_{ji}N_{hk}
\left(a_ia^{\dagger}_k\delta_{jh}+a_ja^{\dagger}_k\delta_{ih}+a^{\dagger}_ha_i\delta_{jk}
+a^{\dagger}_ha_j\delta_{ik}\right)\\
=&M_{ji}N_{hk}\left( \delta_{ik}\delta_{jh} +
a^{\dagger}_ka_i\delta_{jh} +\delta_{jk}\delta_{ih} +
a^{\dagger}_ka_j\delta_{ih}
+a^{\dagger}_ha_i\delta_{jk}+a^{\dagger}_ha_j\delta_{ik}\right)\\
= &M_{ji}N_{hk}\delta_{ik}\delta_{jh} +
M_{ji}N_{hk}a^{\dagger}_ka_i\delta_{jh}
+M_{ji}N_{hk}\delta_{jk}\delta_{ih} +\\
&M_{ji}N_{hk}a^{\dagger}_ka_j\delta_{ih}
+M_{ji}N_{hk}a^{\dagger}_ha_i\delta_{jk}
+M_{ji}N_{hk}a^{\dagger}_ha_j\delta_{ik}\\
=& M_{ji}N_{ji} + M_{ji}N_{jk}a^{\dagger}_ka_i +M_{ji}N_{ij} +
M_{ji}N_{ik}a^{\dagger}_ka_j +M_{ji}N_{hj}a^{\dagger}_ha_i
+M_{ji}N_{hi}a^{\dagger}_ha_j\\
=& 2M_{ji}N_{ij} + N_{kj}M_{ji}a^{\dagger}_ka_i
 + M_{ji}N_{ik}a^{\dagger}_ka_j
+N_{hj}M_{ji}a^{\dagger}_ha_i +M_{ji}N_{ih}a^{\dagger}_ha_j\\
=& 2\hbox{Tr }(MN) + (NM)_{ki}a^{\dagger}_ka_i
 + (MN)_{jk}a^{\dagger}_ka_j
 +(NM)_{hi}a^{\dagger}_ha_i
+(MN)_{jh}a^{\dagger}_ha_j\\
=& 2\hbox{Tr }(MN) + (NM)_{ki}a^{\dagger}_ka_i
 + (NM)_{kj}a^{\dagger}_ka_j
 +(NM)_{hi}a^{\dagger}_ha_i
+(NM)_{hj}a^{\dagger}_ha_j\\
=& 2\hbox{Tr }(MN) + (NM)_{ki}a^{\dagger}_ka_i
 + (NM)_{ki}a^{\dagger}_ka_i
 +(NM)_{ki}a^{\dagger}_ka_i
+(NM)_{kj}a^{\dagger}_ka_j\\
= &2\hbox{Tr }(MN) + 4\left(NM\right)_{ki} a^{\dagger}_ka_i = 2
\,\hbox{Tr}(NM) + 4\,a^{\dagger}NM a  \  .
\end{align*}
Equation (\ref{e}) follows from
\begin{align*}
&\lbrack a^{\dagger} M a, a^{\dagger} N a^{\dagger}\rbrack
=M_{ij}N_{hk}\lbrack a_i^{\dagger}a_j,
a^{\dagger}_ha^{\dagger}_k\rbrack\\
=&^{\eqref{[aiaj,a+h a+k]}}
M_{ij}N_{hk}\left(\delta_{ik}\delta_{jh} +
a^{\dagger}_ka_i\delta_{jh} +\delta_{jk}\delta_{ih} +
a^{\dagger}_ka_j\delta_{ih}
+a^{\dagger}_ha_i\delta_{jk}+a^{\dagger}_ha_j\delta_{ik} \right)\\
=&M_{ij}N_{hk}a_i^{\dagger}a_k^{\dagger}\delta_{jh}+M_{ij}N_{hk}a^{\dagger}_ha_i^{\dagger}\delta_{jk}
+M_{ij}N_{jk}a_i^{\dagger}a_k^{\dagger} +
M_{ij}N_{hj}a^{\dagger}_ha_i^{\dagger}\\
&+(MN)_{ik}a_i^{\dagger}a_k^{\dagger}
+M_{ij}N_{jh}a^{\dagger}_ha_i^{\dagger}\\
=&(MN)_{ik}a_i^{\dagger}a_k^{\dagger} +
(MN)_{ih}a^{\dagger}_ha_i^{\dagger}
+(MN)_{ik}a_i^{\dagger}a_k^{\dagger} +
(MN)_{ki}a_i^{\dagger}a^{\dagger}_k\\
&+(MN)_{ik}a_i^{\dagger}a_k^{\dagger} +
((MN)^T)_{ik}a_i^{\dagger}a^{\dagger}_k\\
=&a_i^{\dagger}((MN)+(MN)^T)_{ik}a_k^{\dagger}\\
=&a^{\dagger}((MN)+(MN)^T)a^{\dagger} \  .
\end{align*}
Equation (\ref{b}) is the adjoint of (\ref{e}). Equation (\ref{c})
follows from
\begin{align*}
\lbrack a^{\dagger}Ma, a^{\dagger}Na \rbrack =&M_{ij}N_{hk}\lbrack
a_i^{\dagger}a_j,a^{\dagger}_ha_k\rbrack
\\
=&^{\eqref{[ai+aj,ah+ak]}} M_{ij}N_{hk}(a_i^{\dagger}a_k
\delta_{j,h} - a^{\dagger}_ha_j\delta_{i,k})\\
=&M_{ij}N_{hk}a_i^{\dagger}\delta_{jh}a_k-M_{ij}N_{hk}a^{\dagger}_h\delta_{ki}a_j\\
=&a_i^{\dagger}M_{ij}N_{jk}a_k-a^{\dagger}_hN_{hi}M_{ij}a_j\\
=&a_i^{\dagger}(MN)_{ik}a_k-a^{\dagger}_h(NM)_{hj}a_j\\
=&a^{\dagger}(MN)a-a^{\dagger}(NM)a\\
 =&a^{\dagger}(MN-NM)a\\
=&a^{\dagger}\lbrack M,N\rbrack a  \  .
\end{align*}
The proof of (\ref{b}) follows directly from the commutation
relations
\[
\lbrack a_i^{\dagger}a^{\dagger}_j,
a^{\dagger}_ha^{\dagger}_k\rbrack=\lbrack a_ia_j, a_h a_k\rbrack=0
\  ,
\]
for all indices $i, j, h, k$.
\end{proof}
 The commutation relations \eqref{c}, \eqref{e} suggest that it is convenient to introduce the
composition law discussed in the following Lemma. In section \ref{adj-act-quad-grp} below
we will see that, using the $\circ$--notation, several formulas acquire a shape that strongly reminds
the corresponding result in the commutative case.
\begin{lemma}
For any $X,Y\in M_{d}(\mathbb{C})$, the binary composition law
\begin{equation}\label{ad-notat-X-circ-Y}
X\circ Y := XY + (XY)^T\  ,
\end{equation}
is commutative, distributive, complex bi--linear and the following properties hold.
\begin{align}
X\circ 1 =& X + X^T\label{ad-hat-notat}\  ,\\
(X\circ Y)^* =& Y^*\circ X^*\label{ad-circ-*}   \  .
\end{align}
\end{lemma}
\begin{proof} Commutativity,
distributivity, complex bi--linearity and \eqref{ad-hat-notat} are
clear. \eqref{ad-circ-*} follows from
\begin{align*}
(X\circ Y)^* := &(XY + (XY)^T)^* = Y^*X^* + ((XY)^T)^*\\
 =& Y^*X^* +((XY)^*)^T= Y^*X^* + (Y^*X^*)^T= Y^*\circ X^*   \  .
\end{align*}
\end{proof}
\begin{remark} \rm
The operation $\circ$ is neither commutative nor associative.
Restricted to the hermitian (or skew--hermitian) elements of
$\hbox{heis}_{2;\mathbb{C}}(n)$ it reduces to the \textbf{Jordan
product} which is commutative.
\end{remark}
\noindent With the $\circ$--notation, Lemma \ref{le2} becomes the following.
\begin{corollary}
For  $M, N\in M_{n}(\mathbb{C})$,
\begin{align}
\lbrack  aMa, a^{\dagger}Na^{\dagger}\rbrack
=& 2 \,\hbox{Tr}(NM) + 4\,a^{\dagger}NMa \,\, , \,\, M, N\in M_{n,sym}(\mathbb{C}) \label{ad-a}\  ,\\
\lbrack  aMa, a^{\dagger}Na\rbrack
=&  a(M\circ N)a  \,\, , \,\, M\in M_{n,sym}(\mathbb{C}) \ , \ N\in M_{n}(\mathbb{C}) \label{ad-b}\  ,\\
\lbrack   a^{\dagger}Ma, a^{\dagger}Na\rbrack
=&a^{\dagger}\lbrack   M,N \rbrack \,\, , \,\, M, N\in M_{n}(\mathbb{C}) a\label{ad-c}\  ,\\
\lbrack   a^{\dagger}Ma^{\dagger}, a^{\dagger}Na^{\dagger} \rbrack
=&\lbrack   aMa, aNa \rbrack  = 0\,\, , \,\, M, N\in M_{n,sym}(\mathbb{C}) \label{ad-d} \  ,\\
\lbrack  a^{\dagger}Ma^{\dagger},a^{\dagger}Na\rbrack
=&-a^{\dagger}(N\circ M)a^{\dagger}\,\, , \,\, M\in M_{n,sym}(\mathbb{C}) \ ,
\ N\in M_{n}(\mathbb{C}) \label{ad-e}  \  .
\end{align}
\end{corollary}
\begin{proof}
\eqref{ad-e} is, up to change of notations, the adjoint of \eqref{ad-b}. In fact
\begin{align*}
\eqref{ad-b} \iff& (\lbrack  aMa, a^{\dagger}Na\rbrack )^* = (a(M\circ N)a)^*\\
\iff &\lbrack  (a^{\dagger}Na)^*,(aMa)^*\rbrack  = a^{\dagger}(M\circ N)^*a^{\dagger}\\
\iff &\lbrack  a^{\dagger}N^*a,a^{\dagger}M^*a^{\dagger}\rbrack  =
a^{\dagger}(N^*\circ M^*)a^{\dagger}  \  .
\end{align*}
Since $M, N$ are arbitrary, one can replace $M$ by
$M^*$ and $N$ by $N^*$ obtaining
$$
\lbrack  a^{\dagger}Ma^{\dagger},a^{\dagger}Na\rbrack  =
-a^{\dagger}(N\circ M)a^{\dagger} \  .
$$
Exchanging the roles of $M$ and $N$, \eqref{ad-e} takes the form \eqref{e}, i.e.
$$
\lbrack   a^{\dagger}Ma, a^{\dagger}Na^{\dagger}\rbrack  =
a^{\dagger}(M\circ N)a^{\dagger}  \ .
$$
\end{proof}

\begin{definition}\label{n-dim-quadr-Bos-alg}
The \textbf{complex $n$--dimensional quadratic Boson algebra}, still denoted
$\hbox{heis}_{2;\mathbb{C}}(n)$, is the $*$--Lie algebra with linearly independent generators
\begin{equation}\label{param-heis(2;C)(n)}
\left\{\mathbf{1} \, , \, B^2_0(A) \, , \, B^1_1(B) \, , \,
B^0_2(C)  \, : \, A, C \in M_{n, sym}(\mathbb{C}) \ , \ B\in
M_{n}(\mathbb{C}) \right\} \  ,
\end{equation}
central element $\mathbf{1}$, involution given by
\begin{equation}\label{inv-Heis2-n-abstr}
B^2_0(A)^*=B^0_2(A^*) \,\,;\,\, B^1_1(B)^*=B^1_1(B^*) \,\,;\,\,
B^0_2(C)^*=B^2_0(C^*) \  ,
\end{equation}
and Lie brackets given by:
\begin{align}
\lbrack B^0_2(M), B^2_0(N) \rbrack =& 2 \,\hbox{Tr}(NM) + 4\,B^1_1(NM) \  ,\label{aa}\\
\lbrack B^1_1(M), B^1_1(N)\rbrack  =& B^1_1(\lbrack M,N \rbrack ) \  ,\label{cc}\\
\lbrack  B^0_2(M), B^1_1(N)\rbrack =&  B^0_2\left(M\circ N\right) \  ,\label{bb}\\
\lbrack B^2_0(M),B^1_1(N)\rbrack =& - B^2_0\left(N\circ M\right)\  ,\label{ee}\\
\lbrack B^2_0(M), B^2_0(N)\rbrack =&\lbrack B^0_2(M),
B^0_2(N)\rbrack=0\  . \label{dd}
\end{align}
The \textbf{the skew--adjoint} elements of $\hbox{heis}_{2;\mathbb{C}}(n)$, are a \textbf{real}
$*$--Lie sub--algebra denoted $\hbox{heis}_{2}(n)$, called the
\textbf{$n$--dimensional quadratic Heisenberg algebra}.
The local Lie groups associated to $\hbox{heis}_{2;\mathbb{C}}(n)$ and $\hbox{heis}_{2}(n)$
are denoted respectively $\hbox{Heis}_{2;\mathbb{C}}(n)$ and $\hbox{Heis}_{2}(n)$.
\end{definition}
\begin{remark} \rm
With the additional prescriptions
\begin{equation}\label{non-symm-test-matr}
B^2_0(A) := B^2_0(A^T) \,\,;\,\, B_2^0(A) := B_2^0(A^T)
\,\,;\,\,\forall A\in M_{n}(\mathbb{C}) \  ,
\end{equation}
(automatically satisfied in the Boson realization (see \eqref{emerg-symm-matr}))
allows to replace the parametrization of $\hbox{heis}_{2;\mathbb{C}}(n)$ given by
\eqref{param-heis(2;C)(n)}
\begin{equation}\label{not-1-1-param-heis(2;C)(n)}
\left\{\mathbf{1} \, , \, B^2_0(A) \, , \, B^1_1(B) \, , \,
B^0_2(C)  \, : \, A, B, C \in M_{n}(\mathbb{C}) \right\} \   .
\end{equation}
The advantage of \eqref{param-heis(2;C)(n)} is that it is one--to--one.
\end{remark}
\begin{remark} \rm
In order to simplify notations we use the \textbf{same symbol}
$\hbox{heis}_{2;\mathbb{C}}(n)$ for the abstract $*$--Lie algebra
and for its Boson realization and the same for their central
element. The identifications
\[
B^2_0(A) \equiv a^{\dagger}Aa^{\dagger}\,\,;\,\, B^1_1(B) \equiv
a^{\dagger}Ba\,\,;\,\, B^0_2(C) \equiv aCa \  ,
\]
are clear from the commutation relations and in the following we will use them
constantly because the Boson notation are more intuitive. These identifications
should not create confusion provided they are handled carefully.
For example the following identities,  where summation on repeated indices is
understood, make sense in a Boson context:
\begin{align}
aAa^{\dagger} = &a_{i}(A)_{ij}a^{\dagger}_{j} =
a_{i}A_{ij}a^{\dagger}_{j} =
A_{ij}(\lbrack  a_{i},a^{\dagger}_{j}\rbrack  + a^{\dagger}_{j}a_{i})\label{Bos-1st-ord-NO-aACa+}\\
=& A_{ij}(\delta_{ij} + a^{\dagger}_{j}a_{i}) = A_{ij}\delta_{ij}
+ A_{ij}a^{\dagger}_{j}a_{i} = A_{ii}\delta_{ij} +
a^{\dagger}_{j}(A^T)_{ji}a_{i}\notag\\
=& \hbox{Tr}(A) + a^{\dagger}_{j}((A)^T)_{ji}a_{i} \notag\\
= &\hbox{Tr}(A) + a^{\dagger}(A)^Ta\notag \  ,
\end{align}
but, while expressions of the form $aAa^{\dagger}$ do not make
sense in the quadratic Boson algebra, both terms in the sum
appearing in the last identity in \ref{Bos-1st-ord-NO-aACa+} are
in the quadratic Boson algebra.
\end{remark}
\begin{remark} \rm
In the physics literature, instead of \eqref{inv-Heis2-n-abstr}, one uses the involution
\[
B^1_1(B)^*=B^1_1(B^*) \,\,;\,\, B^0_2(C)^*=B^2_0(C)\ ,
\]
where annihilators depend \textbf{anti--linearly} on their \textbf{test matrices}.
This has the advantage that annihilators are simply defined as adjoints of creators
and one is not obliged to define an involution on the test function space.
However, since most test function spaces concretely used have a natural involution, the choice
\eqref{inv-Heis2-n-abstr} seems to be more natural.
\end{remark}
\begin{lemma}\label{lm-Dioph-eq}
The pairs $(N,n)\in\mathbb{N}^2$ such that $\hbox{sl}(N,\mathbb{R})$
and $\hbox{heis}_{2}(n)$ are isomorphic as vector spaces if and
only if $N$ and $n$ have the form
\[
n=2 n_{1} +1 \,\,;\,\, N= 2(2 p_{1} +1) \,\,;\,\, n_{1}, p_{1}
\in\mathbb{N}\  ,
\]
where the pair $(n_{1}, p_{1})\in\mathbb{N}^2$ is any solution of the
quadratic diophantine equation
\[
2(2 p_{1} +1)^{2} = (2 n_{1} +1)^{2} + 1  \  .
\]
\end{lemma}
\begin{proof}
 \eqref{lin-iso-heis-2C(n)} implies that the complex dimension of
 $\hbox{heis}_{2;\mathbb{C}}(n)$ as a vector space is
$$
1+ 2 \frac{n(n+1)}{2} + n^{2} =1 +  n^{2} + n +  n^{2} =2n^{2} + n
+ 1\  ,
$$
while the real dimension of $\hbox{heis}_{2}(n)$ is
$$
1 + n(n+1) + n + (n-1)n = 1 + n^{2} + n + n^{2} - n = 1 + 2n^{2}\
,
$$
The real dimension of $\hbox{sl}(N,\mathbb{R})$ is $N^{2} - 1$.
Therefore  $\hbox{heis}_{2}(n)$ can be isomorphic to
$\hbox{sl}(N,\mathbb{R})$ as vector space only if the pair
$(n,N)\in\mathbb{N}^{2}$ satisfies the equation
\begin{equation}\label{Dioph1}
N^{2} - 1 = 2n^{2} +1 \iff N^{2} = 2n^{2} + 2  \  .
\end{equation}
If this is the case, given $n$ (resp. $N$) $N$ (resp. $n$) is
uniquely determined. Since odd numbers are a multiplicative
semi-group, given $n$, \textbf{$N$ has to be an even number}:
$$
N = 2p \  .
$$
In this case \eqref{Dioph1} becomes equivalent to
\begin{equation}\label{Dioph2}
(2p)^{2} = 2n^{2} + 2 \iff 2p^{2} = n^{2} + 1 \iff  2p^{2}-1=
n^{2}  \  .
\end{equation}
This shows that a necessary condition for \eqref{Dioph2} to have a solution is that \textbf{$n$ is odd}:
\begin{equation}\label{Dioph3}
n=2 n_{1} +1  \  .
\end{equation}
In this case \eqref{Dioph2} becomes
\begin{align}
2p^{2}-1 =& (2 n_{1} +1)^{2} = 4 n_{1}^{2} + 4 n_{1} + 1\notag\\
 \iff & 2p^{2} = 4 n_{1}^{2} + 4 n_{1} + 2\notag \\
\iff & p^{2} = 2  n_{1}( n_{1} + 1) + 1 \label{Dioph4}  \  .
\end{align}
It follows, for the same reason as above, that a necessary condition
for \eqref{Dioph4} to have a solution is that \textbf{$p$ is odd}:
\begin{equation}\label{Dioph5}
p=2 p_{1} +1 \  .
\end{equation}
In view of \eqref{Dioph3} and \eqref{Dioph5}, \eqref{Dioph1} becomes
\begin{align}
(2(2 p_{1} +1))^{2} =& 2(2 n_{1} +1)^{2} + 2\label{Dioph6}\\
 \iff &4(2 p_{1} +1)^{2} = 2(2 n_{1} +1)^{2} + 2 \notag\\
 \iff & 2(2 p_{1} +1)^{2}
= (2 n_{1} +1)^{2} + 1 \notag  \  .
\end{align}
\end{proof}
\begin{remark} \rm
Equation \eqref{Dioph6} for given $n_{1}$ has the non--trivial
solution
$$
n_{1} = 3 \Rightarrow (2 n_{1} +1)^{2} + 1 = 50 = 2 \cdot 5^2\  ,
$$
while for given $p_{1}$ it has the non--trivial solution
$$
p_{1} = 2 \Rightarrow  2(2 p_{1} +1)^{2} -1 = 49 = 7^2  \  .
$$
More generally one can prove, by direct computation, that equation
\eqref{Dioph1} has non--trivial solutions in the following cases:
\begin{align*}
n=1\Rightarrow& N^{2} = 2n^{2} + 2 = 4 = 2^2 \Rightarrow N=2 \  , \\
n=2\Rightarrow& N^{2} = 2n^{2} + 2 = 10  \ \Rightarrow \ \hbox{no
solutions}\  ,\\
n=3\Rightarrow & N^{2} = 2n^{2} + 2 = 20  \ \Rightarrow \ \hbox{no
solutions}\  ,\\
n=4\Rightarrow & N^{2} = 2n^{2} + 2 = 34 \ \Rightarrow \ \hbox{no
solutions}\  ,\\
n=5\Rightarrow & N^{2} = 2n^{2} + 2 = 52 \ \Rightarrow \ \hbox{no
solutions}\  ,\\
n=6\Rightarrow & N^{2} = 2n^{2} + 2 = 74 \ \Rightarrow \ \hbox{no
solutions}\  ,\\
n=7\Rightarrow & N^{2} = 2n^{2} + 2 = 100 = 10^2 \ \Rightarrow
N=10\  .
\end{align*}
Therefore non--trivial solutions exist.
\end{remark}

\subsection{Group elements and their $1$--st and $2$--d kind coordinates}

In this section, we identify the first and second kind coordinates in the case of
$heis_{\mathbb{C}}(2;n)$.\\
The elements of $heis_{\mathbb{C}}(2;n)$ are parametrized by quadruples
\[
(x,A,B,C)\in\mathbb{C}\times M_{n,sym}(\mathbb{C})\times
M_{n}(\mathbb{C})\times M_{n,sym}(\mathbb{C})\  ,
\]
and we consider the natural topology induced by this
parametrization.  From this section on we suppose that, for
$(z,A,B,C)$ near the origin the corresponding element of
$heis_{\mathbb{C}}(2;n)$ can be exponentiated in the sense that
the corresponding exponential series converges on a dense
sub--space of the representation space. In section
\ref{Exp-quadr-alg-in-Fck-repr} below we prove that this is always
the case in the Fock representation.

\smallskip

 Following the general theory of Lie groups, we say that
the quadruple \newline $(x,A,B,C)$ defines the \textbf{second kind coordinates} of
\[
G(x,A,B,C) =e^{x\mathbf{1}}e^{ a^{\dagger}A
a^{\dagger}}e^{a^{\dagger}B a}e^{ a C a} =e^{x\mathbf{1}}e^{
B^2_0(A)}e^{ B^1_1(B)}e^{ B^0_2(C)} \in Heis_{\mathbb{C}}(2;n) \
,
\]
and the \textbf{ first kind coordinates} of
\begin{equation}\label{df-W(x,A,B,C)}
W(x,A,B,C)=e^{x\mathbf{1}+ a^{\dagger}A a^{\dagger}+ a^{\dagger}B
a+ a C a}=e^{x\mathbf{1}+ B^2_0(A )+ B^1_1(B)+ B^0_2(C)} \in
Heis_{\mathbb{C}}(2;n)\  .
\end{equation}
In both representations one can find a \textbf{sub--set} of the whole domain of the coordinates, i.e.
$\mathbb{C}\cdot 1_{M_{n}}\times M_{n, sym}(\mathbb{C})\times M_{n}(\mathbb{C})\times
M_{n, sym}(\mathbb{C})$, in which the correspondence
$$
W(x,A,B,C)\mapsto (x,A,B,C) \qquad;\qquad G(x,A,B,C)\mapsto
(x,A,B,C)\  ,
$$
is one--to--one. This domain can be considered as an \textbf{embedding} of the group manifold
of $\hbox{Heis}_{2;\mathbb{C}}(n)$ into the vector space $\mathbb{C}\cdot 1_{M_{n}}\times
M_{n, sym}(\mathbb{C})\times M_{n}(\mathbb{C})\times M_{n, sym}(\mathbb{C})$.
\textbf{On this domain}
 the group multiplication law induces a group composition law through the identities
\begin{align*}
W(x_1,A_1,B_1,C_1)W(x_2,A_2,B_2,C_2)=:&W\left((x_1,A_1,B_1,C_1)\diamond_{1}(x_2,A_2,B_2,C_2)\right) \ ,\\
G(x_1,A_1,B_1,C_1)G(x_2,A_2,B_2,C_2)=:&G\left((x_1,A_1,B_1,C_1)\diamond_{2}(x_2,A_2,B_2,C_2)\right)
\  .
\end{align*}
 Typically both composition laws $\diamond_{1}$ and
$\diamond_{2}$ are strongly non--linear functions of the
coordinates. In the $1$--dimensional case and for the sub--group
$\hbox{Heis}_{2}(1)$ of $\hbox{Heis}_{2;\mathbb{C}}(1)$, i.e. up to
isomorphism $\hbox{sl}(2,\mathbb{R})$, both the domain and the
explicit form of $\diamond_{1}$ were determined in the paper
\cite{[AcOuRe10-QuadrHeisGroup]}. Our goal is to extend this
result to the multi--dimensional case.

\subsection{$*$--Lie sub--algebras of $\hbox{heis}_{2;\mathbb{C}}(n)$}

\noindent For special $*$--Lie sub--algebras of
\newline
$\hbox{heis}_{2;\mathbb{C}}(n)$, the formulas of the
exponential of their elements are considerably simplified. The
following Proposition describes the spaces of \textbf{test
matrices} that parametrize some natural $*$--Lie sub--algebras of
$\hbox{heis}_{2;\mathbb{C}}(n)$.
\begin{proposition}\label{cor-quadr-Heisn}
Let $\mathcal{L}$ be a $*$--Lie sub--algebra of $\hbox{heis}_{2;\mathbb{C}}(n)$. Denote
\begin{align*}
 M_{d}(2,0)
 :=&\{A\in M_{d}(\mathbb{C}) \ : \ B^2_0(A) \in
 \mathcal{L}\}\  ,\\
 M_{d}(1,1)
 :=&\{B\in M_{d}(\mathbb{C}) \ : \ B^1_1(B) \in
 \mathcal{L}\}\  ,\\
 M_{d}(0,2)
 :=&\{C\in M_{d}(\mathbb{C}) \ : \ B^0_2(C)  \in \mathcal{L}\}\  .
\end{align*}
Then $M_{d}(1,1)$ is a sub--$*$--Lie algebra of $M_{d}(\mathbb{C})$. If in addition
\[
\textbf{1}\in  M_{d}(2,0)\  ,
\]
then $M_{d}(1,1)$ is a $*$--sub--algebra of $M_{d}(\mathbb{C})$ \textbf{closed under conjugation}
(or equivalently under transposition) and
\begin{equation}\label{C(a+a+)=C(a+a)=C(aa)}
M_{d}(2,0) = M_{d}(1,1) = M_{d}(0,2) \  .
\end{equation}
\end{proposition}

\begin{proof}
The assumption that $\mathcal{L}$ is closed under involution and \eqref{inv-Heis2-n-abstr} imply that
\begin{equation}\label{C(a+a+)*=C(aa)}
M_{d}(2,0)^* = M_{d}(0,2) \qquad;\qquad M_{d}(1,1)^* = M_{d}(1,1)
\  .
\end{equation}
The assumption that $\mathcal{L}$ is a Lie algebra and the linearity of the maps $B^j_k( \ \cdot \ )$
imply that
$M_{d}(2,0)$, $M_{d}(1,1)$, $M_{d}(0,2)$ are vector spaces.
\eqref{c} implies that, if $M, N\in M_{d}(1,1)$, then
$\lbrack M,N \rbrack\in M_{d}(1,1)$.
Hence \eqref{C(a+a+)*=C(aa)} implies that $M_{d}(1,1)$
is a sub--$*$--Lie algebra of $M_{d}(\mathbb{C})$.
From \eqref{b} it follows that, if $M\in M_{d}(0,2)$ and $N\in M_{d}(1,1)$ then,
(see \eqref{non-symm-test-matr}) since $a\left(MN+(MN)^T\right)a= aMNa$,
$MN\in M_{d}(0,2)$. Equivalently
\begin{equation}\label{incl-C(aa)C(a+a)}
M_{d}(0,2)M_{d}(1,1)\subseteq M_{d}(0,2) \  .
\end{equation}
But, because of \eqref{C(a+a+)*=C(aa)}, $\textbf{1}\in
M_{d}(2,0)\iff \textbf{1}\in M_{d}(0,2)$. Therefore, if
$\textbf{1}\in M_{d}(2,0)$, then $M_{d}(1,1)\subseteq M_{d}(0,2)$.
On the other hand, \eqref{a} implies that \newline $
M_{d}(2,0)M_{d}(0,2)\subseteq M_{d}(1,1)$. By the same argument,
 $M_{d}(0,2)\subseteq M_{d}(1,1)$. Therefore $M_{d}(0,2)
= M_{d}(1,1)$ and \eqref{C(a+a+)*=C(aa)} implies that also
$M_{d}(2,0) = M_{d}(1,1)$. But then \eqref{incl-C(aa)C(a+a)}
becomes
$$
M_{d}(0,2)M_{d}(1,1)=M_{d}(1,1)M_{d}(1,1) \subseteq M_{d}(1,1) \
,
$$
i.e. $M_{d}(1,1)$ is an algebra. Since it is closed under involution, it
is a sub--$*$--algebra $M_{d}(\mathbb{C})$. From \eqref{bb}, \eqref{C(a+a+)=C(a+a)=C(aa)}
and the fact that $M_{d}(1,1)$ is a $*$--algebra with identity, it follows that $M_{d}(1,1)$ is
closed under transposition and, by the $*$--algebra property, this is equivalent to be
closed under conjugation. Conversely, if $M_{d}(1,1)$ is a $*$--algebra (not necessarily
with identity) closed under transposition, then the set
$$
\mathcal{L}:=\left\{ B^2_0(A) , B^1_1(B) , B^0_2(C) \ : \ A, B,
C\in M_{d}(1,1) \right\}\  ,
$$
is closed under the involution \eqref{inv-Heis2-n-abstr} and under the Lie brackets
 \eqref{a}, $\dots$, \eqref{e}, i.e. it is a $*$--Lie sub--algebra of $\hbox{heis}_{2;\mathbb{C}}(n)$.
\end{proof}
\begin{remark}\label{sp} \rm
From \eqref{c} it follows that, for any real or complex Lie sub--algebra
$\mathcal{L}_{n}$ of $M_{n}(\mathbb{C})$ the family
\[
\Lambda_{2}(\mathcal{L}_{n}):=\{ a^{\dagger}Ma \ : \ M\in
\mathcal{L}_{n} \} \  ,
\]
is a real or complex Lie sub--algebra (resp. $*$--Lie sub--algebra) of $heis_{\mathbb{C}}(2;n)$
called \textbf{the quadratic preservation algebra of order $n$}. If
$\mathcal{L}_{n}=M_{n}(\mathbb{C})$, we simply write $\Lambda_{2,n}$.
\end{remark}

\section{The Splitting Lemma }

In this section, we recall, in our notations, Feinsilver--Pap's splitting lemma
\cite{FP} which will be our basic tool for the calculation of vacuum characteristic functions
of homogeneous quadratic fields.
Formulas expressing second kind coordinates in terms of first kind ones are called
\textbf{splitting or disentangling formulas}.
In the case of $Heis_{\mathbb{C}}(2;n)$, they are given by the following lemma.
\begin{lemma}\label{sl}
For $A, C \in Sym\left(M_{n}(\mathbb{C})\right)$, $B\in M_{n}(\mathbb{C})$, define
\[
v:=
\begin{pmatrix}
  B & 2 A \\
  -2 C & - B^T
\end{pmatrix}\  ,
\]
and $P, Q, R, S$ by
\[
\begin{pmatrix}
  P(t)  & Q(t) \\
  -R(t) & S(t)
\end{pmatrix}
:=e^{tv} \  .
\]
Then, for $t\in \mathbb{R}$ sufficiently close to $0$:
\begin{align*}
e^{t\left(a^{\dagger}Aa^{\dagger}+a^{\dagger}Ba+aCa\right)}
=&e^{-\frac{t}{2}\,\hbox{Tr}(B) +\frac{1}{2}\hbox{Tr}\left(g_t(A,B,C)\right)} e^{\frac{1}{2}a^{\dagger}\hat
f_t(A,B,C)a^{\dagger}}\\
& \cdot e^{a^{\dagger}g_t(A,B,C) a}e^{\frac{1}{2}a \hat h_t(A,B,C)
a}\  ,
\end{align*}
or, in the $B$-notation
\begin{align*}
e^{t\left(B^2_0(A)+B^1_1(B)+B^0_2(C)\right)}
=&e^{-\frac{t}{2}\,\hbox{Tr}(B) +\frac{1}{2}\hbox{Tr}\left(g_t(A,B,C)\right)} e^{\frac{1}{2}B^2_0(\hat
f_t(A,B,C))}\\
& \cdot e^{ B^1_1(g_t(A,B,C) )} e^{\frac{1}{2}B^0_2(\hat
h_t(A,B,C))}\  ,
\end{align*}
 where
\[
f_t(A,B,C)=Q(t)S(t)^{-1} \,,\, g_t(A,B,C)=-\log S(t)^{T} \,,\,
h_t(A,B,C)= S(t)^{-1}R(t)\  ,
\]
and $\hat f=(f+f^T)/2, \hat h=(h+h^T)/2$ denote the symmetric parts of $f, h$.
\end{lemma}
\begin{remark} \rm
The invertibility of $S(t)$ for each $t\in\mathbb{R}$ is proved in
Lemma \ref{S(t),P(t)-invert} below.
\end{remark}
\begin{proof}
From Lemma 6 of \cite{FP}, with the following change of notations with respect to those used there:
\[
R_{2 A}\to a^{\dagger}Aa^{\dagger}\,,\, \Delta_{2 C} \to aCa\,,\,
\rho_{B}\to a^{\dagger}Ba+\hbox{Tr}(B/2)\  ,
\]
we obtain
$$
e^{t\left(a^{\dagger}Aa^{\dagger}+a^{\dagger}Ba+aCa\right)}
=e^{-\frac{t}{2}\,\hbox{Tr}(B)}e^{t\left(R_{2 A}+\rho_{B}+\Delta_{2
C}\right)} =e^{-\frac{t}{2}\,\hbox{Tr}(B)}e^{ R_{A_4(t)} } e^{
\rho_{A_5(t)}} e^{ \Delta_{A_6(t)} }\  ,
$$
where
\[
A_4(t)=f_t(A,B,C)\,,\, A_5(t)=g_t(A,B,C)\,,\, A_6(t)=h_t(A,B,C)\
,
\]
with $f_t, g_t, h_t$ as in the statement of this Lemma. Thus, using
\[
e^{aXa}=e^{a \hat X a}\,,\, e^{ a^{\dagger}X
a^{\dagger}}=e^{a^{\dagger}\hat X a^{\dagger}}\  ,
\]
one finds
\begin{align*}
&e^{t\left(a^{\dagger}Aa^{\dagger}+a^{\dagger}Ba+aCa\right)}
=e^{-\frac{t}{2}\,\hbox{Tr}(B)}e^{ R_{A_4(t)} } e^{ \rho_{A_5(t)}} e^{ \Delta_{A_6(t)} }\\
=&e^{-\frac{t}{2}\,\hbox{Tr}(B)}e^{ R_{f_t(A,B,C)} } e^{\rho_{g_t(A,B,C)} } e^{ \Delta_{h_t(A,B,C)} }\\
=&e^{-\frac{t}{2}\,\hbox{Tr}(B)}e^{a^{\dagger}\frac{1}{2}f_t(A,B,C)a^{\dagger}}
e^{a^{\dagger}g_t(A,B,C) a
 + \frac{1}{2}\hbox{Tr}\left(g_t(A,B,C)\right) } e^{ a \frac{1}{2}h_t(A,B,C)}\\
=&e^{-\frac{t}{2}\,\hbox{Tr}(B) +\frac{1}{2}\hbox{Tr}\left(g_t(A,B,C)\right)}e^{\frac{1}{2}a^{\dagger}f_t(A,B,C)a^{\dagger}}
e^{a^{\dagger}g_t(A,B,C) a}e^{\frac{1}{2}a h_t(A,B,C) a}\\
=&e^{-\frac{t}{2}\,\hbox{Tr}(B) +\frac{1}{2}\hbox{Tr}\left(g_t(A,B,C)\right)}e^{\frac{1}{2}a^{\dagger}\hat
f_t(A,B,C)a^{\dagger}} e^{ a^{\dagger}g_t(A,B,C) a}e^{\frac{1}{2}a
\hat h_t(A,B,C) a} \  .
\end{align*}
\end{proof}
\begin{lemma}\label{S(t),P(t)-invert}
$S(t)$ and $P(t)$ are invertible for $t\in\mathbb{R}$ with $|t|$ sufficiently small.
\end{lemma}
\begin{proof}
For any $s,t\in\mathbb{R}$, one has
\begin{align*}
e^{sv}e^{tv}=&\begin{pmatrix}
P_{s} & Q_{s}\\
-R_{s} & S_{s}
\end{pmatrix}
\begin{pmatrix}
P_{t} & Q_{t}\\
-R_{t} & S_{t}
\end{pmatrix}
=\begin{pmatrix}
P_{s}P_{t}- Q_{s}R_{t} &   P_{s}Q_{t}+ Q_{s}S_{t}\\
-R_{s}P_{t}- S_{s}R_{t} &  -R_{s}Q_{t}+S_{s}S_{t}
\end{pmatrix}\\
=&\begin{pmatrix}
P_{s+t} & Q_{s+t}\\
-R_{s+t} & S_{s+t}
\end{pmatrix}
=e^{(s+t)v} \  .
\end{align*}
Moreover
$$
\begin{pmatrix}
P_{0} &  Q_{0}\\
-R_{0} & S_{0}
\end{pmatrix}=
\begin{pmatrix}
1 & 0\\
0 & 1
\end{pmatrix} \  .
$$
By continuity, for each $\epsilon>0$ there exists $t_{\epsilon}$
such that, for any $t\in\mathbb{R}$ with $|t|<t_{\epsilon}$,
\begin{equation}\label{PsSs-near-1}
\|P_{t}-1\| \ , \ \|S_{t}-1\|< \epsilon  \,\,;\,\, \|Q_{t}\| \ , \
\|R_{t}\|<\epsilon  \  .
\end{equation}
Since the set of invertible elements in $M_{n}(\mathbb{C})$ is
open, $\epsilon$ can be chosen so that if $X\in M_{n}(\mathbb{C})$
is such that
\begin{equation}\label{X-inv}
\|X-1\| < \epsilon  \end{equation} then $X$ is invertible. In
particular, if $\epsilon$ is as in \eqref{X-inv} and $P_{s}$ and
$S_{s}$ satisfy \eqref{PsSs-near-1}, then they are invertible.
\end{proof}
\begin{remark} \rm
In general $P_{t}$ and $S_{t}$ \textbf{are not invertible for
each} $t\in\mathbb{R}$. For example, taking
\[
A=\frac{1}{2}i\mathbf{I\ }\quad ,\quad C=-\frac{1}{2}i\mathbf{I}=
A^* \quad ,\quad \ B=0\  ,
\]
$\mathbf{I}$ being the identity matrix, one finds
$$
P\left(  t\right)  =S\left(  t\right)  =\cos\left(  t\right)
\mathbf{I\ } \qquad,\qquad  \forall t \  ,
$$
which is identically zero for $t=\left(  k+\frac{1}{2}\right)  \pi$ for any $k\in\mathbb{Z}$.
\end{remark}
\begin{remark}\label{r} \rm
If $A=0$ then $v$ and $e^{tv}$ are both lower triangular,
therefore $Q=0$, so $f_t(0,B,C)=0$ as well. Similarly, if $C=0$
then $v$ and $e^{tv}$ are both upper triangular, therefore $R=0$,
so $h_t(A,B,0)=0$ as well.
\end{remark}
\begin{notation}\label{n1}
For $t=1$ we will denote the functions $f_t, g_t, h_t$ of Lemma \ref{sl} by
$f, g, h$ respectively.
\end{notation}
\begin{notation}\label{n2}
For $x\in\mathbb{C}$ we denote $e^{x\mathbf{1}}$ by just $e^{x}$.
\end{notation}
\begin{notation}\label{n3}
In view of Lemma \ref{sl},
$$
E=W(x,A,B,C)=G(x',A',B',C')
$$
where
\begin{align*}
x'=&x-\frac{1}{2}\,\hbox{Tr}(B) +\frac{1}{2}\hbox{Tr}\left(g(A,B,C)\right)\\
A'=&\frac{1}{2}\hat f(A,B,C)   \,\,;\,\, B'=  g(A,B,C) \,\,;\,\,
C'=\frac{1}{2} \hat h(A,B,C)  \  .
\end{align*}
\end{notation}
\begin{remark}\label{BCH} \rm
By the Baker- Campbell -Hausdorff formula (see \cite{Hall}), for
all $M, N\in M_{n}(\mathbb{C})$\  ,
$$
e^Me^N =e^{M+\int_0^1 g\left(e^{\hbox{ad}_M }e^{ t
\hbox{ad}_N}\right)(N)\,dt } =e^{M+N+\it{BCH}(M,N)}\  ,
$$
where
$$
g(z)=\frac{\log z}{1-\frac{1}{z}}\  ,
$$
and
\begin{align*}
BCH(M,N)
=&\frac{1}{2}\lbrack M,N \rbrack
+\frac{1}{12}\left(\lbrack M,\lbrack M,N\rbrack\rbrack
+ \lbrack N,\lbrack N,M\rbrack\rbrack\right)
-\frac{1}{24}\lbrack N,\lbrack M, \lbrack M, N \rbrack \rbrack \rbrack \\
&+\mbox{ higher order commutators....}  \  ,
\end{align*}
(see Chapter 5 of \cite{Hall}).
\end{remark}
\begin{lemma}\label{exp}
For  $M, N\in M_{n}(\mathbb{C})$,
\begin{align}
e^{aMa}e^{ a^{\dagger}N a^{\dagger}} =& e^{\frac{1}{2}\hbox{Tr}\left(g\left(N,4NM,2\left(MNM+\left(MNM\right)^T
\right)\right)\right)} \label{11}\\
&\cdot e^{\frac{1}{2}a^{\dagger}\hat f
\left(N,4NM,2\left(MNM+\left(MNM\right)^T
\right)\right)a^{\dagger}}\notag\\
&\cdot
e^{a^{\dagger}g\left(N,4NM,2\left(MNM+\left(MNM\right)^T\right)\right)a}\notag\\
&\cdot e^{\frac{1}{2}a
\hat{h}\left(N,4NM,2\left(MNM+\left(MNM\right)^T\right)\right)a}\notag\  ,\\
e^{aMa}e^{ a^{\dagger}N a} =&e^{\hbox{Tr}\left(-\frac{1}{2}N+\frac{1}{2}g\left(0,N,M
N+(MN)^T\right)\right)}\label{22}\  ,\\
&\cdot e^{ a^{\dagger}g\left(0,N,M N+(M N)^T)\right) a} \cdot
e^{a\left(M+\frac{1}{2}\hat h\left(0,N,M
N+(MN)^T\right)\right)a}\notag\  ,\\
e^{a^{\dagger}Ma}e^{ a^{\dagger}N a^{\dagger}} =&e^{\hbox{Tr}\left(-\frac{1}{2}M + \frac{1}{2}g\left(M N+(MN)^T, M,
0\right)\right)}\label{33} \\
&\cdot e^{a^{\dagger}\left(N+\frac{1}{2}\hat f\left(M N+(MN)^T, M,
0\right) \right)a^{\dagger}} \cdot e^{ a^{\dagger}g\left(M
N+(MN)^T, M, 0\right) a}\notag\  ,
\end{align}
where in (\ref{11}) $M, N\in M_{n,sym}(\mathbb{C})$, in (\ref{22})
$M\in Sym \left(M_{n}(\mathbb{C})\right)$, and  in (\ref{33})
$N\in M_{n,sym}(\mathbb{C})$. In the $B$-notation, (\ref{11})-(\ref{33}) take
the form
\begin{align}
e^{B^0_2(M)}e^{ B^2_0(N)} =& e^{\frac{1}{2}\hbox{Tr}\left(g\left(N,4NM,2\left(MNM+\left(MNM\right)^T
\right)\right)\right)}\label{1}\\
&\cdot
 e^{\frac{1}{2}B^2_0\left(\hat f\left(N,4NM,2\left(MNM+\left(MNM\right)^T
 \right)\right)\right)}\notag\\
 & \cdot e^{B^1_1\left(g\left(N,4NM,2\left(MNM+\left(MNM\right)^T
\right)\right)\right)} \notag\\
&\cdot e^{\frac{1}{2}B^0_2\left( \hat
h\left(N,4NM,2\left(MNM+\left(MNM\right)^T\right)\right)\right)}\notag\  ,\\
 e^{B^0_2(M)}e^{ B^1_1(N)} =&e^{\hbox{Tr}\left(-\frac{1}{2}N +\frac{1}{2}g\left(0,N,M
N+(MN)^T\right)\right)}\label{2} \\
 & \cdot
e^{B^1_1\left(g\left(0,N,M N+(M N)^T)\right) \right)} \cdot
e^{B^0_2\left(M+\frac{1}{2}\hat h\left(0,N,M N+(M
N)^T\right)\right)}\notag\  ,\\
 e^{B^1_1(M)}e^{ B^2_0(N)} =&e^{\hbox{Tr}\left(-\frac{1}{2}M+\frac{1}{2}g\left(M N+(MN)^T, M,
0\right)\right)}\label{3}\  ,\\
& \cdot e^{B^2_0\left(N+\frac{1}{2}\hat f\left(M N+(MN)^T,
M,0\right)\right)} \cdot e^{ B^1_1\left(g\left(M N+(MN)^T, M,
0\right)\right)}\notag\  ,
\end{align}
 where $f, g, h$ are as in Notation
\ref{n1} and  $\hat f, \hat h$ denote the symmetric parts of $f,h$.
\end{lemma}
\begin{proof} Using Lemma \ref{ad-e[Y,.]X} we have
\begin{align*}
e^{aMa}e^{ a^{\dagger}N a^{\dagger}}=&\left(e^{aMa}e^{ a^{\dagger}N a^{\dagger}}e^{-aMa}\right) e^{aMa}\\
=&e^{a^{\dagger}N a^{\dagger}+\lbrack aMa, a^{\dagger}N
a^{\dagger}\rbrack+ \frac{1}{2!}\lbrack aMa,\lbrack aMa,
a^{\dagger}N a^{\dagger}\rbrack \rbrack+\cdots} e^{aMa} \  .
\end{align*}
By Lemma  \ref{le2},
\begin{align*}
\lbrack aMa, a^{\dagger}N a^{\dagger}\rbrack=& 2 \,\hbox{Tr}(M N)
+ 4\,a^{\dagger}NMa \  ,\\
 \lbrack aMa,\lbrack aMa, a^{\dagger}N
a^{\dagger}\rbrack \rbrack=&\lbrack aMa, 2 \,\hbox{Tr}(M N) +
4\,a^{\dagger}NMa\rbrack\\
=&4\lbrack aMa, a^{\dagger}NMa \rbrack=4 a
\left(MNM+\left(MNM\right)^T \right)a\  ,
\end{align*}
and by  (\ref{d}) all higher order commutators in the exponent are
zero. Thus
\begin{align*}
e^{aMa}e^{ a^{\dagger}N a^{\dagger}}=& e^{ 2 \,\hbox{Tr}(M N) } e^{
a^{\dagger}N a^{\dagger}+a^{\dagger} 4N M a+a \,
2\left(MNM+\left(MNM\right)^T \right) a  } e^{aMa}\  ,
\end{align*}
which, using Lemma \ref{sl} to split the middle exponential,
yields (\ref{1}). Similarly, to prove (\ref{2}) we notice that, by
\ref{le2},
\begin{align*}
e^{aMa}e^{ a^{\dagger}N a}=&\left(e^{aMa}e^{ a^{\dagger}N a}e^{-aMa}\right) e^{aMa}\\
=&e^{a^{\dagger}N a+\lbrack aMa, a^{\dagger}N a\rbrack+
\frac{1}{2!}\lbrack aMa,\lbrack aMa, a^{\dagger}N a\rbrack
\rbrack+\cdots} e^{aMa} \  .
\end{align*}
By Lemma  \ref{le2},
\[
\lbrack aMa, a^{\dagger}N a\rbrack=
a\left(MN+\left(NM\right)^T\right)a\  ,
\]
so $ \lbrack aMa,\lbrack aMa, a^{\dagger}N a\rbrack \rbrack$ and
all  higher order commutators in the exponent are all equal to
zero. Thus
\begin{align*}
e^{aMa}e^{ a^{\dagger}N a}=&  e^{ a^{\dagger}N a+a
\left(MN+\left(NM\right)^T\right) a  } e^{aMa}
\end{align*}
which, using Lemma \ref{sl} to split the exponential, yields
\begin{align*}
e^{aMa}e^{ a^{\dagger}N a}=&e^{\hbox{Tr}\left(-\frac{1}{2}N
+\frac{1}{2}g\left(0,N,M N+(MN)^T\right)\right)} \cdot
e^{\frac{1}{2}a^{\dagger}\hat f\left(0,N,M
N+(MN)^T\right)a^{\dagger}}\\
&\cdot e^{ a^{\dagger}g\left(0,N,M N+(M N)^T)\right) a}\notag
\cdot e^{a\left(M+\frac{1}{2}\hat h\left(0,N,M N+(M
N)^T\right)\right)a}\  ,
\end{align*}
from which (\ref{2}) follows with the use of Remark \ref{r}.
Finally, to prove (\ref{3}) we notice that by \ref{le2},
\begin{align*}
e^{a^{\dagger}Ma}e^{a^{\dagger} N a^{\dagger}}=&e^{a^{\dagger} N a^{\dagger}}\left(e^{-a^{\dagger} N a^{\dagger}}e^{ a^{\dagger}M a}e^{a^{\dagger} N a^{\dagger}}\right) \\
=&e^{a^{\dagger} N a^{\dagger}}e^{a^{\dagger}M a+\lbrack
-a^{\dagger} N a^{\dagger}, a^{\dagger}Ma\rbrack+
\frac{1}{2!}\lbrack -a^{\dagger} N a^{\dagger}, \lbrack
-a^{\dagger} N a^{\dagger}, a^{\dagger}M a\rbrack
\rbrack+\cdots}\\
=&e^{a^{\dagger} N a^{\dagger}}e^{a^{\dagger}M a+\lbrack
 a^{\dagger}Ma, a^{\dagger} N a^{\dagger} \rbrack+
\frac{1}{2!}\lbrack a^{\dagger} N a^{\dagger}, \lbrack a^{\dagger}
N a^{\dagger}, a^{\dagger}M a\rbrack \rbrack+\cdots} \  .
\end{align*}
By Lemma  \ref{le2},
\[
\lbrack a^{\dagger} Ma, a^{\dagger}N a^{\dagger}\rbrack=
a^{\dagger}\left(MN+\left(NM\right)^T\right)a^{\dagger}\  ,
\]
so $ \lbrack a^{\dagger} N a^{\dagger},\lbrack a^{\dagger} N
a^{\dagger}, a^{\dagger}M a\rbrack \rbrack$ and all higher order
commutators in the exponent are all equal to zero. Thus
\begin{align*}
e^{a^{\dagger}Ma}e^{a^{\dagger} N a^{\dagger}}=&e^{ a^{\dagger}N
a^{\dagger}} e^{ a^{\dagger} \left(MN+\left(NM\right)^T\right)
a^{\dagger} +a^{\dagger}Ma}\  ,
\end{align*}
which, using Lemma \ref{sl} to split the exponential, yields
\begin{align*}
e^{a^{\dagger}Ma}e^{a^{\dagger} N a^{\dagger}} =&e^{a^{\dagger}
Na^{\dagger}} e^{\hbox{Tr}\left(-\frac{1}{2}M+\frac{1}{2}g\left(M
N+(MN)^T, M, 0\right)\right)} \cdot e^{\frac{1}{2}a^{\dagger}\hat
f\left(M N+(MN)^T, M, 0\right)a^{\dagger}}\\
&\cdot e^{ a^{\dagger}g\left(M N+(MN)^T, M, 0\right) a} \cdot
e^{\frac{1}{2}a \hat h\left(M N+(MN)^T, M, 0\right)a}\  ,
\end{align*}
from which (\ref{2}) follows with the use of Remark \ref{r}.
\end{proof}
\begin{lemma}\label{combexp}
For  $M, N\in M_{n}(\mathbb{C})$,
\begin{align}
e^{ a^{\dagger}M a^{\dagger}}e^{ a^{\dagger}N a^{\dagger}}=&e^{ a^{\dagger}\left(M+N\right) a^{\dagger}}\  ,  \label{111} \\
&\notag\\
e^{aMa}e^{ aN a}=&e^{a\left(M+N\right)a}\  ,\label{222}\\
  &\notag\\
e^{a^{\dagger}Ma}e^{ a^{\dagger}N
a}=&e^{a^{\dagger}\left(M+N+BCH(M,N)\right)a}\  ,\label{333}
\end{align}
or in the $B$-notation
\begin{align*}
e^{B^2_0(M)}e^{ B^2_0(N)}=&e^{B^2_0(M+N)}\  , \\
 &\\
e^{B^0_2(M)}e^{ B^0_2(N)}=& e^{B^0_2(M+N)} \  ,\\
 &\\
e^{B^1_1(M)}e^{ B^1_1(N)}=& e^{B^1_1(M+N+BCH(M,N))}\  ,
\end{align*}
 where $BCH(M,N)$ is as in Remark \ref{BCH}.
\end{lemma}
\begin{proof} For (\ref{333}), by the Baker-Campbell-Hausdorff  formula of
Remark \ref{BCH}, we have
\begin{align*}
e^{a^{\dagger}Ma}e^{ a^{\dagger}N
a}=&e^{a^{\dagger}Ma+a^{\dagger}N a+BCH( a^{\dagger}Ma
,a^{\dagger}N a )} \   .
\end{align*}
Using  (\ref{c})  we see that
\[
BCH( a^{\dagger}Ma ,a^{\dagger}N a )=a^{\dagger} BCH(M,N) a \   .
\]
Thus,
\[
 e^{a^{\dagger}Ma}e^{a^{\dagger}Na}=e^{a^{\dagger}\left(M+N+BCH(M,N)\right)a} \   .
 \]
 The proof of (\ref{111}) and (\ref{222}) is similar, with the $BCH$
 term equal to $0$ in both cases by (\ref{d}).
\end{proof}

\section{The Group Law in Coordinates of the Second Kind}\label{grp-law-coord-2dkd}

The homogeneous quadratic Weyl operators are group elements of
$heis_{\mathbb{C}}(2;n)$ in coordinates of the second kind. In
order to calculate their vacuum expectation values, i.e. the
characteristic function of the vacuum distribution of the
corresponding hermitian elements of $heis_{\mathbb{C}}(2;n)$, we
have to find the transition formula, from coordinates of the
second kind to coordinates of the first kind (which, in physical
language corresponds to find the \textbf{normally ordered form} of
these expressions). This is done in this and the next section.
\begin{theorem}\label{gl2}  In the notation of Section \ref{intro}, let $x_i
\in \mathbb{C}$, $A_i, C_i \in M_{n,sym}(\mathbb{C})$ and $B_i\in M_{n}(\mathbb{C})$, $i=1,2$. Then
\begin{equation}\label{gl}
G(x_1,A_1,B_1,C_1)G(x_2,A_2,B_2,C_2)=G(x, A,B,  C)\  ,
\end{equation}
 where the coordinates $x, A, B, C$ are given by,
\begin{align*}
x=& x_1+x_2+\frac{1}{2}\hbox{Tr}\left(-\left(B_1+B_2\right)+ g\left(A_2,4A_2C_1,2(C_1A_2C_1+(C_1A_2C_1)^T)\right) \right.\\
&\left.  +g\left(B_1 X+\left(B_1 X\right)^T,B_1,0
\right)  +g\left(0,B_2,Z  B_2+(Z B_2)^T\right) \right)\  ,\notag \\
& \notag \\
A=&X +A_1+\frac{1}{2}\hat f\left(B_1 X+\left(B_1 X\right)^T,B_1,0
\right) \  ,  \\
& \notag \\
B=&Y + g\left(B_1 X+\left(B_1 X\right)^T,B_1,0 \right)
+ g\left(0,B_2,Z  B_2+\left(Z  B_2\right)^T\right)\notag \\
&+BCH\left(g\left(B_1 X+\left(B_1 X\right)^T,B_1,0 \right) ,
Y\right)\notag \\
&+BCH \left( E  , g\left(0,B_2,Z  B_2+\left(Z B_2\right)^T\right) \right)  \  , \notag \\
& \notag\\
C=&C_2+Z +\frac{1}{2}\hat h\left(0,B_2,Z  B_2+(Z  B_2)^T\right)\notag\  , \\
& \notag \\
E=& Y+ g\left(B_1 X+\left(B_1 X\right)^T,B_1,0 \right)
 +BCH  \left(  g\left(B_1 X+\left(B_1 X\right)^T,B_1,0 \right), Y
\right)\notag \  ,
\end{align*}
and
\begin{align}
X=&\frac{1}{2}\hat f\left(A_2,4A_2C_1,2\left(C_1A_2C_1+\left(C_1A_2C_1\right)^T\right)\right) \  ,\label{X}\\
& \notag \\
Y=&g\left(A_2,4A_2C_1,2\left(C_1A_2C_1+\left(C_1A_2C_1\right)^T\right)\right) \  ,\label{Y}\\
& \notag \\
Z=&\frac{1}{2}\hat h\left(A_2,4A_2C_1,2\left(C_1A_2C_1
+\left(C_1A_2C_1\right)^T\right)\right)\label{Z} \  .
\end{align}
\end{theorem}
\begin{proof}
Computing $e^{ a C_1 a} e^{ a^{\dagger}A_2 a^{\dagger}}$ with the
use of Lemma \ref{exp}, we obtain
\begin{align*}
G(x_1,A_1,B_1,C_1)&G(x_2,A_2,B_2,C_2)\\
=&e^{x_1+x_2}e^{ a^{\dagger}A_1 a^{\dagger}}e^{ a^{\dagger}B_1
a}e^{ a C_1 a} e^{ a^{\dagger}A_2 a^{\dagger}}e^{ a^{\dagger}B_2
a}e^{ a C_2 a}\\
 =&e^{x_1+x_2}e^{ a^{\dagger}A_1
a^{\dagger}}e^{a^{\dagger}B_1 a}\cdot
e^{\hbox{Tr}\left(\frac{1}{2}
g\left(A_2,4A_2C_1,2\left(C_1A_2C_1+\left(C_1A_2C_1\right)^T
\right)\right)\right)}\notag\\
&\cdot e^{\frac{1}{2}
a^{\dagger}\hat{f\left(A_2,4A_2C_1,2\left(C_1A_2C_1+\left(C_1A_2C_1\right)^T
\right)\right)a^{\dagger}}}\notag\\
& \cdot
e^{a^{\dagger}g\left(A_2,4A_2C_1,2\left(C_1A_2C_1+\left(C_1A_2C_1\right)^T
\right)\right) a}\notag\\
&\cdot e^{\frac{1}{2}a
\hat{h}\left(A_2,4A_2C_1,2\left(C_1A_2C_1+\left(C_1A_2C_1\right)^T\right)\right)a}
e^{a^{\dagger}B_2 a}e^{ a C_2 a}\notag\\
=&e^{x_1+x_2 + \hbox{Tr}\left(\frac{1}{2}
g\left(A_2,4A_2C_1,2\left(C_1A_2C_1+\left(C_1A_2C_1\right)^T\right)\right)\right)}\notag\\
&\cdot e^{a^{\dagger}A_1 a^{\dagger}}e^{a^{\dagger}B_1 a}
e^{a^{\dagger}Xa^{\dagger}} e^{a^{\dagger} Ya}e^{aZa} e^{
a^{\dagger}B_2 a}e^{ a C_2 a}\notag \  ,
\end{align*}
where $X, Y, Z$ are as in (\ref{X})-(\ref{Z}). Computing
$e^{a^{\dagger}B_1 a} e^{a^{\dagger} Xa^{\dagger}}$ and $e^{a  Za}
e^{ a^{\dagger}B_2 a}$ with the use of Lemma \ref{exp}, we obtain
\begin{align*}
G(x_1,&A_1,B_1,C_1)G(x_2,A_2,B_2,C_2)  \\
 =&e^{x_1+x_2+\hbox{ Tr}\left(\frac{1}{2}g\left(A_2,4A_2C_1,2\left(C_1A_2C_1+\left(C_1A_2C_1\right)^T\right)\right)\right)}
 e^{a^{\dagger}A_1 a^{\dagger}}\\
&  \cdot e^{  \hbox{Tr}  \left(-\frac{1}{2} B_1 + \frac{1}{2}g
\left(B_1 X+\left(B_1 X\right)^T,
B_1,0\right)\right)}e^{a^{\dagger}\left(\frac{1}{2}\hat f\left(B_1
X+\left(B_1 X\right)^T,B_1,0 \right)+X\right)a^{\dagger}}\\
& \cdot e^{a^{\dagger}g\left(B_1X+\left(B_1 X\right)^T,
B_1,0\right) a} e^{a^{\dagger}Y a}e^{\hbox{Tr}\left(-\frac{1}{2}B_2
+\frac{1}{2}g\left(0,B_2,Z  B_2+(Z B_2)^T\right)\right)}\\
 & \cdot e^{ a^{\dagger}g\left(0,B_2,Z  B_2+(Z  B_2)^T\right) a}
 e^{a\left(Z +\frac{1}{2}\hat h\left(0,B_2,Z  B_2+(Z
 B_2)^T\right)\right)a}e^{ a C_2 a}\\
=&e^{x_1+x_2+\hbox{Tr}\left(\frac{1}{2}g\left(A_2,4A_2C_1,2\left(C_1A_2C_1+\left(C_1A_2C_1\right)^T\right)\right)\right)}
\\
& \cdot e^{  \hbox{Tr}  \left(-\frac{1}{2} B_1 + \frac{1}{2}g
\left(B_1 X+\left(B_1 X\right)^T, B_1,0\right)\right)} e^{\hbox{Tr}\left(-\frac{1}{2}B_2
+\frac{1}{2}g\left(0,B_2,Z  B_2+(Z B_2)^T\right)\right)}\\
&  \cdot e^{a^{\dagger}A_1
a^{\dagger}}e^{a^{\dagger}\left(\frac{1}{2}\hat f\left(B_1
X+\left(B_1 X\right)^T,B_1,0 \right)+X\right)a^{\dagger}}\\
& \cdot e^{a^{\dagger}g\left(B_1X+\left(B_1 X\right)^T,
B_1,0\right) a} e^{a^{\dagger}Y a}e^{ a^{\dagger}g\left(0,B_2,Z  B_2+(Z  B_2)^T\right) a}\\
 & \cdot
 e^{a\left(Z +\frac{1}{2}\hat h\left(0,B_2,Z  B_2+(Z
 B_2)^T\right)\right)a}e^{ a C_2 a} \  ,
\end{align*}
from which (\ref{gl}) follows by combining the exponentials with
the use of Lemma \ref{combexp}.
\end{proof}

\section{The Group Law in Coordinates of the First Kind}
\label{grp-law-coord-1st-kd}

\begin{theorem}\label{gl1}
In the notation of Section \ref{intro}, let $x_i \in \mathbb{C}$,
$A_i, C_i \in M_{n,sym}(\mathbb{C})$ and $B_i\in M_{n}(\mathbb{C})$, $i=1,2$. Then
\[
W(x_1,A_1,B_1,C_1)W(x_2,A_2,B_2,C_2)=W(x, A,B, C) \  ,
\]
where $x, A, B, C$ are determined by the system
\begin{align*}
x'=&x-\frac{1}{2}\,\hbox{Tr}(B) +\frac{1}{2}\hbox{Tr}\left(g(A,B,C)\right)\  ,\\
A'=&\frac{1}{2}\hat f(A,B,C)\  ,\\
 B'=& g(A,B,C) \  ,\\
  C'=&\frac{1}{2} \hat h(A,B,C) \  ,
\end{align*}
 where
\begin{align*}
x'=& x_1+x_2+\frac{1}{2}\hbox{Tr}\left(-\left(L_1+L_2\right)+ g\left(K_2,4K_2M_1,2(M_1K_2M_1+(M_1K_2M_1)^T)\right) \right.\\
&\left.  +g\left(L_1 X+\left(L_1 X\right)^T,L_1,0 \right)
+g\left(0,L_2,Z  L_2+(Z L_2)^T\right) \right)\\
&+\frac{1}{2}\hbox{Tr}\left(-(B_1+B_2)
+g(A_1,B_1,C_1)+g(A_2,B_2,C_2)\right)  \  ,\\
& \notag \\
A'=&X +K_1+\frac{1}{2}\hat f\left(L_1 X+\left(L_1 X\right)^T,L_1,0
\right)   \  , \\
& \notag \\
B'=&Y + g\left(L_1 X+\left(L_1 X\right)^T,L_1,0 \right)
+ g\left(0,L_2,Z  L_2+(Z  L_2)^T\right)  \\
&+BCH\left(g\left(L_1 X+\left(L_1 X\right)^T,L_1,0 \right) ,
Y\right)\notag \\
&+BCH \left( E  , g\left(0,L_2,Z  L_2+\left(Z L_2\right)^T\right) \right) \  ,  \notag \\
& \notag \\
C'=&M_2+Z +\frac{1}{2}\hat h\left(0,L_2,Z  L_2+(Z
 L_2)^T\right) \  ,
\end{align*}
and
\begin{align*}
E=& Y+ g\left(L_1 X+\left(L_1 X\right)^T,L_1,0 \right) +BCH
\left(  g\left(L_1 X+\left(L_1 X\right)^T,L_1,0 \right), Y
\right) \  ,\\
&\\
X=&\frac{1}{2}\hat f\left(K_2,4K_2M_1,2\left(M_1K_2M_1+\left(M_1K_2M_1\right)^T\right)\right) \  ,\\
&\\
Y=&g\left(K_2,4K_2M_1,2\left(M_1K_2M_1+\left(M_1K_2M_1\right)^T\right)\right) \  ,\\
& \\
Z=&\frac{1}{2}\hat h\left(K_2,4K_2M_1,2\left(M_1K_2M_1
+\left(M_1K_2M_1\right)^T\right)\right) \  ,
\end{align*}
where, for $i=1,2$,
$$
K_i=\frac{1}{2}\hat f(A_i,B_i,C_i)   \,\,;\,\, L_i= g(A_i,B_i,C_i)
\,\,;\,\, M_i=\frac{1}{2} \hat h(A_i,B_i,C_i) \  .
$$
\end{theorem}
\begin{proof} We have
\begin{align*}
W(x_1,A_1,B_1,C_1)&W(x_2,A_2,B_2,C_2)&\\
&=e^{x_1\mathbf{1}+ a^{\dagger}A_1 a^{\dagger}+ a^{\dagger}B_1 a+
a C_1 a}e^{x_2\mathbf{1}+ a^{\dagger}A_2 a^{\dagger}+
a^{\dagger}B_2
a+ a C_2 a}\\
&=e^{x_1 +x_2}e^{ a^{\dagger}A_1 a^{\dagger}+ a^{\dagger}B_1 a+ a
C_1 a}e^{ a^{\dagger}A_2 a^{\dagger}+ a^{\dagger}B_2 a+ a C_2 a} \
.
\end{align*}
Splitting the exponentials with the use of  Lemma \ref{sl}, we
obtain
\begin{align*}
W(&x_1,A_1,B_1,C_1)W(x_2,A_2,B_2,C_2)\\
=&e^{x_1 +x_2+\hbox{Tr}\left(-\frac{1}{2}(B_1+B_2)
+\frac{1}{2}g(A_1,B_1,C_1)+\frac{1}{2}g(A_2,B_2,C_2)\right)}\\
&\cdot e^{\frac{1}{2}a^{\dagger}\hat f(A_1,B_1,C_1)a^{\dagger}}e^{
a^{\dagger}g(A_1,B_1,C_1) a}e^{\frac{1}{2}a \hat h(A_1,B_1,C_1)a}\\
&\cdot e^{\frac{1}{2}a^{\dagger}\hat f(A_2,B_2,C_2)a^{\dagger}}e^{
a^{\dagger}g(A_2,B_2,C_2) a}e^{\frac{1}{2}a \hat h(A_2,B_2,C_2) a}\\
=&e^{\frac{1}{2}\hbox{Tr}\left(-(B_1+B_2)
+g(A_1,B_1,C_1)+g(A_2,B_2,C_2)\right)}G(x_1,
K_1,L_1,M_1)G(x_2,K_2,L_2,M_2) \  ,
\end{align*}
$$
K_i= \frac{1}{2}\hat f(A_i,B_i,C_i)    \,\,;\,\, L_i=
g(A_i,B_i,C_i)   \,\,;\,\, M_i=\frac{1}{2} \hat
h(A_i,B_i,C_i)\,,\, i=1,2 \ .
$$
Thus, by Theorem \ref{gl2},
\[
W(x_1,A_1,B_1,C_1)W(x_2,A_2,B_2,C_2)=G(x', A',B',  C')
\]
where
\begin{align*}
x'=& x_1+x_2+\frac{1}{2}\hbox{Tr}\left(-\left(L_1+L_2\right)+ g\left(K_2,4K_2M_1,2(M_1K_2M_1+(M_1K_2M_1)^T)\right) \right.\\
&\left.  +g\left(L_1 X +\left(L_1 X \right)^T,L_1,0 \right)
+g\left(0,L_2,Z  L_2+(Z L_2)^T\right) \right)\\
&+\frac{1}{2}\hbox{Tr}\left(-(B_1+B_2)
+g(A_1,B_1,C_1)+g(A_2,B_2,C_2)\right) \  , \\
& \notag \\
A'=&X  +K_1+\frac{1}{2}\hat f\left(L_1 X +\left(L_1 X
\right)^T,L_1,0
\right)   \  , \\
& \notag \\
B'=&Y + g\left(L_1 X +\left(L_1 X \right)^T,L_1,0 \right)
+ g\left(0,L_2,Z  L_2+(Z  L_2)^T\right)  \\
&+BCH\left(g\left(L_1 X +\left(L_1 X \right)^T,L_1,0 \right) ,
Y\right)\notag \\
&+BCH \left( E  , g\left(0,L_2,Z  L_2+\left(Z L_2\right)^T\right) \right) \  ,  \notag \\
& \notag \\
C'=&M_2+Z +\frac{1}{2}\hat h\left(0,L_2,Z  L_2+(Z
 L_2)^T\right) \  ,
\end{align*}
and
\begin{align*}
E=& Y+ g\left(L_1 X +\left(L_1 X \right)^T,L_1,0 \right) +BCH
\left(  g\left(L_1 X +\left(L_1 X \right)^T,L_1,0 \right), Y
\right) \  ,\\
X=&\frac{1}{2}\hat f\left(K_2,4K_2M_1,2\left(M_1K_2M_1+\left(M_1K_2M_1\right)^T\right)\right) \  ,\\
&  \\
Y=&g\left(K_2,4K_2M_1,2\left(M_1K_2M_1+\left(M_1K_2M_1\right)^T\right)\right) \  ,\\
& \\
Z =&\frac{1}{2}\hat h\left(K_2,4K_2M_1,2\left(M_1K_2M_1
+\left(M_1K_2M_1\right)^T\right)\right) \   .
\end{align*}
Therefore, in Notation \ref{n3},
\[
W(x_1,A_1,B_1,C_1)W(x_2,A_2,B_2,C_2)=W(x,  A, B,  C) \  ,
\]
where $x, A, B, C$ are defined by
\begin{align*}
x'=&x-\frac{1}{2}\,\hbox{Tr}(B) +\frac{1}{2}\hbox{Tr}\left(g(A,B,C)\right)\\
A'=&\frac{1}{2}\hat f(A,B,C)    \,\,;\,\, B'=  g(A,B,C) \,\,;\,\,
C'=\frac{1}{2} \hat h(A,B,C) \  .
\end{align*}
\end{proof}
\begin{example} \rm  For  $j=1, 2$, assuming
 $x_j \mathbf{1}=A_j=C_j=0$, by Lemma \ref{combexp} we know
that
\begin{align*}
W(0, 0, B_1, 0)W(0, 0, B_2, 0)=&e^{a^{\dagger} B_1 a}e^{
a^{\dagger}  B_2 a}=e^{a^{\dagger}\left(B_1+B_2+BCH( B_1,
B_2)\right)a}\\
=&W(0, 0, B_1+B_2+BCH( B_1, B_2), 0) \  .
\end{align*}
To verify this using the  group law, we notice that in the
notation of Theorem \ref{gl1} we have
\[
K_j=\frac{1}{2}\hat f\left( 0, B_j, 0\right)\,,\,L_j=g\left( 0,
B_j, 0\right)\,,\,M_j=\frac{1}{2}\hat h\left( 0, B_j, 0\right) \
.
\]
In the notation of Lemma \ref{sl},
\[
v=\left(
\begin{array}{ll}
B_j & 0 \\
0& -  B_j^T
\end{array}
\right) \,\,,\,\,e^v=\left(
\begin{array}{ll}
e^{B_j} & 0 \\
0& e^{-B_j^T}
\end{array}
\right) \  ,
\]
i.e.,
\[
P=e^{B_j}\,,\,Q=R=0\,,\,S=e^{-B_j^T} \  ,
\]
so
\[
\hat f\left( 0, B_j, 0\right)=\hat h\left( 0, B_j, 0\right) =0 \
,
\]
and
\[
g\left( 0, B_j, 0\right)=-\log \left(e^{ -  B_j^T} \right)^T= B_j
\  .
\]
Thus
\[
K_j=M_j=0\,,\,L_j= B_j \  ,
\]
and, by Remark \ref{r},
\[
X=\frac{1}{2}\hat f\left( 0, 0, 0\right)=0\,,\,Y=g(0, 0,
0)=0\,,\,Z=\frac{1}{2}\hat h\left( 0, 0, 0\right)=0 \  ,
\]
which imply that $A'=C'=0$,
\begin{align*}
B' =& g\left( 0, B_1, 0\right)+g\left( 0, B_2, 0\right) +BCH\left(
g\left( 0, B_1, 0\right),0\right)\\
& +BCH\left(g\left( 0, B_1,
0\right), g\left( 0, B_2, 0\right)\right)\\
=& B_1+B_2+0+BCH\left(B_1, B_2\right) = B_1+B_2+BCH\left(B_1,
B_2\right) \  ,
\end{align*}
and
\begin{align*}
x'=& \frac{1}{2}\hbox{Tr}\left(-\left( B_1+ B_2\right) +g\left(0,
B_1,0 \right) +g\left(0,B_2, 0\right) \right) \\
&+\frac{1}{2}\hbox{Tr}\left(-(B_1+B_2) +g(0,B_1,0)+g(0,B_2,0)\right)\\
=&\hbox{Tr}\left(-\left( B_1+ B_2\right)+ B_1+ B_2\right) =0 \  .
\end{align*}
Therefore $x, A, B, C$ are determined by the system
\begin{align*}
0=&x-\frac{1}{2}\,\hbox{Tr}(B) +\frac{1}{2}\hbox{Tr}\left(g(A,B,C)\right) \\
0=&\frac{1}{2}\hat f(A,B,C)= \frac{1}{2}QS^{-1}  \\
B_1+B_2+BCH\left(B_1, B_2\right)=&  g(A,B,C)=-\log S^T   \\
0=&\frac{1}{2} \hat h(A,B,C)=\frac{1}{2}S^{-1} R \  .
\end{align*}
We see that
\[
 Q=R=0\,,\,S=e^{-\left(B_1+B_2+BCH\left(B_1,
 B_2\right)\right)^T} \  ,
 \]
where these new $P, Q, R, S$ are again as in  Lemma \ref{sl}.
Thus, in the notation of Theorem \ref{gl1},
 \[
e^v=\left(
\begin{array}{ll}
P & Q\\
-R& S
\end{array}\right)=\left(
\begin{array}{ll}
P & 0\\
0& e^{-\left(B_1+B_2+BCH\left(B_1,
 B_2\right)\right)^T}
 \end{array}\right)=e^{ \left(
\begin{array}{ll}
B &2 A \\
-2 C& - B^T
\end{array}
\right) } \  ,
 \]
which implies
 \[
A=C=0\,,\,B=B_1+B_2+BCH\left(B_1, B_2\right) \  ,
 \]
 and the equation
\[
0=x-\frac{1}{2}\,\hbox{Tr}(B) +\frac{1}{2}\hbox{Tr}\left(g(A, B, C)
\right) \  ,
\]
becomes
$$
0=x-\frac{1}{2}\,\hbox{Tr}\left(B_1+B_2+BCH\left(B_1, B_2\right)  \right)
 +\frac{1}{2}\hbox{Tr}\left(g(0, B_1+B_2+BCH\left(B_1, B_2\right) ,0)\right)
$$
i.e.,
$$
0=x-\frac{1}{2}\,\hbox{Tr}\left(B_1+B_2+BCH\left(B_1,B_2\right)
\right) +\frac{1}{2}\hbox{Tr}\left(B_1+B_2+ BCH\left(B_1,
B_2\right) \right) \  ,
$$
so
\[
x=0 \  .
\]
Therefore
\begin{align*}
W(0, 0, B_1, 0)W(0, 0, B_2, 0)&=W(x, A, B, C)\\
&=W(0, 0, B_1+B_2+BCH\left(B_1, B_2\right), 0) \  .
\end{align*}
\end{example}

\section{Quadratic Weyl Operators and Vacuum Characteristic Function}
\label{Weyl-ops-Vac-Char-Fctn}

In this section we consider the Fock representation of
$heis_{\mathbb{C}}(2;n)$ and we apply all the tools developed in
the first part of the paper to calculate the vacuum characteristic
functions of its hermitian elements, considered as real valued
classical random variables, as well as the explicit form of the
scalar product of the vacuum cyclic space of
$heis_{\mathbb{C}}(2;n)$. Let $\lambda\in\mathbb{R}$, $A$ a
symmetric matrix, and $B$ a hermitian matrix. Then
$$
H(\lambda, A, B) =\lambda\mathbf{1}+ a^{\dagger}A a^{\dagger}+
a^{\dagger}B a+ a \bar{A} a =\lambda\mathbf{1}+ B^2_0(A )+
B^1_1(B)+ B^0_2(\bar{A}) \  ,
$$
is a hermitian operator (in Section
\eqref{Exp-quadr-alg-in-Fck-repr} we will prove that it is
self-adjoint). For  $A, B$ as above,  we define the unitary
\textit{Weyl operator}
$$
U(\lambda, A, B) =e^{i H(\lambda, A, B)} =
e^{i\left(\lambda\mathbf{1}+ a^{\dagger}A a^{\dagger}+
a^{\dagger}B a + a \bar{A}a\right)} =e^{i\left(\lambda\mathbf{1}+
B^2_0(A )+ B^1_1(B)+ B^0_2(\bar{A})\right)} \  .
$$
Notice that, in the notation of section \ref{intro},
\[
U(\lambda, A, B)=W(i\lambda,i A, i B, i \bar{A}) \  .
\]
 The Weyl group multiplication law is
\begin{align*}
U (\lambda_1, A_1, B_1)U (\lambda_2, A_2, B_2) =&W(i \lambda_1,i
A_1, i  B_1, i  \bar{A_1}) W(i \lambda_2,i  A_2, i  B_2,
i\bar{A_2})\\
=&W(i \lambda,i   A, i  B, i  \hat{\bar{A}}) =U (\lambda,  A, B) \
,
\end{align*}
where $\lambda, A, B$ are determined by Theorem  \ref{gl1}.
\begin{theorem}
If $\Phi$ is a normalized \textit{Fock vacuum vector}, i.e. $a\Phi=0$ and $\|\Phi\|=1$, then
the \textit{vacuum characteristic function} of the \textit{quantum observable} $H(\lambda, A, B)$
is given by
\[
\langle \Phi, U(\lambda, A, B) \Phi
\rangle=e^{i\lambda-\frac{1}{2}\,\hbox{Tr}(iB) +\frac{1}{2}\hbox{Tr}\left(g(iA,iB, i \bar{A})\right)} \  .
\]
\end{theorem}
\begin{proof} In the context of Notation \ref{n3}, we have
\begin{align*}
\langle \Phi, U(\lambda, A, B) \Phi \rangle =&\langle \Phi,
W(i\lambda,i A, i B, i \bar{A}) \Phi \rangle \\
=&\langle \Phi, G(x',A',B',C') \Phi \rangle\\
=&\langle \Phi,  e^{x'}e^{ a^{\dagger}A'
a^{\dagger}}e^{a^{\dagger}B' a}e^{ a C' a} \Phi \rangle\\
=&e^{x'}\langle \Phi,   \Phi \rangle =e^{x'} \  ,
\end{align*}
 where
\begin{align*}
x'=&i\lambda-\frac{1}{2}\,\hbox{Tr}(iB) +\frac{1}{2}\hbox{Tr}\left(g(iA,iB, i \bar{A})\right)\\
A'=&\frac{1}{2} \hat f(iA,iB, i \bar{A})   \,\,;\,\, B'= g(iA,iB,
i \bar{A})   \,\,;\,\, C'=\frac{1}{2} \hat h(iA,iB, i \bar{A}) \
.
\end{align*}
\end{proof}
\begin{proposition} Let  $A\in M_{n,sym}(\mathbb{C})$ and let $\Phi$ be a
normalized Fock vacuum vector. Then
\[
\|e^{ a^{\dagger}A a^{\dagger}} \Phi \|=e^{ \frac{1}{4}
 \hbox{Tr}\left( g\left( A,4A \bar{ A}, 4\bar{ A} A \bar{ A} \right)\right)
 } \  .
\]
\end{proposition}
\begin{proof} By (\ref{11}) of  Lemma \ref{exp}
\begin{align*}
\|e^{ a^{\dagger}A a^{\dagger}}\Phi \|^2 &=\langle e^{a^{\dagger}A
a^{\dagger}} \Phi, e^{ a^{\dagger}A a^{\dagger}} \Phi\rangle =
\langle \left(e^{ a^{\dagger}A a^{\dagger}}\right)^*
e^{a^{\dagger}A a^{\dagger}} \Phi,  \Phi\rangle
\\
=& e^{\frac{1}{2}\hbox{Tr}\left(g\left(A,4A \bar A,4\bar AA\bar
A\right)\right)}\\
&\cdot \langle e^{\frac{1}{2}a^{\dagger}\hat f\left(A,4A \bar
A,4\bar AA \bar A \right)a^{\dagger}}e^{a^{\dagger}g\left(A,4A
\bar A,   4 \bar A A \bar A \right)a} e^{\frac{1}{2}a \hat
h\left(A,4A \bar A, 4 \bar A A
\bar A\right)a}\Phi, \Phi \rangle\\
=& e^{\frac{1}{2}\hbox{Tr}\left(g\left(A,4A \bar A,   4 \bar A A
\bar A \right)\right)} \langle  e^{\frac{1}{2}a^{\dagger}\hat
f\left(A,4A \bar A,   4 \bar A A \bar A \right)a^{\dagger}} \Phi,
\Phi \rangle\\
=& e^{\frac{1}{2}\hbox{Tr}\left(g\left(A,4A \bar A,   4 \bar A A
\bar A \right)\right)} \langle\Phi, e^{\frac{1}{2}a \left(\hat
f\left(A,4A \bar A,   4 \bar A A \bar A \right)\right)^*a}\Phi
\rangle\\
 =&e^{\frac{1}{2}\hbox{Tr}\left(g\left(A,4A \bar A,   4
\bar A A \bar A \right)\right)}\langle   \Phi, \Phi \rangle\\
=&e^{\frac{1}{2}\hbox{Tr}\left(g\left(A,4A \bar A,   4 \bar A A\bar
A \right)\right)} \  .
\end{align*}
\end{proof}
\begin{proposition} Let  $A, B\in M_{n,sym}(\mathbb{C})$ and let $\Phi$ be a normalized Fock
vacuum vector. Then
\[
\langle e^{ a^{\dagger}A a^{\dagger}} \Phi,  e^{ a^{\dagger}B
a^{\dagger}} \Phi\rangle=e^{ \frac{1}{2}
 \hbox{Tr}\left( g\left( A,4A \bar{ B}, 4\bar{ B} A \bar{ B} \right)\right)
 } \  .
\]
\end{proposition}
\begin{proof} By (\ref{11}) of  Lemma \ref{exp}
\begin{align*}
&\langle e^{ a^{\dagger}A a^{\dagger}}\Phi,
e^{a^{\dagger}Ba^{\dagger}}\Phi\rangle =\langle \left(e^{
a^{\dagger}Ba^{\dagger}}\right)^* e^{a^{\dagger}A a^{\dagger}}
\Phi,  \Phi\rangle =\langle e^{ a \bar B a} e^{a^{\dagger}A
a^{\dagger}} \Phi,  \Phi\rangle\\
=& e^{\frac{1}{2}\hbox{Tr}\left(g\left(A,4A \bar B,   4 \bar B A
\bar B \right)\right)}\\
& \cdot \langle e^{\frac{1}{2}a^{\dagger}\hat f\left(A,4A \bar B,
4 \bar B A \bar B \right)a^{\dagger}} e^{a^{\dagger}g\left(A,4A
\bar B,   4 \bar B A \bar B \right)a} e^{\frac{1}{2}a \hat
h\left(A,4A \bar B,4\bar BA\bar
B\right)a}\Phi,\Phi \rangle\\
= &e^{\frac{1}{2}\hbox{Tr}\left(g\left(A,4A \bar B,4\bar BA\bar
B\right)\right)} \langle  e^{\frac{1}{2}a^{\dagger}\hat
f\left(A,4A \bar B,   4 \bar B A \bar B \right)a^{\dagger}}  \Phi,
\Phi \rangle\\
=& e^{\frac{1}{2}\hbox{Tr}\left(g\left(A,4A \bar B,   4 \bar B A
\bar B \right)\right)} \langle \Phi, e^{\frac{1}{2}a \left(\hat
f\left(A,4A \bar B,   4 \bar B A \bar B \right)\right)^*a}\Phi
\rangle\\
=&e^{\frac{1}{2}\hbox{Tr}\left(g\left(A,4A \bar B,   4 \bar B A
\bar B \right)\right)} \langle   \Phi, \Phi \rangle\\
=&e^{\frac{1}{2}\hbox{Tr}\left(g\left(A,4A \bar B,   4 \bar B A
\bar B \right)\right)} \  .
\end{align*}
\end{proof}
\begin{proposition}\label{proj1}
Let  $E$ be a real projection matrix and let $\Phi$ be a
normalized Fock vacuum vector. Then for $n=1, 2, 3,...$
\[
\| \left(a^{\dagger}E a^{\dagger}\right)^n \Phi\|^2=4^n n!
 \left(\frac{\hbox{Tr} E }{2}\right)^{(n)} \  ,
\]
where for $x\in \mathbb{R}$,
\[
x^{(n)}=x(x +1 )(x +2 )\cdots (x +n-1 ) \  ,
\]
is the rising factorial of $x$, and
\[
(x )_n=x(x -1 )(x -2 )\cdots (x -n+1 ) \  ,
\]
is the falling factorial of $x$.
\end{proposition}
\begin{proof}
Using the Lie algebraic identity
\[
\lbrack   X, Y^n \rbrack=\sum_{j=0}^{n-1}Y^{n-1-j} \lbrack   X, Y
\rbrack Y^j \  ,
\]
and Lemma \ref{le2} we have
\begin{align*}
\lbrack   a^{\dagger}Ea,& \left(a^{\dagger}E a^{\dagger}\right)^n
\rbrack = \sum_{j=0}^{n-1}\left(a^{\dagger}E
a^{\dagger}\right)^{n-1-j} \lbrack   a^{\dagger}Ea, a^{\dagger}E
a^{\dagger}\rbrack \left(a^{\dagger}E a^{\dagger}\right)^j \\
=&\sum_{j=0}^{n-1}\left(a^{\dagger}E a^{\dagger}\right)^{n-1-j}
 2  (a^{\dagger}E
a^{\dagger}) \left(a^{\dagger}E a^{\dagger}\right)^j \\
=&2 (a^{\dagger}E a^{\dagger})^n \sum_{j=0}^{n-1}1 =2 n
(a^{\dagger}E a^{\dagger})^n \  ,
\end{align*}
and
\begin{align*}
\lbrack   aEa, &\left(a^{\dagger}E a^{\dagger}\right)^n \rbrack
\Phi= \sum_{j=0}^{n-1}\left(a^{\dagger}E
a^{\dagger}\right)^{n-1-j} \lbrack   aEa, a^{\dagger}E
a^{\dagger}\rbrack \left(a^{\dagger}E a^{\dagger}\right)^j \Phi\\
=&\sum_{j=0}^{n-1}\left(a^{\dagger}E a^{\dagger}\right)^{n-1-j} (
2 \, \hbox{Tr} E+4 a^{\dagger}E a)\left(a^{\dagger}E
a^{\dagger}\right)^j \Phi\\
=&2n \,\hbox{Tr} E \,\left(a^{\dagger}E
a^{\dagger}\right)^{n-1}\Phi+4 \sum_{j=0}^{n-1}\left(a^{\dagger}E
a^{\dagger}\right)^{n-1-j}( a^{\dagger}E a)\left(a^{\dagger}E
a^{\dagger}\right)^j \Phi\\
=&2n \,\hbox{Tr} E \,\left(a^{\dagger}E
a^{\dagger}\right)^{n-1}\Phi+4 \sum_{j=0}^{n-1}\left(a^{\dagger}E
a^{\dagger}\right)^{n-1-j}(  \lbrack    a^{\dagger}E a,
\left(a^{\dagger}E a^{\dagger}\right)^j \rbrack\\
& + \left(a^{\dagger}E a^{\dagger}\right)^j  a^{\dagger}E a)\Phi\\
=&2n \,\hbox{Tr} E \,\left(a^{\dagger}E
a^{\dagger}\right)^{n-1}\Phi+4 \sum_{j=0}^{n-1}\left(a^{\dagger}E
a^{\dagger}\right)^{n-1-j}  \lbrack    a^{\dagger}E a,
\left(a^{\dagger}E a^{\dagger}\right)^j \rbrack \Phi\\
=&2n \,\hbox{Tr} E \,\left(a^{\dagger}E
a^{\dagger}\right)^{n-1}\Phi+4 \sum_{j=0}^{n-1}\left(a^{\dagger}E
a^{\dagger}\right)^{n-1-j}  2 j \left(a^{\dagger}E
a^{\dagger}\right)^j  \Phi\\
=&2n \,\hbox{Tr} E \,\left(a^{\dagger}E
a^{\dagger}\right)^{n-1}\Phi+4 n(n-1)\left(a^{\dagger}E
a^{\dagger}\right)^{n-1} \Phi \\
=& \left( 2n \,\hbox{Tr} E +4 n(n-1)\right)\left(a^{\dagger}E
a^{\dagger}\right)^{n-1} \Phi    \  .
\end{align*}
Thus
\begin{align*}
m_n:=&\| \left(a^{\dagger}E a^{\dagger}\right)^n \Phi\|^2\\
=&\langle \left(a^{\dagger}E a^{\dagger}\right)^n  \Phi,
 \left(a^{\dagger}E a^{\dagger}\right)^n  \Phi\rangle\\
=&\langle (aEa)\left(a^{\dagger}E a^{\dagger}\right)^n  \Phi,
\left(a^{\dagger}E a^{\dagger}\right)^{n-1}  \Phi\rangle\\
=&\langle( \lbrack   aEa, \left(a^{\dagger}E a^{\dagger}\right)^n
\rbrack+ \left(a^{\dagger}E a^{\dagger}\right)^n aEa ) \Phi,
\left(a^{\dagger}E a^{\dagger}\right)^{n-1}  \Phi\rangle\\
=&\langle \lbrack   aEa, \left(a^{\dagger}E a^{\dagger}\right)^n
\rbrack \Phi,
\left(a^{\dagger}E a^{\dagger}\right)^{n-1}  \Phi\rangle\\
=&\langle \left(2n \,\hbox{Tr} E +4 n(n-1)\right)\left(a^{\dagger}E
a^{\dagger}\right)^{n-1} \Phi ,
\left(a^{\dagger}E a^{\dagger}\right)^{n-1}  \Phi\rangle\\
=&\left(2n \,\hbox{Tr} E +4 n(n-1)\right)\langle \left(a^{\dagger}E
a^{\dagger}\right)^{n-1} \Phi ,
\left(a^{\dagger}E a^{\dagger}\right)^{n-1}  \Phi\rangle\\
=&(  2n \,\hbox{Tr} E +  4 n(n-1) )m_{n-1}\\
&\vdots \\
=&\Pi_{k=1}^n (2k \,\hbox{Tr} E+4 k (k-1))m_0 \\
=&\Pi_{k=1}^n (2k \,\hbox{Tr} E+4 k (k-1))\\
=& 4^n n! \left(\frac{\hbox{Tr} E }{2}\right)^{(n)} \  .
\end{align*}
\end{proof}
\begin{lemma}\label{s} Let  $\Phi$ be a
normalized Fock vacuum vector and let $n\in \mathbb{N}$. Then, for
all $M, N\in Sym \left(\mathbb{R}^{n \times n}\right)$ with
$\lbrack   M, N \rbrack=0$,
\[
\lbrack   a^{\dagger}N a, \left(a^{\dagger}M a^{\dagger}\right)^n
\rbrack =2 n \,a^{\dagger}MN a^{\dagger} \, \left(a^{\dagger}M
a^{\dagger}\right)^{n-1} \  ,
\]
and
\begin{align*}
\lbrack   aNa, \left(a^{\dagger}M a^{\dagger}\right)^n \rbrack
\Phi&=2n \,\hbox{Tr} (NM) \,\left(a^{\dagger}M
a^{\dagger}\right)^{n-1}\Phi\\
&+4n(n-1) \left(a^{\dagger}M^2N a^{\dagger}\right)
\,\left(a^{\dagger}M a^{\dagger}\right)^{n-2}
 \Phi \  .
\end{align*}
\end{lemma}
\begin{proof}
As in the proof of Proposition \ref{proj1},
\begin{align*}
\lbrack   a^{\dagger}N a, \left(a^{\dagger}M a^{\dagger}\right)^n
\rbrack  =& \sum_{j=0}^{n-1}\left(a^{\dagger}M
a^{\dagger}\right)^{n-1-j} \lbrack   a^{\dagger}N a, a^{\dagger}M
a^{\dagger}\rbrack \left(a^{\dagger}M a^{\dagger}\right)^j  \\
=&\sum_{j=0}^{n-1}\left(a^{\dagger}M a^{\dagger}\right)^{n-1-j}
 2  (a^{\dagger}MN
a^{\dagger}) \left(a^{\dagger}M a^{\dagger}\right)^j \\
=&2 \,a^{\dagger}MN a^{\dagger} \, \left(a^{\dagger}M a^{\dagger}\right)^{n-1}\sum_{j=0}^{n-1}1\\
=&2 n \,a^{\dagger}MN a^{\dagger} \, \left(a^{\dagger}M
a^{\dagger}\right)^{n-1} \  ,
\end{align*}
and
\begin{align*}
\lbrack   aNa, \left(a^{\dagger}M a^{\dagger}\right)^n \rbrack
\Phi&=\sum_{j=0}^{n-1}\left(a^{\dagger}M
a^{\dagger}\right)^{n-1-j} \lbrack   aNa, a^{\dagger}M
a^{\dagger}\rbrack \left(a^{\dagger}M a^{\dagger}\right)^j \Phi\\
=&\sum_{j=0}^{n-1}\left(a^{\dagger}M a^{\dagger}\right)^{n-1-j} (
2 \, \hbox{Tr} (NM)\\
&+4 a^{\dagger}NM a)\left(a^{\dagger}M
a^{\dagger}\right)^j \Phi\\
=&2n \,\hbox{Tr} (NM) \,\left(a^{\dagger}M
a^{\dagger}\right)^{n-1}\Phi\\
& +4 \sum_{j=0}^{n-1}\left(a^{\dagger}M
a^{\dagger}\right)^{n-1-j}( a^{\dagger}NM a)\left(a^{\dagger}M
a^{\dagger}\right)^j \Phi\\
=&2n \,\hbox{Tr} (NM) \,\left(a^{\dagger}M
a^{\dagger}\right)^{n-1}\Phi\\
&+4 \sum_{j=0}^{n-1}\left(a^{\dagger}M a^{\dagger}\right)^{n-1-j}
\lbrack    a^{\dagger}NM a,
\left(a^{\dagger}M a^{\dagger}\right)^j \rbrack\Phi\\
=&2n \,\hbox{Tr} (NM) \,\left(a^{\dagger}M
a^{\dagger}\right)^{n-1}\Phi\\
&+4 \sum_{j=0}^{n-1}\left(a^{\dagger}M a^{\dagger}\right)^{n-1-j}
\,2 j\,\left(a^{\dagger}M^2N a^{\dagger}\right) \,
\left(a^{\dagger}M a^{\dagger}\right)^{j-1}  \Phi\\
=&2n \,\hbox{Tr} (NM) \,\left(a^{\dagger}M
a^{\dagger}\right)^{n-1}\Phi\\
&+8 \left(a^{\dagger}M^2N a^{\dagger}\right) \,\left(a^{\dagger}M
a^{\dagger}\right)^{n-2}
\sum_{j=0}^{n-1} j  \Phi\\
=&2n \,\hbox{Tr} (NM) \,\left(a^{\dagger}M
a^{\dagger}\right)^{n-1}\Phi\\
&+4n(n-1) \left(a^{\dagger}M^2N a^{\dagger}\right)
\,\left(a^{\dagger}M a^{\dagger}\right)^{n-2}
 \Phi \  .
\end{align*}
\end{proof}
The vacuum cyclic space of $heis_{\mathbb{C}}(2;n)$ is by definition the sub--space of the Fock space
obtained by applying to the vacuum vectors all the elements of the polynomial algebra generated by
all possible $a^{\dagger} M a^{\dagger}$, $a^{\dagger} H a$, $a N a$. Using the commutation
relations and the Fock property, one verifies that this space is the linear span of vectors of the
form
\begin{equation}\label{gen-el-cycl-sp1}
(a^{\dagger} M_{n} a^{\dagger})\cdots (a^{\dagger} M_{1}
a^{\dagger})\Phi \   .
\end{equation}
Moreover, since each $M_{j}$ has the form $M_{j}=M_{j;R}+iM_{j;I}$ with $M_{j;R}$ and $M_{j;I}$
having real entries, one can suppose that, in \eqref{gen-el-cycl-sp1}, all matrices
have real entries. Using polarization and the commutativity of the creators, one concludes that
the vacuum cyclic space of $heis_{\mathbb{C}}(2;n)$ is the linear span of the vectors
of the form
\begin{equation}\label{gen-el-cycl-sp2}
\left(a^{\dagger}N a^{\dagger}\right)^{n} \Phi \  ,
\end{equation}
when $n$ varies in $\mathbb{N}$ and $N$ varies in $M_{n,sym}(\mathbb{R})$.
Therefore the scalar product on the vacuum cyclic space is uniquely determined by the scalar product
of vectors of the form \eqref{gen-el-cycl-sp2} with $N$ real symmetric.

\begin{proposition}\label{proj2} Let  $\Phi$ be a
normalized Fock vacuum vector and let $l\in\{1,2,..\}$. Then, for
all $M, N\in Sym \left(\mathbb{R}^{l \times l}\right)$ with
$\lbrack M, N \rbrack=0$,
\[
m_n:=\langle \left(a^{\dagger} M a^{\dagger}\right)^n \Phi,\left(
a^{\dagger}N a^{\dagger}\right)^{n} \Phi\rangle \  ,
\]
satisfies, for $n=1,2,...$,  the recursion
\[
m_n=\frac{1}{n}\sum_{k=1}^n 2^{2k-1} \left((n)_k\right)^2 {\rm
Tr}\left((MN)^{k}\right) m_{n-k} \  ,
\]
where
\[
m_0=1 \  ,
\]
and
\[
(n)_k=n(n-1)(n-2)\cdots (n-k+1) \  ,
\]
is the lowering factorial of $n$.
\end{proposition}
\begin{proof}By Lemma \ref{s}
\begin{align*}
m_n:=&\langle \left(a^{\dagger}M a^{\dagger}\right)^n  \Phi,
 \left(a^{\dagger}N a^{\dagger}\right)^n  \Phi\rangle\\
=&\langle (aNa)\left(a^{\dagger}M a^{\dagger}\right)^n  \Phi,
\left(a^{\dagger}N a^{\dagger}\right)^{n-1}  \Phi\rangle\\
=&\langle \lbrack   aNa, \left(a^{\dagger}M a^{\dagger}\right)^n
\rbrack \Phi,\left(a^{\dagger}N a^{\dagger}\right)^{n-1}  \Phi\rangle\\
=&\langle 2n \,{\rm Tr} (MN) \left(a^{\dagger}M
a^{\dagger}\right)^{n-1}\Phi, \left(a^{\dagger}N
a^{\dagger}\right)^{n-1}  \Phi\rangle \\
&+\langle 4 n(n-1)\,a^{\dagger}M^2N a^{\dagger}
\,\left(a^{\dagger}M a^{\dagger}\right)^{n-2} \Phi ,
\left(a^{\dagger}N a^{\dagger}\right)^{n-1}  \Phi\rangle\\
=&2n \,{\rm Tr} (MN) m_{n-1}+4 n(n-1)\langle \left(a^{\dagger}M
a^{\dagger}\right)^{n-2} \Phi , aM^2N a \,\left(a^{\dagger}N a^{\dagger}\right)^{n-1}  \Phi\rangle\\
=&2n \,{\rm Tr} (MN) m_{n-1}+4 n(n-1)\langle \left(a^{\dagger}M
a^{\dagger}\right)^{n-2} \Phi , \lbrack   aM^2N
a,\left(a^{\dagger}N a^{\dagger}\right)^{n-1} \rbrack \Phi\rangle
\  .
\end{align*}
By Lemma \ref{s},
\begin{align*}
\lbrack   aM^2N a, \left(a^{\dagger}N a^{\dagger}\right)^{n-1}
\rbrack \Phi=&2 (n-1){\rm Tr}(M^2N^2)\,\left(a^{\dagger}N
a^{\dagger}\right)^{n-2}\Phi\\
&+ 4(n-2)(n-1)\, a^{\dagger}M^2N^3 a^{\dagger}
\,\left(a^{\dagger}N a^{\dagger}\right)^{n-3}\Phi \  .
\end{align*}
Thus
\begin{align*}
m_n=&2n \,{\rm Tr} (MN) m_{n-1}+4 n(n-1)\langle \left(a^{\dagger}M
a^{\dagger}\right)^{n-2} \Phi ,\\
& 2 (n-1){\rm Tr}(M^2N^2)\,\left(a^{\dagger}N
a^{\dagger}\right)^{n-2}\Phi \\
& + 4(n-2)(n-1)\, a^{\dagger}M^2N^3 a^{\dagger}
\,\left(a^{\dagger}N a^{\dagger}\right)^{n-3}\Phi
\rangle\\
=&2n \,{\rm Tr} (MN) m_{n-1}+8 n(n-1)^2{\rm Tr}(M^2N^2)\, m_{n-2}\\
&+16 n(n-1)^2(n-2)\langle \left(a^{\dagger}M
a^{\dagger}\right)^{n-2} \Phi ,  a^{\dagger}M^2N^3 a^{\dagger}
\,\left(a^{\dagger}N a^{\dagger}\right)^{n-3}\Phi \rangle\\
=&2n \,{\rm Tr} (MN) m_{n-1}+8 n(n-1)^2 {\rm Tr}(M^2N^2)\,m_{n-2}\\
&+16 n(n-1)^2(n-2)\langle  aM^2N^3 a\, \left(a^{\dagger}M
a^{\dagger}\right)^{n-2} \Phi , \left(a^{\dagger}N
a^{\dagger}\right)^{n-3}\Phi \rangle \\
=&2n \,{\rm Tr} (MN) m_{n-1}+8 n(n-1)^2{\rm Tr}(M^2N^2)\, m_{n-2}\\
&+16 n(n-1)^2(n-2)\langle \lbrack   aM^2N^3 a, \left(a^{\dagger}M
a^{\dagger}\right)^{n-2}\rbrack \Phi , \left(a^{\dagger}N
a^{\dagger}\right)^{n-3}\Phi \rangle \ .
\end{align*}
By Lemma \ref{s},
\begin{align*}
\lbrack   aM^2N^3 a, \left(a^{\dagger}M
a^{\dagger}\right)^{n-2}\rbrack \Phi=&2 (n-2){\rm
Tr}(M^3N^3)\,\left(a^{\dagger}M
a^{\dagger}\right)^{n-3}\Phi\\
&+ 4(n-3)(n-2)\, a^{\dagger}M^4N^3 a^{\dagger}
\,\left(a^{\dagger}M a^{\dagger}\right)^{n-4}\Phi \  .
\end{align*}
Thus, since $M, N$ are real,
\begin{align*}
m_n=&2n \,{\rm Tr} (MN) m_{n-1}+8 n(n-1)^2{\rm Tr}(M^2N^2)\, m_{n-2}\\
&+16 n(n-1)^2(n-2)\langle  2 (n-2){\rm
Tr}(M^3N^3)\,\left(a^{\dagger}M
a^{\dagger}\right)^{n-3}\Phi\\
&+ 4(n-3)(n-2)\, a^{\dagger}M^4N^3 a^{\dagger}
\,\left(a^{\dagger}M a^{\dagger}\right)^{n-4}\Phi   ,
\left(a^{\dagger}N a^{\dagger}\right)^{n-3}\Phi \rangle\\
=&2n \,{\rm Tr} (MN) m_{n-1}+8 n(n-1)^2{\rm Tr}(M^2N^2)\, m_{n-2}\\
&+32 n(n-1)^2(n-2)^2  {\rm Tr}(M^3N^3)\,m_{n-3}\\
&+64 n(n-1)^2(n-2)^2 (n-3)\langle a^{\dagger}M^4N^3 a^{\dagger}
\,\left(a^{\dagger}M a^{\dagger}\right)^{n-4}\Phi   ,
\left(a^{\dagger}N a^{\dagger}\right)^{n-3}\Phi \rangle\\
=&2n \,{\rm Tr} (MN) m_{n-1}+8 n(n-1)^2{\rm Tr}(M^2N^2)\, m_{n-2}\\
&+32 n(n-1)^2(n-2)^2  {\rm Tr}(M^3N^3)\,m_{n-3}\\
&+64 n(n-1)^2(n-2)^2 (n-3)\langle \left(a^{\dagger}M
a^{\dagger}\right)^{n-4}\Phi   , \lbrack   a M^4N^3 a,
\left(a^{\dagger}N a^{\dagger}\right)^{n-3} \rbrack \Phi \rangle \
.
\end{align*}
Proceeding in this way, until the $n-4$ in $\left(a^{\dagger}M
a^{\dagger}\right)^{n-4}\Phi $ is reduced to zero, using the fact
that
 \[
 m_0=\|\Phi\|^2=1\ ,
 \]
  we end up with
\begin{align*}
m_n=& \frac{1}{n}\sum_{k=1}^n 2^{2k-1}
\left(n(n-1)(n-2)\cdots(n-k+1)\right)^2  {\rm
Tr}\left((MN)^{k}\right) m_{n-k}\\
=&\frac{1}{n}\sum_{k=1}^n 2^{2k-1} \left((n)_k\right)^2 {\rm
Tr}\left((MN)^{k}\right) m_{n-k}\ .
\end{align*}
\end{proof}

\begin{proposition}\label{proj3} Let  $\Phi$ be a
normalized Fock vacuum vector and let $l=1,2,..$. Then, for all
$M, N\in Sym \left(\mathbb{R}^{l \times l}\right)$ with $\lbrack
M, N \rbrack=0$,
\[
\sigma_{n,m}:=\langle \left(a^{\dagger} M a^{\dagger}\right)^n
\Phi,\left( a^{\dagger}N a^{\dagger}\right)^{m} \Phi\rangle \  ,
\]
satisfies, for $n,m\in\{1,2,...\}$ with $n>m$,  the recursion
\[
\sigma_{n,m}=\frac{1}{m}\sum_{k=1}^{m-1} 2^{2k-1} (n)_k (m)_k{\rm
Tr}\left((MN)^{k}\right)\sigma_{n-k,m-k}  \  ,
\]
where
\[
(x)_k=x(x-1)(x-2)\cdots (x-k+1) \  .
\]
\end{proposition}
\begin{proof}By Lemma \ref{s}
\begin{align*}
\sigma_{n,m}=&\langle \left(a^{\dagger}M a^{\dagger}\right)^n
\Phi,
 \left(a^{\dagger}N a^{\dagger}\right)^m  \Phi\rangle\\
=&\langle (aNa)\left(a^{\dagger}M a^{\dagger}\right)^n  \Phi,
\left(a^{\dagger}N a^{\dagger}\right)^{m-1}  \Phi\rangle\\
=&\langle \lbrack   aNa, \left(a^{\dagger}M a^{\dagger}\right)^n
\rbrack \Phi,\left(a^{\dagger}N a^{\dagger}\right)^{m-1}  \Phi\rangle\\
=&\langle 2n \,{\rm Tr} (MN) \left(a^{\dagger}M
a^{\dagger}\right)^{n-1}\Phi, \left(a^{\dagger}N
a^{\dagger}\right)^{m-1}  \Phi\rangle \\
&+\langle 4 n(n-1)\,a^{\dagger}M^2N a^{\dagger}
\,\left(a^{\dagger}M a^{\dagger}\right)^{n-2} \Phi ,
\left(a^{\dagger}N a^{\dagger}\right)^{m-1}  \Phi\rangle\\
=&2n \,{\rm Tr} (MN) \sigma_{n-1,m-1}+4 n(n-1)\\
&\cdot\langle \left(a^{\dagger}M
a^{\dagger}\right)^{n-2} \Phi , aM^2N a \,\left(a^{\dagger}N a^{\dagger}\right)^{m-1}  \Phi\rangle\\
=&2n \,{\rm Tr} (MN)  \sigma_{n-1,m-1}+4 n(n-1)\\
&\cdot\langle \left(a^{\dagger}M a^{\dagger}\right)^{n-2} \Phi ,
\lbrack   aM^2N a,\left(a^{\dagger}N a^{\dagger}\right)^{m-1}
\rbrack \Phi\rangle \  .
\end{align*}
By Lemma \ref{s},
\begin{align*}
\lbrack   aM^2N a, \left(a^{\dagger}N a^{\dagger}\right)^{m-1}
\rbrack \Phi=&2 (m-1){\rm Tr}(M^2N^2)\,\left(a^{\dagger}N
a^{\dagger}\right)^{m-2}\Phi\\
&+ 4(m-2)(m-1)\, a^{\dagger}M^2N^3 a^{\dagger}
\,\left(a^{\dagger}N a^{\dagger}\right)^{m-3}\Phi \  .
\end{align*}
Thus
\begin{align*}
\sigma_{n,m}=&2n \,{\rm Tr} (MN) \sigma_{n-1,m-1}+4 n(n-1)\langle
\left(a^{\dagger}M
a^{\dagger}\right)^{n-2} \Phi ,\\
& 2 (m-1){\rm Tr}(M^2N^2)\,\left(a^{\dagger}N
a^{\dagger}\right)^{m-2}\Phi \\
& + 4(m-2)(m-1)\, a^{\dagger}M^2N^3 a^{\dagger}
\,\left(a^{\dagger}N a^{\dagger}\right)^{m-3}\Phi
\rangle\\
=&2n \,{\rm Tr} (MN) \sigma_{n-1,m-1}+8 n(n-1)(m-1){\rm Tr}(M^2N^2)\, \sigma_{n-2,m-2}\\
&+16 n(n-1)(m-1)(m-2)\langle \left(a^{\dagger}M
a^{\dagger}\right)^{n-2} \Phi ,  a^{\dagger}M^2N^3 a^{\dagger}
\,\left(a^{\dagger}N a^{\dagger}\right)^{m-3}\Phi \rangle\\
=&2n \,{\rm Tr} (MN) \sigma_{n-1,m-1}+8 n(n-1)(m-1) {\rm Tr}(M^2N^2)\,\sigma_{n-2,m-2}\\
&+16 n(n-1)(m-1)(m-2)\langle  aM^2N^3 a\, \left(a^{\dagger}M
a^{\dagger}\right)^{n-2} \Phi , \left(a^{\dagger}N
a^{\dagger}\right)^{m-3}\Phi \rangle \\
=&2n \,{\rm Tr} (MN) \sigma_{n-1,m-1}+8 n(n-1)(m-1){\rm Tr}(M^2N^2)\, \sigma_{n-2,m-2}\\
&+16 n(n-1)(m-1)(m-2)\langle \lbrack   aM^2N^3 a,
\left(a^{\dagger}M a^{\dagger}\right)^{n-2}\rbrack \Phi ,
\left(a^{\dagger}N a^{\dagger}\right)^{m-3}\Phi \rangle \ .
\end{align*}
By Lemma \ref{s},
\begin{align*}
\lbrack   aM^2N^3 a, \left(a^{\dagger}M
a^{\dagger}\right)^{n-2}\rbrack \Phi=&2 (n-2){\rm
Tr}(M^3N^3)\,\left(a^{\dagger}M
a^{\dagger}\right)^{n-3}\Phi\\
&+ 4(n-3)(n-2)\, a^{\dagger}M^4N^3 a^{\dagger}
\,\left(a^{\dagger}M a^{\dagger}\right)^{n-4}\Phi \  .
\end{align*}
Thus, since $M, N$ are real
\begin{align*}
\sigma_{n,m}=&2n \,{\rm Tr} (MN) \sigma_{n-1,m-1}+8 n(n-1)(m-1){\rm Tr}(M^2N^2)\, \sigma_{n-2,m-2}\\
&+16 n(n-1)(m-1)(m-2)\langle  2 (n-2){\rm
Tr}(M^3N^3)\,\left(a^{\dagger}M
a^{\dagger}\right)^{n-3}\Phi\\
&+ 4(n-3)(n-2)\, a^{\dagger}M^4N^3 a^{\dagger}
\,\left(a^{\dagger}M a^{\dagger}\right)^{n-4}\Phi   ,
\left(a^{\dagger}N a^{\dagger}\right)^{m-3}\Phi \rangle\\
=&2n \,{\rm Tr} (MN) \sigma_{n-1,m-1}+8 n(n-1)(m-1){\rm Tr}(M^2N^2)\, \sigma_{n-2,m-2}\\
&+32 n(n-1)(n-2)(m-1)(m-2)  {\rm Tr}(M^3N^3)\,\sigma_{n-3,m-3}\\
&+64 n(n-1)(n-2)(n-3)(m-1)(m-2)\\
&\cdot \langle a^{\dagger}M^4N^3 a^{\dagger} \,\left(a^{\dagger}M
a^{\dagger}\right)^{n-4}\Phi   ,
\left(a^{\dagger}N a^{\dagger}\right)^{m-3}\Phi \rangle\\
=&2n \,{\rm Tr} (MN) \sigma_{n-1,m-1}+8 n(n-1)(m-1){\rm Tr}(M^2N^2)\, \sigma_{n-2,m-2}\\
&+32 n(n-1)(n-2)(m-1)(m-2)  {\rm Tr}(M^3N^3)\,\sigma_{n-3,m-3}\\
 &+64 n(n-1)(n-2)(n-3)(m-1)(m-2)\\
&\cdot\langle \left(a^{\dagger}M a^{\dagger}\right)^{n-4}\Phi   ,
\lbrack   a M^4N^3 a, \left(a^{\dagger}N a^{\dagger}\right)^{m-3}
\rbrack \Phi \rangle \\
=&\frac{1}{m}\left(2nm \,{\rm Tr} (MN) \sigma_{n-1,m-1}+8 n(n-1)m(m-1){\rm Tr}(M^2N^2)\, \sigma_{n-2,m-2}\right.\\
&+32 n(n-1)(n-2)m(m-1)(m-2)  {\rm Tr}(M^3N^3)\,\sigma_{n-3,m-3}\\
 &+64 n(n-1)(n-2)(n-3)m(m-1)(m-2)\\
&\left.\cdot\langle \left(a^{\dagger}M
a^{\dagger}\right)^{n-4}\Phi , \lbrack   a M^4N^3 a,
\left(a^{\dagger}N a^{\dagger}\right)^{m-3} \rbrack \Phi \rangle
\right)\ .
\end{align*}
Proceeding in this way, until the  $\left(a^{\dagger}M
a^{\dagger}\right)^{n-4}\Phi $ term is reduced to
$\left(a^{\dagger}M a^{\dagger}\right)^{n-m}\Phi $ and the
$\left(a^{\dagger}N a^{\dagger}\right)^{m-3}$ term is reduced to
$a^{\dagger}N a^{\dagger}\Phi $, using the fact that by Lemma
\ref{le2} and the fact that $a\Phi=0$ and $n>m$,
\[
\langle \left(a^{\dagger}M a^{\dagger}\right)^{n-m}\Phi , \lbrack
a M^mN^{m-1} a, a^{\dagger}N a^{\dagger} \rbrack \Phi \rangle=0\ ,
\]
we arrive at
\[
\sigma_{n,m}=\frac{1}{m}\sum_{k=1}^{m-1} 2^{2k-1} (n)_k (m)_k{\rm
Tr}\left((MN)^{k}\right)\sigma_{n-k,m-k}  \  .
\]
\end{proof}

\section{ The Adjoint Action of the Quadratic Lie Group on the Quadratic $*$--Lie Algebra}
\label{adj-act-quad-grp}

In this section, we show that, expressing the adjoint representation of $heis_{\mathbb{C}}(2;n)$
in terms of the $\circ$--operation, in several cases one obtains rather explicit formulae for this
action.
The following Lemma collects some known formulas that we will need.
\begin{lemma}\label{ad-lm-adj-act-eaa}
\begin{align}
e^{Y}Xe^{-Y} =&e^{\lbrack   Y, \ \cdot \rbrack}X\label{ad-e[Y,.]X}\  , \\
e^{\lbrack   Y^*,\ \cdot \rbrack}X =& \left(e^{\lbrack   -Y, \
\cdot \rbrack}X^*\right)^*\label{ad-e[Y*,.]} \  .
\end{align}
Moreover, for any holomorphic function $f$,
\begin{equation}\label{ad-e[Y,.]f(X)}
e^{\lbrack   Y, \ \cdot \rbrack}f(X) = f\left(e^{\lbrack  Y, \
\cdot \rbrack}X\right) \  .
\end{equation}
\end{lemma}
\begin{proof}
\eqref{ad-e[Y,.]X} is a well known identity.
Taking adjoint of both sides of \eqref{ad-e[Y,.]X} one finds
$$
e^{-Y^*}X^*e^{Y^*} =e^{\lbrack  -Y^*, \ \cdot \rbrack}X^* \  .
$$
Therefore
$$
e^{\lbrack  -Y^*, \ \cdot \rbrack}X^* = \left(e^{\lbrack  Y, \
\cdot \rbrack}X\right)^* \  .
$$
Exchanging $Y$ to $-Y$ and $X$ to $X^*$ one finds
\eqref{ad-e[Y*,.]}. \eqref{ad-e[Y,.]f(X)} follows expanding $f$ in
power series and using the fact that $e^{\lbrack  Y, \ \cdot
\rbrack}$ is an homomorphism.
\end{proof}

\begin{remark} \rm
We will mainly use \eqref{ad-e[Y,.]f(X)} in the form
\[
e^{Y}e^{X}e^{-Y} =  e^{e^{\lbrack  Y, \ \cdot \rbrack}X} \  .
\]
\end{remark}
\begin{lemma}
\begin{align}
e^{\lbrack  a C a, \ \cdot \rbrack}a C' a   =& a C' a \  ,
\label{ad-e[aCa,.]aC'a}\\
e^{\lbrack  a C a, \ \cdot \rbrack}a^{\dagger} B a =& a^{\dagger}
B a + a(C\circ
B)a  \  ,\label{ad-e[a C a,.]a+Ba}\\
e^{\lbrack  a C a, \ \cdot \rbrack}a^{\dagger}Aa^{\dagger} =&
a^{\dagger}Aa^{\dagger} + 4\,a^{\dagger}ACa +
4a(C\circ (AC))a  \  ,\label{ad-e[a C a,.]a+Aa+}\\
e^{\lbrack  a^{\dagger}Aa^{\dagger}, \ \cdot
\rbrack}a^{\dagger}A'a^{\dagger} =&
a^{\dagger}A'a^{\dagger} \  ,\label{ad-e[a+Aa+,.]a+A'a+}\\
e^{\lbrack  a^{\dagger}Aa^{\dagger}, \ \cdot \rbrack}a^{\dagger}Ba
= & a^{\dagger}Ba - a(B\circ
A)a  \  ,\label{ad-e[a+Aa+,.]a+Ba}\\
e^{\lbrack  a^{\dagger}Aa^{\dagger}, \ \cdot \rbrack}aCa =&  aCa -
4\,\hbox{Tr}(AC) - 4\,a^{\dagger}(AC)^Ta + 4a^{\dagger}((AC)\circ
A)a^{\dagger}\label{ad-e[a+Aa+,.]aAa} \  .
\end{align}
\end{lemma}
\begin{proof}
\eqref{ad-e[aCa,.]aC'a} is clear.
\begin{align*}
e^{a C a}a^{\dagger} B ae^{-a C a} =&e^{\lbrack   a C a, \ \cdot \rbrack}a^{\dagger} B a\\
=&\sum_{n\ge 0}\frac{1}{n!}\lbrack   a C a, \ \cdot \rbrack^{n}a^{\dagger} B a\\
=&a^{\dagger} B a + \sum_{n\ge 1}\frac{1}{n!}\lbrack   a C a, \ \cdot \rbrack^{n}a^{\dagger} B a\\
 =&a^{\dagger} B a + \sum_{n\ge 1}\frac{1}{n!}\lbrack   a C a, \
\cdot \rbrack^{n-1}\lbrack   a C a,a^{\dagger} B a \rbrack\\
=&^{\eqref{ad-b}} a^{\dagger} B a + \sum_{n\ge
1}\frac{1}{n!}\lbrack   a C a, \
\cdot \rbrack^{n-1}a(C\circ B)a\\
=& a^{\dagger} B a + a(C\circ B)a + \sum_{n\ge
2}\frac{1}{n!}\lbrack   a
C a, \ \cdot \rbrack^{n-1}a(C\circ B)a\\
 =& a^{\dagger} B a + a(C\circ
B)a \  ,
\end{align*}
which is \eqref{ad-e[a C a,.]a+Ba}. Similarly
\begin{align*}
e^{\lbrack  a C a, \ \cdot \rbrack}a^{\dagger}Aa^{\dagger}
=&\sum_{n\ge
0}\frac{1}{n!}\lbrack  a C a, \ \cdot \rbrack^{n}a^{\dagger}Aa^{\dagger}\\
=&a^{\dagger}Aa^{\dagger} + \sum_{n\ge 1}\frac{1}{n!}\lbrack  a C
a, \ \cdot
\rbrack^{n}a^{\dagger}Aa^{\dagger} \\
=&a^{\dagger}Aa^{\dagger} + \sum_{n\ge 1}\frac{1}{n!}\lbrack  a C
a, \ \cdot
\rbrack^{n-1}\lbrack  a C a,a^{\dagger}Aa^{\dagger}\rbrack\\
 =&^{\eqref{ad-a}}
a^{\dagger}Aa^{\dagger} +\sum_{n\ge 1}\frac{1}{n!}\lbrack  a C a,
\ \cdot
\rbrack^{n-1}(2 \,\hbox{Tr}(CA)+4\,a^{\dagger}ACa)\\
=&a^{\dagger}Aa^{\dagger} +\sum_{n\ge 1}\frac{1}{n!}\lbrack  a C
a, \ \cdot
\rbrack^{n-1}(4\,a^{\dagger}ACa)\\
=&a^{\dagger}Aa^{\dagger} + 4\,a^{\dagger}ACa + \sum_{n\ge
2}\frac{1}{n!}\lbrack  a
C a, \ \cdot \rbrack^{n-1}(4\,a^{\dagger}ACa)\\
=&a^{\dagger}Aa^{\dagger} + 4\,a^{\dagger}ACa + 4\,\sum_{n\ge
2}\frac{1}{n!} \lbrack
a C a, \ \cdot \rbrack^{n-2}\lbrack  a C a, a^{\dagger}ACa\rbrack\\
=&^{\eqref{ad-b}}a^{\dagger}Aa^{\dagger} + 4\,a^{\dagger}ACa +
4\,\sum_{n\ge 2}\frac{1}{n!} \lbrack  a C a, \ \cdot
\rbrack^{n-2}\lbrack  a C
a, a^{\dagger}ACa\rbrack\\
=&a^{\dagger}Aa^{\dagger} + 4\,a^{\dagger}ACa + 4\,\sum_{n\ge
2}\frac{1}{n!} \lbrack
a C a, \ \cdot \rbrack^{n-2}a\left(CAC+(CAC)^T\right)a\\
=&^{\eqref{ad-notat-X-circ-Y}}a^{\dagger}Aa^{\dagger} +
4\,a^{\dagger}ACa + 4\,\sum_{n\ge 2}\frac{1}{n!} \lbrack  a C a,
\cdot
\rbrack^{n-2}a(C\circ (AC))a\\
=&^{\eqref{ad-notat-X-circ-Y}}a^{\dagger}Aa^{\dagger} +
4\,a^{\dagger}ACa +
4a(C\circ (AC))a \\
&+4\,\sum_{n\ge 3}\frac{1}{n!}\lbrack  a C a,
\cdot \rbrack^{n-2}a(C\circ (AC))a \\
=& a^{\dagger}Aa^{\dagger} + 4\,a^{\dagger}ACa + 2a(C\circ (AC))a
\  ,
\end{align*}
which is \eqref{ad-e[a C a,.]a+Aa+}. Applying \eqref{ad-e[Y*,.]}
with $Y=aA^*a$, $X=a^{\dagger} B a$ and \eqref{ad-e[a C a,.]a+Ba}
with $C\to -A^*$, $B\to B^*$, one finds
\begin{align*}
e^{\lbrack  a^{\dagger}Aa^{\dagger}, \ \cdot \rbrack}a^{\dagger}Ba
=& \left(e^{\lbrack  -aA^*a, \ \cdot
\rbrack}a^{\dagger}B^*a\right)^*
\\
=& \left(a^{\dagger}B^* a + a(-A^*)\circ (B^*)a\right)^* \\
=&  a^{\dagger}Ba - a(A^*\circ B^*)^*a \\
=&  a^{\dagger}Ba - a(B\circ A)a \  ,
\end{align*}
 which is
\eqref{ad-e[a+Aa+,.]a+Ba}. Applying \eqref{ad-e[Y*,.]} with
$Y=aA^*a$ and $X=aCa$, one obtains
\begin{equation}\label{ad-e[a+Aa+,.]aCa}
e^{\lbrack  a^{\dagger}Aa^{\dagger}, \ \cdot \rbrack}aCa =
\left(e^{\lbrack -aA^*a, \ \cdot
\rbrack}a^{\dagger}C^*a^{\dagger}\right)^*  \ .
\end{equation}
Using \eqref{ad-e[a C a,.]a+Aa+} with $C\to -A^*$ and $A\to C^*$, one has
\begin{align*}
e^{\lbrack  -aA^*a, \ \cdot \rbrack}a^{\dagger}C^*a^{\dagger} =&
a^{\dagger}C^*a^{\dagger}
+ 4\,a^{\dagger}C^*(-A^*)a + 4a((-A^*)\circ (C^*(-A^*)))a\\
=& a^{\dagger}C^*a^{\dagger} - 4\,a^{\dagger}C^*A^*a + 4a(A^*\circ
(C^*A^*))a \  ,
\end{align*}
and, replacing this in the right hand side of
\eqref{ad-e[a+Aa+,.]aCa} one finds
\begin{align*}
e^{\lbrack  a^{\dagger}Aa^{\dagger}, \ \cdot \rbrack}aCa =&
\left(a^{\dagger}C^*a^{\dagger}
- 4\,a^{\dagger}C^*A^*a + 4a(A^*\circ (C^*A^*))a\right)^*\\
=& aCa - 4\,aACa^{\dagger} + 4a^{\dagger}((C^*A^*)^*\circ A)a^{\dagger}\\
=& aCa - 4\,aACa^{\dagger} + 4a^{\dagger}((AC)\circ A)a^{\dagger}\\
=&  aCa - 4\,(\hbox{Tr}(AC) + a^{\dagger}(AC)^Ta) +
4a^{\dagger}((AC)\circ
  A)a^{\dagger}\\
=&  aCa - 4\,\hbox{Tr}(AC) - 4\,a^{\dagger}(AC)^Ta +
4a^{\dagger}((AC)\circ A)a^{\dagger} \  ,
\end{align*}
which is \eqref{ad-e[a+Aa+,.]aAa}.
\end{proof}
\begin{lemma}\label{ad-lm-exp[a+Ba,.](aCa)}
Introducing, for $n\in\mathbb{N}$ and $B, C\in M_{n}(\mathbb{C})$, the inductively defined
notations (see \eqref{ad-notat-X-circ-Y})
\begin{align}
C \, \widehat{\circ} \, B^{\widehat{\circ} \, (n+1)} :=& (C \,
\widehat{\circ} \, B^{\widehat{\circ} \, n})\circ B \,,\,
B^{\widehat{\circ} \, 0} \, \widehat{\circ} \, B:= B
\label{ad-B-hat-circ-n-C} \  ,\\
C \, \widehat{\circ} \, e^{ \, \widehat{\circ} \, (-B)} :=&
\sum_{n\ge 0}\frac{1}{n!}(-1)^{n}C \, \widehat{\circ} \,
B^{\widehat{\circ} \, n} \  ,\label{ad-df-exp-hat-circ(-B)}
\end{align}
one has
\begin{equation}\label{ad-exp[a+Ba,.](aCa)}
e^{\lbrack  a^{\dagger}Ba, \ \cdot \rbrack  }(aCa) =a\left(C \,
\widehat{\circ} \, e^{ \, \widehat{\circ} \, (-B)}\right)a  \  .
\end{equation}

\end{lemma}
\begin{proof}

\begin{align*}
e^{\lbrack  a^{\dagger}Ba, \ \cdot \rbrack  }aCa =&\sum_{n\ge
0}\frac{1}{n!}\lbrack  a^{\dagger}Ba, \ \cdot \rbrack  ^{n}aCa\\
=&a^{\dagger}Ba + \sum_{n\ge 1}\frac{1}{n!}\lbrack  a^{\dagger}Ba,
\ \cdot
\rbrack  ^{n}aCa \\
=&a^{\dagger}Ba + \sum_{n\ge 1}\frac{1}{n!}\lbrack
a^{\dagger}Ba, \ \cdot \rbrack  ^{n-1}\lbrack  a^{\dagger}Ba,aCa \rbrack\\
=&^{\eqref{ad-b}}a^{\dagger}Ba + \sum_{n\ge 1}\frac{1}{n!}\lbrack
a^{\dagger}Ba, \ \cdot \rbrack  ^{n-1}(-1)a(C\circ B)a  \  .
\end{align*}
In order to calculate $\lbrack  a^{\dagger}Ba, \ \cdot \rbrack
^{n-1}(-1)a(C\circ B)a$, note that using the notation
\eqref{ad-B-hat-circ-n-C},
\begin{align*}
\lbrack  a^{\dagger}Ba, \ \cdot \rbrack  ^{2}aCa =&  (-1)\lbrack
a^{\dagger}Ba, a(C\circ B)a\rbrack\\
 =&  (-1)^{2}\lbrack  a^{\dagger}Ba,a((C\circ B)\circ
B)a\rbrack\\
=&a((-1)^{2}C\circ B^{\circ 2})a  \  .
\end{align*}
Suppose by induction that
\begin{equation}\label{ad-[aBa,.]n(a+Ca)}
\lbrack  aBa, \ \cdot \rbrack  ^{n}(a^{\dagger}Ca) =a((-1)^{n}C \,
\widehat{\circ} \, B^{\widehat{\circ} \, n})a  \  .
\end{equation}
Then
\begin{align*}
\lbrack  aBa, \ \cdot \rbrack  ^{n+1}(a^{\dagger}Ca) =&\lbrack
aBa,
\lbrack aBa, \ \cdot \rbrack  ^{n}(a^{\dagger}Ca)\rbrack \\
=&\lbrack aBa, a((-1)^{n}C \, \widehat{\circ} \, B^{\widehat{\circ} \, n})a\rbrack\\
=&^{\eqref{ad-b}}a((-1)((-1)^{n}C \, \widehat{\circ} \,
B^{\widehat{\circ} \, n})\circ B)a \\
=&a((-1)^{n+1}C\circ B^{\circ (n+1)})a  \  .
\end{align*}
Therefore by induction \eqref{ad-[aBa,.]n(a+Ca)} holds for each
$n\in\mathbb{N}$. This implies, in the notation
\eqref{ad-df-exp-hat-circ(-B)},
\begin{align*}
e^{\lbrack  a^{\dagger}Ba, \ \cdot \rbrack  }(aCa) =&\sum_{n\ge
0}\frac{1}{n!}\lbrack  a^{\dagger}Ba, \ \cdot \rbrack  ^{n}(aCa)\\
=&\sum_{n\ge 0}\frac{1}{n!}a((-1)^{n}C \, \widehat{\circ} \,
B^{\widehat{\circ} \, n})a\\
 =&a\left(\sum_{n\ge
0}\frac{1}{n!}(-1)^{n}C \, \widehat{\circ} \, B^{\widehat{\circ}
\, n}\right)a\\
 =&a\left(C \, \widehat{\circ} \, e^{ \,
\widehat{\circ} \, (-B)}\right)a \  ,
\end{align*}
which is \eqref{ad-exp[a+Ba,.](aCa)}.
\end{proof}

\begin{remark} \rm
Note the big difference between the symbols $\circ$ and
$\widehat{\circ}$. The former is a binary operation,
non--commutative and non--associative, but bi--linear and
distributive in the two factors and well behaved with respect to
the adjoint (see \eqref{ad-circ-*}). The latter is \textbf{a
purely symbolic notation} that has only a global meaning. In
particular \textbf{it is not distributive}. The following Lemma
shows however that the notation $\widehat{\circ}$ is well behaved
with respect to the adjoint.
\end{remark}
\begin{lemma}\label{ad-lm-adj-act-ea+a+}
Introducing, for $n\in\mathbb{N}$ and $B, C, G, H\in M_{n}(\mathbb{C})$, the inductively defined
notations (see \eqref{ad-notat-X-circ-Y})
\begin{align*}
B^{\widehat{\circ} \, (n+1)} \, \widehat{\circ} \, C :=&
B\circ(B^{\widehat{\circ} \, n} \, \widehat{\circ} \, C) \,,\,
B^{\widehat{\circ} \, 0} \, \widehat{\circ} \, C:=
C \  ,\\
 e^{ \, \widehat{\circ} \, (-B)}
\, \widehat{\circ} \, C :=&\sum_{n\ge
0}\frac{1}{n!}(-1)^{n}B^{\widehat{\circ} \, n} \, \widehat{\circ}
\, C \  ,
\end{align*}
the following identities hold:
\begin{align}
\left( G \, \widehat{\circ} \, e^{\widehat{\circ} \, (-H)}
\right)^* =& e^{\widehat{\circ} \, (-H^*)} \, \widehat{\circ} \,
G^*\label{ad-exp-hat-circ(-H)2} \  ,\\
 e^{\lbrack  a^{\dagger}Ba, \ \cdot
\rbrack  }a^{\dagger}Aa^{\dagger} =& \left(e^{\lbrack
-a^{\dagger}B^*a, \ \cdot \rbrack  }aA^*a\right)^*
=a^{\dagger}\left(e^{\circ (-B)} \,
\widehat{\circ} \, A\right)a^{\dagger}\label{ad-exp[a+Ba,.](a+Aa+)} \  ,\\
e^{\lbrack  a^{\dagger}Ba, \ \cdot \rbrack  }a^{\dagger}B'a
=&a^{\dagger}\left(e^{\lbrack B, \ \cdot \rbrack  }B'\right)a
=a^{\dagger}\left(e^{B}B'e^{-B}\right)a\label{ad-exp[a+Ba,.](a+B'a)}
\ .
\end{align}
\end{lemma}

\begin{proof}
Recalling \eqref{ad-df-exp-hat-circ(-B)}, one has for general matrices $G$ and $H$
\begin{equation}\label{ad-exp-hat-circ(-H)}
G \, \widehat{\circ} \, e^{\widehat{\circ} \, (-H)} := \sum_{n\ge
0}\frac{1}{n!}(-1)^{n}G \, \widehat{\circ} \,  H^{ \,
\widehat{\circ} \,  n} \  ,
\end{equation}
and
\[
G \, \widehat{\circ} \,  H^{ \, \widehat{\circ} \,  (n+1)}
:= (G \, \widehat{\circ} \, H^{ \, \widehat{\circ} \,  n})\circ H   \\
\,,\, H^{\widehat{\circ} \, 0}\circ B := B  \  .
\]
From \eqref{ad-circ-*} one deduces that
\[
(G \, \widehat{\circ} \,  H^{ \, \widehat{\circ} \,  1})^* =
(G\circ H)^* = H^*\circ G^*  =(H^*)^{ \, \widehat{\circ} \,  1} \,
\widehat{\circ} \,  G^* \  .
\]
Suppose by induction that
\begin{equation}\label{ad-ind-Ghat-circ-Hn-adj}
(G \, \widehat{\circ} \,  H^{ \, \widehat{\circ} \, n})^* =(H^*)^{
\, \widehat{\circ} \,  n} \, \widehat{\circ} \,  G^* \  .
\end{equation}
Then, recalling the second and the first identity in \eqref{ad-B-hat-circ-n-C} and using
\eqref{ad-ind-Ghat-circ-Hn-adj}, one has
\begin{align*}
(G \, \widehat{\circ} \,  H^{ \, \widehat{\circ} \,  (n+1)})^* =&
\left((G \, \widehat{\circ} \,  H^{ \, \widehat{\circ} \, n})\circ
H\right)^* \\
=& H^*\circ (G \, \widehat{\circ} \,  H^{ \,
\widehat{\circ} \,  n})^*\\
=& H^*\circ ((H^*)^{ \, \widehat{\circ} \,  n} \, \widehat{\circ}
\,  G^*)\\
 =&(H^*)^{ \, \widehat{\circ} \, (n+1)} \,
\widehat{\circ} \,  G^*  \  .
\end{align*}
Therefore \eqref{ad-ind-Ghat-circ-Hn-adj} holds for all
$n\in\mathbb{N}$. \eqref{ad-exp-hat-circ(-H)} then
implies\begin{align*}
\left( G \, \widehat{\circ} \,
e^{\widehat{\circ} \, (-H)} \right)^* =&\sum_{n\ge
0}\frac{1}{n!}(-1)^{n}\left(G \, \widehat{\circ} \, H^{ \,
\widehat{\circ} \, n}\right)^*\\
 =&\sum_{n\ge
0}\frac{1}{n!}(-1)^{n}\left((H^*)^{\, \widehat{\circ} \, n} \,
\widehat{\circ} \, G^* \right)\\
 =& e^{\widehat{\circ} \, (-H^*)}
\, \widehat{\circ} \, G^* \  ,
\end{align*}
which is \eqref{ad-exp-hat-circ(-H)2}. Applying the identity
\eqref{ad-e[Y*,.]} with $Y=a^{\dagger}B^*a$,
$X=a^{\dagger}Aa^{\dagger}$
\begin{align*}
e^{\lbrack  a^{\dagger}Ba, \ \cdot \rbrack
}a^{\dagger}Aa^{\dagger} =&
\left(e^{\lbrack -a^{\dagger}B^*a, \ \cdot \rbrack  }aA^*a\right)^*\\
=&^{\eqref{ad-exp[a+Ba,.](aCa)}} \left(a\left((A^*)\,
\widehat{\circ} \, e^{\circ (-B^*)}\right)a\right)^*\\
=&a^{\dagger}\left((A^*)\, \widehat{\circ} \, e^{\circ
(-B^*)}\right)^*a^{\dagger}\\
=&^{\eqref{ad-exp-hat-circ(-H)2}} a\left((A^*)\, \widehat{\circ}
\, e^{\circ (-B^*)}\right)^*a\\
 =&a^{\dagger}\left(e^{\circ (-B)} \,
\widehat{\circ} \, A\right)a^{\dagger} \  ,
\end{align*}
which is \eqref{ad-exp[a+Ba,.](a+Aa+)}. Similarly,
\begin{align*}
e^{\lbrack  a^{\dagger}Ba, \ \cdot \rbrack  }a^{\dagger}B'a
=&\sum_{n\ge
0}\frac{1}{n!}\lbrack  a^{\dagger}Ba, \ \cdot \rbrack  ^{n}a^{\dagger}B'a\\
=&a^{\dagger}B'a + \sum_{n\ge 1}\frac{1}{n!}\lbrack
a^{\dagger}Ba, \ \cdot
\rbrack  ^{n}a^{\dagger}B'a\\
 =&a^{\dagger}B'a + \sum_{n\ge
1}\frac{1}{n!}\lbrack a^{\dagger}Ba, \ \cdot \rbrack ^{n-1}\lbrack
a^{\dagger}Ba,a^{\dagger}B'a \rbrack \  .
\end{align*}
One has
\[
\lbrack  a^{\dagger}Ba,a^{\dagger}B'a] \lbrack  a^{\dagger}Ba, \
\cdot \rbrack a^{\dagger}B'a =^{\eqref{ad-c}} a^{\dagger}\lbrack
B,B']a = a^{\dagger}(\lbrack  B, \ \cdot \rbrack  B')a \  .
\]
Suppose by induction that
\begin{equation}\label{ad-[a+Ba,.]na+B'a}
\lbrack  a^{\dagger}Ba, \ \cdot \rbrack  ^{n}a^{\dagger}B'a =
a^{\dagger}(\lbrack B, \ \cdot \rbrack  ^{n}B')a \  .
\end{equation}
Then
\begin{align*}
\lbrack  a^{\dagger}Ba, \ \cdot \rbrack  ^{n+1}a^{\dagger}B'a
=&\lbrack a^{\dagger}Ba,\lbrack  a^{\dagger}Ba, \ \cdot \rbrack
^{n}a^{\dagger}B'a\rbrack =\lbrack
a^{\dagger}Ba,a^{\dagger}(\lbrack  B, \ \cdot \rbrack  ^{n}B')a\rbrack\\
=&a^{\dagger}(\lbrack  B,\lbrack  B, \ \cdot \rbrack
^{n}B'\rbrack) = a^{\dagger}(\lbrack B, \ \cdot \rbrack
^{n+1}B')a \  .
\end{align*}
Therefore \eqref{ad-[a+Ba,.]na+B'a} holds for all
$n\in\mathbb{N}$. It follows that
\begin{align*}
e^{\lbrack  a^{\dagger}Ba, \ \cdot \rbrack  }a^{\dagger}B'a
=&\sum_{n\ge
0}\frac{1}{n!}\lbrack  a^{\dagger}Ba, \ \cdot \rbrack  ^{n}a^{\dagger}B'a\\
=&\sum_{n\ge 0}\frac{1}{n!}a^{\dagger}(\lbrack  B, \ \cdot \rbrack
^{n}B')a\\
=&a^{\dagger}\left(\sum_{n\ge 0}\frac{1}{n!}\lbrack  B, \ \cdot
\rbrack
^{n}B'\right)a\\
 =&a^{\dagger}\left(e^{\lbrack  B, \ \cdot \rbrack
}B'\right)a \  ,
\end{align*} which is equivalent to
\eqref{ad-exp[a+Ba,.](a+B'a)}.
\end{proof}
\begin{lemma}\label{ad-C-B-comm-sym}
If $C$ and $B$ commute and are both symmetric
\begin{equation}\label{ad-exp-hat-circ(-B)comm}
C \, \widehat{\circ} \, e^{ \, \widehat{\circ} \, (-B)}
=Ce^{-2B}=Ce^{-(B\circ 1)} \  .
\end{equation}
\end{lemma}
\begin{proof}
Since $C$ and $B$ commute and are both symmetric, one has
$$
C\cdot B=CB+ (CB)^{T}=cB+B^{c} =2CB \  .
$$
Therefore
$$
(C\circ B)\circ B=2(2cB)=c(2B)^{2} \  .
$$
Suppose by induction that
\begin{equation}\label{ad-C-circ-Bcircn=C(2B)n-comm}
C\circ B^{\circ n}=C(2B)^{n} \  .
\end{equation}
Then, since $(C\circ B^{\circ n})$ is symmetric and commutes with $C$,
$$
C\circ B^{\circ(n+1)}=(C\circ B^{\circ n})\circ B=(C(2B)^{n})\circ
B =2(C(2B)^{n})B=C(2B)^{n+1} \  .
$$
Therefore by induction \eqref{ad-C-circ-Bcircn=C(2B)n-comm} holds for each $n\in\mathbb {N}$.
In this case \eqref{ad-df-exp-hat-circ(-B)} becomes
\begin{align*}
C \, \widehat{\circ} \, e^{ \, \widehat{\circ} \, (-B)}
=&\sum_{n\ge 0}\frac{1}{n!}(-1)^{n}C \, \widehat{\circ} \,
B^{\widehat{\circ} \, n}\\
 =&\sum_{n\ge
0}\frac{1}{n!}(-1)^{n}C(2B)^{n} =C\sum_{n\ge
0}\frac{1}{n!}(-1)^{n}(2B)^{n} =Ce^{-2B} \  ,
\end{align*} which is
\eqref{ad-exp-hat-circ(-B)comm}.
\end{proof}

\subsection{Commutation relations among exponentials}

Recall that, according to the splitting lemma
\begin{align*}
e^{\left(a^{\dagger}Aa^{\dagger}+a^{\dagger}Ba+aCa\right)}
=&e^{-\frac{1}{2}\,\hbox{Tr}(B)
+\frac{1}{2}\hbox{Tr}\left(g(A,B,C)\right)}\\
&\cdot
e^{\frac{1}{2}a^{\dagger}f(A,B,C)a^{\dagger}}e^{a^{\dagger}g(A,B,C)
a} e^{\frac{1}{2}a \hat h(A,B,C) a} \  .
\end{align*}

\begin{lemma}\label{ad-sl}

\begin{align}
e^{aMa}e^{a^{\dagger}N a^{\dagger}} =&e^{-\frac{1}{2}\,\hbox{Tr}(4MN) +\frac{1}{2}\hbox{ Tr}\left(g(N,4MN,2(M\circ
(NM)))\right)}\label{ad-11n}\\
&\cdot e^{a^{\dagger}Na^{\dagger} + 4\,a^{\dagger}NMa + 4a(M\circ
(NM))a}e^{aMa}
e^{\frac{1}{2}a^{\dagger}\hat{f}(N,4MN,4(M\circ (NM)))a^{\dagger}}\notag\\
&\cdot e^{a^{\dagger}g(N,4MN,2(M\circ (NM))) a} e^{\frac{1}{2}a(M
+ \hat
h(N,4MN,2(M\circ (NM)))) a}\ ,\notag\\
e^{aMa}e^{ a^{\dagger}N a} =&e^{a^{\dagger} N a + a(M\circ
N)a}e^{aMa}\label{ad-22n}\\
 =&e^{-\frac{1}{2}\,\hbox{Tr}(N)
+\frac{1}{2}\hbox{Tr}\left(g(0,N,(M\circ N))\right)}
e^{\frac{1}{2}a^{\dagger}f(0,N,(M\circ
N))a^{\dagger}}e^{a^{\dagger}g(0,N,(M\circ
N)) a}\notag\\
&\cdot e^{a(M + \frac{1}{2} \hat h(0,N,(M\circ N))) a}\ ,\notag\\
e^{a^{\dagger}Ma}e^{a^{\dagger}N a^{\dagger}} =&
e^{a^{\dagger}Na^{\dagger}}e^{a^{\dagger}(M\circ
N)a^{\dagger} + a^{\dagger}Ma}\label{ad-33n}\\
=& e^{a^{\dagger}Na^{\dagger}} e^{-\frac{1}{2}\,\hbox{Tr}(M)
+\frac{1}{2}\hbox{Tr}\left(g((M\circ N),M,0)\right)}
e^{\frac{1}{2}a^{\dagger}f((M\circ
N),M,0)a^{\dagger}}\notag\\
&\cdot e^{a^{\dagger}g((M\circ N),M,0) a} e^{\frac{1}{2}a \hat
h((M\circ
N),M,0) a}\notag\ ,\\
 e^{a^{\dagger}Ba}e^{a^{\dagger}Aa^{\dagger}} =&
e^{a^{\dagger}\left(e^{\circ (-B)} \, \widehat{\circ} \,
A\right)a^{\dagger}}e^{a^{\dagger}Ba}\label{ad-e(a+Ba)-e(a+Aa+)} \
.
\end{align}
\end{lemma}
\begin{proof}
From \eqref{ad-e[a C a,.]a+Aa+} and \eqref{ad-e[Y,.]f(X)}
\begin{align*}
e^{aMa}e^{a^{\dagger}Na^{\dagger}}e^{-aMa} =&e^{\lbrack  aMa, \
\cdot \rbrack }e^{a^{\dagger}Na^{\dagger}} =e^{e^{\lbrack  aMa, \
\cdot \rbrack
}a^{\dagger}Na^{\dagger}}\\
=&e^{a^{\dagger}Na^{\dagger} + 4\,a^{\dagger}NMa + 2a(M\circ
(NM))a}\ ,
\end{align*}
which is equivalent to
\begin{align*}
 e^{aMa}e^{a^{\dagger}Na^{\dagger}} =&e^{a^{\dagger}Na^{\dagger} +
4\,a^{\dagger}NMa + a2(M\circ (NM))a}e^{aMa}\\
=& e^{-\frac{1}{2}\,\hbox{Tr}(4MN)
+\frac{1}{2}\hbox{Tr}\left(g(N,4MN,2(M\circ (NM)))\right)}\\
& e^{\frac{1}{2}a^{\dagger}f(N,4MN,4(M\circ (NM)))a^{\dagger}}
e^{a^{\dagger}g(N,4MN,2(M\circ (NM))) a}\\
&e^{\frac{1}{2}a \hat h(N,4MN,2(M\circ (NM))) a}e^{aMa}\\
 =&e^{-\frac{1}{2}\,\hbox{Tr}(4MN) +\frac{1}{2}\hbox{Tr}\left(g(N,4MN,2(M\circ (NM)))\right)}\\
& e^{\frac{1}{2}a^{\dagger}f(N,4MN,4(M\circ (NM)))a^{\dagger}}
e^{a^{\dagger}g(N,4MN,2(M\circ (NM))) a}\\
&e^{\frac{1}{2}a(M + \hat h(N,4MN,2(M\circ (NM)))) a}\ ,
\end{align*}
which is \eqref{ad-11n}. Recalling \eqref{ad-notat-X-circ-Y} one
verifies that
\begin{align*}
2(M\circ (NM)) =&2( M(NM) + (M(NM))^T) =2( MNM + (NM)^TM^T)\\
=&2( MNM + M^TN^TM^T) =2( MNM +(MNM)^T) \  .
\end{align*}
Comparing \eqref{ad-11n} with equation (2.13) in the paper, i.e.
\begin{align*}
e^{aMa}e^{ a^{\dagger}N a^{\dagger}} =& e^{\frac{1}{2}\hbox{Tr}\left(g\left(N,4NM,2\left(MNM+\left(MNM\right)^T\right)\right)\right)}\\
& \cdot
 e^{\frac{1}{2}a^{\dagger}\hat f\left(N,4NM,2\left(MNM+\left(MNM\right)^T\right)\right)a^{\dagger}}\\
&\cdot
e^{a^{\dagger}g\left(N,4NM,2\left(MNM+\left(MNM\right)^T\right)\right)a}\\
& \cdot e^{\frac{1}{2}a \hat
h\left(N,4NM,2\left(MNM+\left(MNM\right)^T\right)\right)a}\ ,
\end{align*}
we see that in the scalar term,
$$
-\frac{1}{2}\,\hbox{Tr}(4MN)
$$
is missing, and in the $a$--$a$--term, $M + $ is missing. From
\eqref{ad-e[a C a,.]a+Ba} one has
\begin{align*}
e^{aMa}e^{ a^{\dagger}N a}e^{-aMa} =&e^{\lbrack  aMa, \ \cdot
\rbrack }e^{a^{\dagger}Na} =e^{e^{\lbrack  aMa, \ \cdot \rbrack
}a^{\dagger}Na}
=e^{a^{\dagger} N a + a(M\circ N)a}\\
 =&e^{-\frac{1}{2}\,\hbox{Tr}(N) +\frac{1}{2}\hbox{Tr}\left(g(0,N,(M\circ N))\right)}
e^{\frac{1}{2}a^{\dagger}f(0,N,(M\circ N))a^{\dagger}}\\
&\cdot e^{a^{\dagger}g(0,N,(M\circ N)) a} e^{\frac{1}{2}a \hat
h(0,N,(M\circ N)) a} \  .
\end{align*}
This implies
\begin{align*}
e^{aMa}e^{ a^{\dagger}N a} =&e^{a^{\dagger} N a + a(M\circ N)a}\\
=&e^{-\frac{1}{2}\,\hbox{Tr}(N) +\frac{1}{2}\hbox{Tr}\left(g(0,N,(M\circ N))\right)}
e^{\frac{1}{2}a^{\dagger}f(0,N,(M\circ N))a^{\dagger}}\\
&\cdot e^{a^{\dagger}g(0,N,(M\circ N)) a} e^{\frac{1}{2}a \hat
h(0,N,(M\circ N)) a}e^{aMa}\ ,
\end{align*}
 or equivalently
\begin{align*}
e^{aMa}e^{ a^{\dagger}N a} =&e^{a^{\dagger} N a + a(M\circ N)a}\\
=&e^{-\frac{1}{2}\,\hbox{Tr}(N) +\frac{1}{2}\hbox{Tr}\left(g(0,N,(M\circ N))\right)}
e^{\frac{1}{2}a^{\dagger}f(0,N,(M\circ N))a^{\dagger}}\\
&\cdot e^{a^{\dagger}g(0,N,(M\circ N)) a} e^{a(M + \frac{1}{2}
\hat h(0,N,(M\circ N))) a}\ ,
\end{align*}
which is \eqref{ad-22n}. Comparing \eqref{ad-22n} with equation
(2.14) in the paper, i.e.
\begin{align*}
e^{aMa}e^{ a^{\dagger}N a} =&e^{\hbox{Tr}\left(-\frac{1}{2}N
+\frac{1}{2}g\left(0,N,M N+(MN)^T\right)\right)}\\
& \cdot e^{ a^{\dagger}g\left(0,N,M N+(M N)^T)\right) a} \cdot
e^{a\left(M+\frac{1}{2}\hat h\left(0,N,M N+(M
N)^T\right)\right)a}\ ,
\end{align*}
we see that they coincide. Taking the adjoint of \eqref{ad-22n}
one finds
\[
(e^{aMa}e^{ a^{\dagger}N a})^* =(e^{a^{\dagger} N a + a(M\circ
N)a}e^{aMa})^*\ ,
\]
which is equivalent to
\[
e^{ a^{\dagger}N^* a}e^{a^{\dagger}M^*a^{\dagger}}
e^{a^{\dagger}M^*a^{\dagger}}e^{a^{\dagger} N^* a +
a^{\dagger}(M\circ N)^*a^{\dagger}} =
e^{a^{\dagger}M^*a^{\dagger}}e^{a^{\dagger} N^* a +
a^{\dagger}(N^*\circ M^*)a^{\dagger}} \  .
\]
With the changes
$$
N^*\to M \,,\, M^*\to N\ ,
$$
one finds
\begin{align*}
e^{ a^{\dagger}Ma}e^{a^{\dagger}Na^{\dagger}} =&
e^{a^{\dagger}Na^{\dagger}}e^{a^{\dagger}(M\circ
N)a^{\dagger} + a^{\dagger}Ma}\\
=&^{\label{ad-Split-lm}} e^{a^{\dagger}Na^{\dagger}}
e^{-\frac{1}{2}\,\hbox{Tr}(M) +\frac{1}{2}\hbox{Tr}\left(g((M\circ
N),M,0)\right)}
e^{\frac{1}{2}a^{\dagger}f((M\circ N),M,0)a^{\dagger}}\\
&\cdot e^{a^{\dagger}g((M\circ N),M,0) a} e^{\frac{1}{2}a \hat
h((M\circ N),M,0) a}\ ,
\end{align*}
which is \eqref{ad-33n}. Finally, from the identities
$$
e^{a^{\dagger}Ba}e^{a^{\dagger}Aa^{\dagger}}e^{-a^{\dagger}Ba} =
e^{\lbrack a^{\dagger}Ba, \ \cdot \rbrack
}e^{a^{\dagger}Aa^{\dagger}} = e^{e^{\lbrack a^{\dagger}Ba, \
\cdot \rbrack }a^{\dagger}Aa^{\dagger}}
=^{\eqref{ad-exp[a+Ba,.](a+Aa+)}} e^{a^{\dagger}\left(e^{\circ
(-B)} \, \widehat{\circ} \, A\right)a^{\dagger}}
$$
one deduces
$$
e^{a^{\dagger}Ba}e^{a^{\dagger}Aa^{\dagger}} =
e^{a^{\dagger}\left(e^{\circ (-B)} \, \widehat{\circ} \,
A\right)a^{\dagger}}e^{a^{\dagger}Ba}\ ,
$$
which is \eqref{ad-e(a+Ba)-e(a+Aa+)}.
\end{proof}

\subsection{Normal order of products of elements in type $2$ coordinates}

\begin{theorem}\label{ad-th-norm-ord-prod-type2-els}
For $A, C, A', C'\in M_{d, sym}(\mathbb{C})$, $ B, B'\in M_{d}(\mathbb{C})$, the normally ordered
form of the product
\[
e^{ a^{\dagger}A a^{\dagger}} e^{a^{\dagger}B a} e^{ a C a} e^{
a^{\dagger}A'a^{\dagger}} e^{a^{\dagger}B'a} e^{ a C'a}\ ,
\]
is
\begin{equation}\label{ad-prod-type2-els5}
c_{1}c_{2}e^{a^{\dagger}A_{4}a^{\dagger}}e^{a^{\dagger}Ba}e^{a^{\dagger}B_{1}a}e^{a^{\dagger}B_{2}a}e^{aC_{3}a}\
,
\end{equation}
with $c_{1}$ given by \eqref{ad-c1-type2}, $c_{2}$ by
\eqref{ad-c2-type2}, $A_{4}$, $B_{1}$, $B_{2}$ and $C_{3}$
respectively by \eqref{ad-A4-prod}, \eqref{ad-A3-C3-prod},
\eqref{ad-A2-B2-C2-type2}, \eqref{ad-A1-B1-type2}.
\end{theorem}
\begin{proof}.
One has
\begin{equation}\label{ad-prod-type2-els2}
e^{ a^{\dagger}A a^{\dagger}} e^{a^{\dagger}B a} \left(e^{ a C a}
e^{ a^{\dagger}A'a^{\dagger}}\right) e^{a^{\dagger}B'a} e^{ a C'a}
\  .
\end{equation}
Use \eqref{ad-11n} to write
$$
e^{aCa} e^{a^{\dagger}A'a^{\dagger}} =
c_{1}e^{a^{\dagger}A_{1}a^{\dagger}}e^{a^{\dagger}B_{1}a}e^{aC_{1}a}\
,
$$
with
\begin{align}
c_{1} =&e^{-\frac{1}{2}\,\hbox{Tr}(4CA') +\frac{1}{2}\hbox{Tr}\left(g(A',4CA',2(C\circ (A'C)))\right)}\label{ad-c1-type2}\\
A_{1} =& \frac{1}{2}\hat{f}(A',4CA',4(C\circ (A'C)))\notag\\
B_{1}& = g(A',4CA',2(C\circ (A'C))) \label{ad-A3-C3-prod}\\
C_{1} =& \frac{1}{2}(C + \hat h(A',4CA',2(C\circ (A'C))))\notag \
.
\end{align}
Thus \eqref{ad-prod-type2-els2} becomes
\begin{equation}\label{ad-prod-type2-els3}
c_{1}e^{ a^{\dagger}A a^{\dagger}} e^{a^{\dagger}B
a}e^{a^{\dagger}A_{1}a^{\dagger}}e^{a^{\dagger}B_{1}a}
\left(e^{aC_{1}a} e^{a^{\dagger}Ba}\right) e^{ a C'a} \  .
\end{equation}
Use \eqref{ad-22n} to write
$$
e^{aC_{1}a} e^{a^{\dagger}Ba} =
c_{2}e^{a^{\dagger}A_{2}a^{\dagger}}
e^{a^{\dagger}B_{2}a}e^{aC_{2}a}\ ,
$$
with
\begin{align}
c_{2} =&e^{-\frac{1}{2}\,\hbox{Tr}(B) +\frac{1}{2}\hbox{Tr}\left(g(0,B,(C_{1}\circ B))\right)}\label{ad-c2-type2}\\
A_{2} =& \frac{1}{2}f(0,B,(C_{1}\circ B))\\
B_{2} = &g(0,B,(C_{1}\circ B))\label{ad-A2-B2-C2-type2} \\
 C_{2} =& C_{1} +
\frac{1}{2} \hat h(0,B,(C_{1}\circ B)) \  .
\end{align}
Thus \eqref{ad-prod-type2-els3} becomes
\begin{align*}
&c_{1}e^{ a^{\dagger}A a^{\dagger}} e^{a^{\dagger}B
a}e^{a^{\dagger}A_{1}a^{\dagger}}e^{a^{\dagger}B_{1}a}
\left(c_{2}e^{a^{\dagger}A_{2}a^{\dagger}}
e^{a^{\dagger}B_{2}a}e^{aC_{2}a}\right)
e^{ a C'a}\\
=& c_{1}c_{2} e^{ a^{\dagger}A a^{\dagger}} e^{a^{\dagger}B
a}e^{a^{\dagger}A_{1}a^{\dagger}}e^{a^{\dagger}B_{1}a}
e^{a^{\dagger}A_{2}a^{\dagger}}
e^{a^{\dagger}B_{2}a}e^{aC_{2}a} e^{ a C'a}\\
=& c_{1}c_{2} e^{ a^{\dagger}A a^{\dagger}} e^{a^{\dagger}B
a}e^{a^{\dagger}A_{1}a^{\dagger}}e^{a^{\dagger}B_{1}a}
e^{a^{\dagger}A_{2}a^{\dagger}}
e^{a^{\dagger}B_{2}a} e^{a(C_{2}+C')a}\\
=& c_{1}c_{2} e^{ a^{\dagger}A a^{\dagger}} e^{a^{\dagger}B
a}e^{a^{\dagger}A_{1}a^{\dagger}}
\left(e^{a^{\dagger}B_{1}a}e^{a^{\dagger}A_{2}a^{\dagger}}\right)
e^{a^{\dagger}B_{2}a}e^{aC_{3}a}\\
=&^{\eqref{ad-e(a+Ba)-e(a+Aa+)}} c_{1}c_{2} e^{ a^{\dagger}A
a^{\dagger}} e^{a^{\dagger}B a}e^{a^{\dagger}A_{1}a^{\dagger}}
e^{a^{\dagger}\left(e^{\circ (-B_{1})} \, \widehat{\circ} \,
A_{2}\right)a^{\dagger}}e^{a^{\dagger}B_{1}a}
e^{a^{\dagger}B_{2}a}e^{aC_{3}a}\\
=& c_{1}c_{2} e^{ a^{\dagger}A a^{\dagger}} e^{a^{\dagger}B a}
e^{a^{\dagger}\left(A_{1} + e^{\circ (-B_{1})} \, \widehat{\circ}
\, A_{2}\right) a^{\dagger}}
e^{a^{\dagger}B_{1}a} e^{a^{\dagger}B_{2}a}e^{aC_{3}a}\\
=& c_{1}c_{2} e^{ a^{\dagger}A a^{\dagger}} e^{a^{\dagger}B a}
e^{a^{\dagger} A_{3}a^{\dagger}} e^{a^{\dagger}B_{1}a}
e^{a^{\dagger}B_{2}a}e^{aC_{3}a}\ ,
\end{align*}
with
\begin{equation}
C_{3}:= C_{2}+C' \,\,;\,\, A_{3}:=  A_{1} + e^{\circ (-B_{1})} \,
\widehat{\circ} \, A_{2}\label{ad-A1-B1-type2}  \  .
\end{equation}
Finally
\begin{align*}
&c_{1}c_{2} e^{ a^{\dagger}A a^{\dagger}} \left(e^{a^{\dagger}B
a}e^{a^{\dagger}
A_{3}a^{\dagger}}\right) e^{a^{\dagger}B_{1}a} e^{a^{\dagger}B_{2}a}e^{aC_{3}a}\\
=&^{\eqref{ad-e(a+Ba)-e(a+Aa+)}} c_{1}c_{2} \left(e^{ a^{\dagger}A
a^{\dagger}} e^{a^{\dagger}\left(e^{\circ (-B)} \, \widehat{\circ}
\, A_{3}\right)a^{\dagger}}\right) e^{a^{\dagger}Ba}
e^{a^{\dagger}B_{1}a}e^{a^{\dagger}B_{2}a}e^{aC_{3}a}\\
=& c_{1}c_{2} \left(e^{a^{\dagger}\left(A + e^{\circ (-B)} \,
\widehat{\circ} \, A_{3}\right)a^{\dagger}}\right)
e^{a^{\dagger}Ba}e^{a^{\dagger}B_{1}a}e^{a^{\dagger}B_{2}a}e^{aC_{3}a}\\
=&c_{1}c_{2}e^{a^{\dagger}A_{4}a^{\dagger}}e^{a^{\dagger}Ba}e^{a^{\dagger}B_{1}a}e^{a^{\dagger}B_{2}a}e^{aC_{3}a}\
,
\end{align*}
with
\begin{equation}
A_{4} := A + e^{\circ (-B)} \, \widehat{\circ} \,
A_{3}\label{ad-A4-prod}\ ,
\end{equation}
which is \eqref{ad-prod-type2-els5}.
\end{proof}

\section{Exponentiability of the Quadratic Algebra in the Fock Representation}
\label{Exp-quadr-alg-in-Fck-repr}

In this section we prove an estimate (see Theorem \ref{2eqr-06}
below) which implies that, in the Fock representation, any vector
in the dense sub--space linearly spanned by the number vectors is
analytic for any multiple of any skew--adjoint element of the
quadratic algebra. Thus, by Nelson's analytic vector theorem, the
hermitian elements of this algebra are essentially self--adjoint
and their exponential series converges strongly on the linear span
of the number vectors.
From this the existence and unitarity of the quadratic Weyl operators follows.\\
In the Fock representation the Boson creation--annihilation operators $a_{j}^{\pm}$ are realized
on $\Gamma_{sym}\left(  \mathcal{H}\right):=\bigoplus_{n=0}^{\infty}\mathcal{H}^{\widehat{\otimes}n}$, where $\mathcal{H}$ is a $d$--dimensional complex Hilbert space and
\begin{equation}\label{Fock1}
\mathcal{H}^{\widehat{\otimes}0}:=\mathbb{C}\cdot\Phi \,,\,
\|\Phi\| =1 \  .
\end{equation}
The vectors
$$
\{a_{d}^{+ n_{d}}\cdots a_{1}^{+ n_{1}}\Phi \ : \
(n_{1},\dots,n_{d})\in\mathbb{N}^d\}\ ,
$$
are total in $\Gamma_{sym}\left(  \mathcal{H}\right)$ and their scalar product is uniquely determined
by \eqref{Fock1} and the conditions
$$
\lbrack  a_{j}^{-},a_{k}^{\dagger}\rbrack = \delta_{jk}\,,\,
a_{j}^{-}\Phi =0 \  .
$$
The following is a known result whose proof is included for completeness.
\begin{lemma}
Uniformly in $d\in\mathbb{N}$ one has, for all
$k\in\left\{  1,\ldots,d\right\}$and $n\in\mathbb{N}$,
\begin{equation}\label{2eqr-01}
\left\Vert
a_{k}^{-}\Big|_{\mathcal{H}^{\widehat{\otimes}n}}\right\Vert\leq\sqrt{n}
\,,\, \left\Vert
a_{k}^{\dagger}\Big|_{\mathcal{H}^{\widehat{\otimes}n}}\right\Vert
\leq\sqrt{n+1}\ ,
\end{equation}
where here and in the following $\Big|$ denotes restriction.
\end{lemma}
\begin{proof}
\begin{align*}
a_{k}^{-}a_{d}^{+ n_{d}}\cdots a_{1}^{+ n_{1}}\Phi =&\lbrack
a_{k}^{-},a_{d}^{+ n_{d}}\cdots a_{1}^{+ n_{1}}\rbrack\Phi
=\sum_{h=0}^{d-1} \cdots \lbrack  a_{k}^{-},a_{d-h}^{+
n_{d-h}}\rbrack \cdots a_{1}^{+ n_{1}}\Phi\\
=&\sum_{h=0}^{d-1} \delta_{k,d-h}n_{d-h}\cdots a_{d-h}^{+
(n_{d-h}-1)} \cdots a_{1}^{+ n_{1}}\Phi  \  .
\end{align*}
In particular, if $k=d$,
$$
a_{d}^{-}a_{d}^{+ n_{d}}\cdots a_{1}^{+ n_{1}}\Phi =n_{d}a_{d}^{+
(n_{d}-1)}a_{d-1}^{+ (n_{d-1})} \cdots a_{1}^{+ n_{1}}\Phi\ ,
$$
so that, by induction
$$
a_{d}^{-n_{d}}a_{d}^{+ n_{d}}\cdots a_{1}^{+ n_{1}}\Phi
=n_{d}!a_{d-1}^{+ (n_{d-1})} \cdots a_{1}^{+ n_{1}}\Phi\ ,
$$
hence, again by induction,
\[
\|a_{d}^{+ n_{d}}\cdots a_{1}^{+ n_{1}}\Phi\|^2
=n_{d}!n_{d-1}!\cdots n_{1}! = \prod_{j=1}^{d}n_{j}!  \  .
\]
From this it follows that, for any finite set $F_{n}\subset\mathbb{N}^d$ and scalars
$x_{n_{1},\dots,n_{d}}\in\mathbb{C}$, $(n_{1},\dots,n_{d})\in F_{n}$,
\begin{align*}
&\left\|a_{k}^{-}\sum_{(n_{1},\dots,n_{d})\in
F_{n}}x_{n_{1},\dots,n_{d}} a_{d}^{+ n_{d}}\cdots a_{1}^{+
n_{1}}\Phi\right\|^2\\
=&\left\|\sum_{(n_{1},\dots,n_{d})\in
F_{n}}x_{n_{1},\dots,n_{d}}\sum_{h=0}^{d-1}
\delta_{k,d-h}n_{d-h}\cdots a_{d-h}^{+ (n_{d-h}-1)} \cdots
a_{1}^{+ n_{1}}\Phi \right\|^2\\
=&\left\|\sum_{h=0}^{d-1}\delta_{k,d-h}n_{d-h}
\sum_{(n_{1},\dots,n_{d})\in F_{n}}x_{n_{1},\dots,n_{d}}\cdots
a_{d-h}^{+ (n_{d-h}-1)} \cdots a_{1}^{+ n_{1}}\Phi \right\|^2\\
\le &\left(\sum_{h=0}^{d-1}n_{d-h} \left\|
\sum_{(n_{1},\dots,n_{d})\in F_{n}}x_{n_{1},\dots,n_{d}}\cdots
a_{d-h}^{+ (n_{d-h}-1)} \cdots a_{1}^{+ n_{1}}\Phi
\right\|\right)^2\\
=&\left(\sum_{h=0}^{d-1}n_{d-h}\left(
\left\|\sum_{(n_{1},\dots,n_{d})\in
F_{n}}x_{n_{1},\dots,n_{d}}\cdots a_{d-h}^{+ (n_{d-h}-1)} \cdots
a_{1}^{+ n_{1}}\Phi \right\|^2\right)^{1/2}\right)^2\\
=&\left(\sum_{h=0}^{d-1}n_{d-h}\left(\sum_{(n_{1},\dots,n_{d})\in
F_{n}}|x_{n_{1},\dots,n_{d}}|^2 \prod_{j=1}^{d-h-1}n_{j}! \
(n_{d-h}-1)! \ \prod_{j=d-h+1}^{d}n_{j}! \right)^{1/2}\right)^2\\
\le
&\left(\sum_{h=0}^{d-1}n_{d-h}\left(\sum_{(n_{1},\dots,n_{d})\in
F_{n}}|x_{n_{1},\dots,n_{d}}|^2
\prod_{j=1}^{d}n_{j}!\right)^{1/2}\right)^2\\
=&\left(\sum_{h=1}^{d}n_{h}\left(\sum_{(n_{1},\dots,n_{d})\in
F_{n}}|x_{n_{1},\dots,n_{d}}|^2
\prod_{j=1}^{d}n_{j}!\right)^{1/2}\right)^2\\
=&\left(n \left(\left\|\sum_{(n_{1},\dots,n_{d})\in
F_{n}}x_{n_{1},\dots,n_{d}} a_{d}^{+ n_{d}}\cdots a_{1}^{+
n_{1}}\Phi\right\|^2\right)^{1/2}\right)^2\\
=&n^2 \left\|\sum_{(n_{1},\dots,n_{d})\in
F_{n}}x_{n_{1},\dots,n_{d}} a_{d}^{+ n_{d}}\cdots a_{1}^{+
n_{1}}\Phi\right\|^2\ ,
\end{align*}
which is equivalent to the first inequality in \eqref{2eqr-01}.
The second  inequality in \eqref{2eqr-01} is proved similarly
\end{proof}

 Denoting for any $n\in\mathbb{N}$ and $\
A=\left(A_{j,k}\right)\in M_{d\times d}\left(\mathbb{C}\right)$,
\[
\left\vert A\right\vert :=d^{2}\max\left\{  \left\vert
A_{j,k}\right\vert \right\}\ ,
\]
one has
\begin{align}
\left\Vert \left.  B_{0}^{2}(A)\right\vert _{\mathcal{H}^{\widehat{\otimes}n}%
}\right\Vert =&\left\Vert \sum_{j,k}A_{j,k}\left.
a_{j}^{\dagger}a_{k} ^{\dagger}\right\vert
_{\mathcal{H}^{\widehat{\otimes}n}}\right\Vert \leq\left\vert
A\right\vert \sqrt{\left(  n+1\right)  \left(
n+2\right)  }\ ,\label{2eqr-01f}\\
\left\Vert \left.  B_{2}^{0}(A)\right\vert _{\mathcal{H}^{\widehat{\otimes}n}%
}\right\Vert =&\left\Vert \sum_{j,k}A_{j,k}\left.  a_{j}^{-}a_{k}
^{-}\right\vert _{\mathcal{H}^{\widehat{\otimes}n}}\right\Vert
\leq\left\vert A\right\vert \sqrt{n\left(  n-1\right)\ ,
}\label{2eqr-01g}\\
\left\Vert \left.  B_{A}^{0}\right\vert
_{\mathcal{H}^{\widehat{\otimes}n}
}\right\Vert =&\left\Vert \sum_{j,k}A_{j,k}\left.  a_{j}^{\dagger}a_{k}%
^{-}\right\vert _{\mathcal{H}^{\widehat{\otimes}n}}\right\Vert
\leq\left\vert A\right\vert n\label{2eqr-01h} \  .
\end{align}
In particular
\[
\left\Vert \left.
a_{j}^{\epsilon}a_{k}^{\epsilon^{\prime}}\right\vert
_{\mathcal{H}^{\widehat{\otimes}n}}\right\Vert \leq n+2,\
\forall\left\{ j,k\right\}  \subset\left\{  1,\ldots,d\right\}
,\text{ }n\in\mathbb{N}\text{ and }\epsilon,\epsilon^{\prime}=\pm\
,
\]
and, for any $n\in\mathbb{N}$, $A=\left(  A_{j,k}\right)  \in
M_{d\times d}\left(\mathbb{C}\right)$ and $\epsilon\in\left\{
0,\pm\right\}$,
\[
\left\Vert \left.  B_{A}^{\epsilon}\right\vert
_{\mathcal{H}^{\widehat {\otimes}n}}\right\Vert \leq\left\vert
A\right\vert \left(  n+2\right) \  .
\]
\begin{proposition}\label{2eqr-03}
For any $A\in M_{d\times d}\left(  \mathbb{C}\right)  $, for any
$n,m\in\mathbb{N}$ and $\xi\in\mathcal{H}^{\widehat{\otimes}n}$
\begin{align}
\left\Vert \left(  B_{0}^{2}(A)\right)^{m}\xi\right\Vert \leq
&\left\vert A\right\vert^{m}\sqrt{\left(  n+1\right)  \left(
n+2\right)  \ldots \left(n+2m-1\right)  \left(  n+2m\right)
}\left\Vert \xi\right\Vert \label{2eqr-03a}\  ,\\\\
\left\Vert \left(  B_{2}^{0}(A)\right)^{m}\xi\right\Vert \leq &
\chi_{\frac{n-1}{2}]}\left(  m\right) \left\vert
A\right\vert^{m}\notag\\
&\cdot\sqrt{n\left(n-1\right)  \ldots\left(  n-2\left( m-1\right)
\right)  \left(  n-2\left(m-1\right)  -1\right)
}\left\Vert \xi\right\Vert \label{2eqr-03b}\  ,\\\\
\left\Vert \left(  B_{A}^{0}\right)^{m}\xi\right\Vert \leq
&\left\vert A\right\vert^{m}n^{m}\left\Vert \xi\right\Vert
\label{2eqr-03c}  \  .
\end{align}
\end{proposition}
\begin{proof}  It is clear that for any
$\xi\in\mathcal{H}^{\widehat{\otimes}n}$, $\left(
B_{2}^{0}(A)\right)^{m}\xi$ differs from zero only if $2m-1<n$,
i.e. $m<\frac{n+1}{2}$, or equivalently,
$m\leq\frac{n+1}{2}-1=\frac{n-1}{2}$. We know from
(\ref{2eqr-01f}), (\ref{2eqr-01g}) and (\ref{2eqr-01h}) that the
thesis is true for $m=1$.   Noting that, for any $n,
m\in\mathbb{N}$, $\xi\in\mathcal{H}^{\widehat{\otimes}n}$, one has
 $\left(B_{A}^{0}\right)^{m}\xi\in\mathcal{H}^{\widehat{\otimes} n},\
\left(B_{0}^{2}(A)\right)^{m}\xi\in\mathcal{H}^{\widehat{\otimes}\left(
n+2m\right)  }$ and $\left(  B_{2}^{0}(A)\right)^{m}\xi\in\mathcal{H}%
^{\widehat{\otimes}\left(  n-2m\right)  }$, where
$\mathcal{H}^{\widehat {\otimes}\left(  n-2m\right) }:=\left\{
0\right\}  $ if $2m>n$, the thesis follows by induction.
\end{proof}

\begin{remark} \rm (\ref{2eqr-03a}), (\ref{2eqr-03b}) and (\ref{2eqr-03c})
guarantee that for any $A\in M_{d\times d}\left(\mathbb{C}\right)
$, $\epsilon\in\left\{  0,\pm\right\}$, $n,m\in\mathbb{N}$ and
$\xi\in\mathcal{H}^{\widehat{\otimes}n}$,
\begin{equation}
\left\Vert \left(  B_{A}^{\epsilon}\right)^{m}\xi\right\Vert
\leq\left\vert A\right\vert^{m}\sqrt{\left(  n+1\right)  \left(
n+2\right) \ldots\left(  n+2m-1\right)  \left(  n+2m\right)
}\left\Vert \xi\right\Vert \label{2eqr-03d} \  .
\end{equation}
\end{remark}
\begin{proposition}
\label{2eqr-04}For any $n,m\in\mathbb{N}$, $\left\{  A_{k}\right\}  _{k=1}%
^{m}\subset M_{d\times d}\left(  \mathbb{C}\right)  $ and
$\xi\in\mathcal{H}^{\widehat{\otimes}n}$, for any $\epsilon=\left(
\epsilon\left(  1\right) ,\ldots,\epsilon\left( m\right) \right)
\in\left\{ 0,\pm\right\}  $,
\begin{align*}
\left\Vert B_{A_{1}}^{\epsilon\left(  1\right)  }\ldots B_{A_{m}%
}^{\epsilon\left(  m\right)  }\xi\right\Vert \leq&\left(
\max_{1\leq k\leq m}\left\vert A_{k}\right\vert
\right)^{m}\\
&\cdot\sqrt{\left(  n+1\right)  \left( n+2\right) \ldots\left(
n+2m-1\right)  \left(  n+2m\right)  }\left\Vert \xi\right\Vert  \
.
\end{align*}
\end{proposition}
\begin{proof}  For $m=1$, (\ref{2eqr-03d}) gives the thesis.
Supposing by induction that the thesis is true for $m$ consider
the case $m+1$. By definition, $B_{A_{m+1}}^{\epsilon\left(
m+1\right)  }\xi$ belongs to
$\mathcal{H}^{\widehat{\otimes}n^{\prime}}$ with
\[
n^{\prime}=\left\{
\begin{array}
[c]{ll}%
n+2, & \text{if }\epsilon\left(  m+1\right)  =+\\
n, & \text{if }\epsilon\left(  m+1\right)  =0\\
n-2, & \text{if }\epsilon\left(  m+1\right)  =-
\end{array}
\right.  \  .
\]
So the assumption of induction gives
\begin{align}
&\left\Vert B_{A_{1}}^{\epsilon\left(  1\right)  }\ldots B_{A_{m}
}^{\epsilon\left(  m\right)
}B_{A_{m+1}}^{\epsilon\left(  m+1\right) }\xi\right\Vert  \label{2eqr-04c}\\
\leq&\left(  \max_{1\leq k\leq m}\left\vert A_{k}\right\vert
\right)^{m}\sqrt{\left(  n^{\prime}+1\right)  \left(
n^{\prime}+2\right) \ldots\left(  n^{\prime}+2m-1\right)  \left(
n^{\prime}+2m\right) }\left\Vert B_{A_{m+1}}^{\epsilon\left(
m+1\right)  }\xi\right\Vert \notag
\\
 \leq &\left(  \max_{1\leq k\leq m}\left\vert A_{k}\right\vert
\right) ^{m}\sqrt{\left(  n+3\right)  \left(  n+4\right)
\ldots\left(  n+2m+1\right) \left(  n+2m+2\right)  }\left\Vert
B_{A_{m+1}}^{\epsilon\left(  m+1\right) }\xi\right\Vert \notag \
.
\end{align}
Moreover, (\ref{2eqr-03d}) tells us that
\[
\left\Vert B_{A_{m+1}}^{\epsilon\left(  m+1\right) }\xi\right\Vert
\leq\left\vert A_{m+1}\right\vert \sqrt{\left( n+1\right)  \left(
n+2\right)  }\left\Vert \xi\right\Vert  \  .
\]
Therefore (\ref{2eqr-04c}) becomes
\begin{align*}
& \left\Vert B_{A_{1}}^{\epsilon\left(  1\right)  }\ldots B_{A_{m}%
}^{\epsilon\left(  m\right)  }B_{A_{m+1}}^{\epsilon\left(
m+1\right)
}\xi\right\Vert \\
& \leq\left\vert A_{m+1}\right\vert \left(  \max_{1\leq k\leq
m}\left\vert A_{k}\right\vert \right)^{m}\\
&\cdot \sqrt{\left( n+1\right) \left(  n+2\right) \left(
n+3\right) \left( n+4\right) \ldots\left(  n+2m+1\right)  \left(
n+2m+2\right)  }\left\Vert \xi\right\Vert \\
& \leq\left(  \max_{1\leq k\leq m+1}\left\vert A_{k}\right\vert
\right) ^{m+1}\\
&\cdot \sqrt{\left(  n+1\right)  \left(  n+2\right) \left(
n+3\right) \left(  n+4\right)  \ldots\left(  n+2\left( m+1\right)
-1\right) \left( n+2\left(  m+1\right)  \right) }\left\Vert
\xi\right\Vert \  .
\end{align*}
The thesis then follows by induction.
\end{proof}
\begin{theorem}\label{2eqr-06}
For any $n\in\mathbb{N}$, $\left\{  A,C,D\right\}\subset M_{d\times d}\left(  \mathbb{C}\right)$
and $\xi\in\mathcal{H}^{\widehat{\otimes}n}$, the series
\begin{equation}\label{2eqr-06a}
\sum_{m=1}^{\infty}\frac{\left\Vert \left(
uB_{C}^{\dagger}+vB_{2}^{0}(A)+wB_{D}^{0}\right)^{m}
\xi\right\Vert }{m!}z^{m}\ ,
\end{equation}
has positive convergence radius.
\end{theorem}
\begin{proof}  For any $n\in\mathbb{N}$, $\left\{  A,C,D\right\}
\subset M_{d\times d}\left(  \mathbb{C}\right)  $ and $\xi\in\mathcal{H}%
^{\widehat{\otimes}n}$, Proposition \ref{2eqr-04} gives
\begin{align*}
& \left\Vert \left(
B_{C}^{\dagger}+B_{2}^{0}(A)+B_{D}^{0}\right)^{m}\xi\right\Vert
\\
& \leq\sum_{\epsilon\in\left\{  0,\pm\right\}^{m}}\left\Vert B_{A_{1}%
}^{\epsilon\left(  1\right)  }\ldots B_{A_{m}}^{\epsilon\left(
m\right)  }B_{A_{m+1}}^{\epsilon\left(  m+1\right)  }\xi\right\Vert \\
& \leq3^{m}\left(  \max_{1\leq k\leq m}\left\vert A_{k}\right\vert \right)
^{m}\sqrt{\left(  n+1\right)  \left(  n+2\right)  \ldots\left(  n+2m-1\right)
\left(  n+2m\right)  }\left\Vert \xi\right\Vert \\
& \leq3^{m}\left(  \max\left\{  \left\vert A\right\vert
,\left\vert C\right\vert ,\left\vert D\right\vert \right\}
\right)^{m}\sqrt{\left( n+1\right)  \left(  n+2\right)
\ldots\left(  n+2m-1\right)  \left( n+2m\right)  }\left\Vert
\xi\right\Vert\ ,
\end{align*}
where $A_{k}\in\left\{  A,C,D\right\}  $ for any $k$. So
\begin{align*}
\frac{\left\Vert \left(
B_{C}^{\dagger}+B_{2}^{0}(A)+B_{D}^{0}\right)^{m} \xi\right\Vert
}{m!} \leq &  3^{m}\left(  \max\left\{  \left\vert A\right\vert
,\left\vert C\right\vert ,\left\vert D\right\vert \right\}  \right)^{m}\\
&\cdot \frac{\left(  n+2\right)  \left(  n+4\right)  \ldots\left(  n+2m\right)  }%
{m!}\left\Vert \xi\right\Vert \\
& \leq\left(9n\max\left\{  \left\vert A\right\vert ,\left\vert
C\right\vert ,\left\vert D\right\vert \right\}
\right)^{m}\left\Vert \xi\right\Vert \  .
\end{align*}
Thus the series \eqref{2eqr-06a} converges uniformly for $z\in\mathbb{C}$ such that
$$
|z|< \left(9n\max\left\{  \left\vert A\right\vert ,\left\vert
C\right\vert ,\left\vert D\right\vert \right\}  \right) \  .
$$
\end{proof}


\par\bigskip\noindent
\textbf{Acknowledgement}. The authors are grateful to the referee for his critical comments
that allowed us to improve several parts of the present paper.

\bibliographystyle{amsplain}

\end{document}